\documentclass[reqno,a4paper]{amsart}
\usepackage{fancyhdr}
\usepackage{epigraph}
\usepackage{cite}
\usepackage{pifont}

\newcommand{\dbar}{\,\mathchar'26\mkern-12mu d}
\DeclareRobustCommand{\SkipTocEntry}[5]{}
\usepackage{etoolbox}
\makeatletter
\let\ams@starttoc\@starttoc
\makeatother
\usepackage{parskip} 
\makeatletter
\let\@starttoc\ams@starttoc
\patchcmd{\@starttoc}{\makeatletter}{\makeatletter\parskip\z@}{}{}
\makeatother
\makeatletter
\def\@tocline#1#2#3#4#5#6#7{\relax
  \ifnum #1>\c@tocdepth 
  \else
    \par \addpenalty\@secpenalty\addvspace{#2}%
    \begingroup \hyphenpenalty\@M
    \@ifempty{#4}{%
      \@tempdima\csname r@tocindent\number#1\endcsname\relax
    }{%
      \@tempdima#4\relax
    }%
    \parindent\z@ \leftskip#3\relax \advance\leftskip\@tempdima\relax
    \rightskip\@pnumwidth plus4em \parfillskip-\@pnumwidth
    #5\leavevmode\hskip-\@tempdima
      \ifcase #1
      \or\or \hskip 1.7em \or \hskip 2em \else \hskip 3em \fi
      #6\nobreak\relax
    \dotfill\hbox to\@pnumwidth{\@tocpagenum{#7}}\par
    \nobreak
    \endgroup
  \fi}
\makeatother
\DeclareRobustCommand{\SkipTocEntry}[5]{}
\makeatletter
\def\@seccntformat#1{%
  \protect\textup{%
    \protect\@secnumfont
    \expandafter\protect\csname format#1\endcsname %
    \csname the#1\endcsname
    \protect\@secnumpunct
  }%
}

\makeatother
\makeatletter
\renewcommand\subsubsection{\@secnumfont}{\bfseries}%
\renewcommand\subsubsection{\@startsection{subsubsection}{3}
  \z@{.5\linespacing\@plus.7\linespacing}{-.5em}%
  {\normalfont\bfseries}}
\makeatother
\usepackage{times}
\usepackage{mathrsfs} 
\usepackage{latexsym,amsopn,amssymb}
\usepackage{amsfonts,amsbsy,amscd,stmaryrd}
\usepackage{mathtools} 
\usepackage{amsthm}
\usepackage{cases}
\usepackage{paralist}
\usepackage{enumitem}
\usepackage{leftidx}
\usepackage{tensor} 
\usepackage{bm}
\usepackage{mathrsfs}
\usepackage{esvect}
\usepackage[hypertexnames=true]{hyperref}
\hypersetup{colorlinks=true,linkcolor=RoyalBlue,%
urlcolor=RoyalBlue,citecolor=RoyalBlue}
\usepackage[usenames,dvipsnames,svgnames,table]{xcolor}

\newcommand\ssubsection{\subsubsection{\!\!\!}}
\usepackage{cleveref}
\usepackage[normalem]{ulem}
\def\ss#1{\scriptscriptstyle #1}

\def\ssw{{\scriptscriptstyle W}}
\def\oplussymp{\oplus^{\sigma}}
\def\tr#1{\hbox{\rm Tr}\,(#1)}
\def\Com{\mathrm{Com}}
\def\com{\mathrm{com}}
\def\norm#1{\left\Vert #1\right\Vert }

\newcommand{\C}{\mathbb{C}}

\newcommand{\mbg}{\mathbb{G}}
\newcommand{\G}{\mathbb{G}}
\newcommand{\mbh}{\mathbb{H}}

\newcommand{\N}{\mathbb{N}}
\newcommand{\mbp}{\mathbb{P}}

\newcommand{\R}{\mathbb{R}}

\newcommand{\Z}{\mathbb{Z}}
\newcommand{\mbs}{\mathbb{S}}
\newcommand{\mbfb}{\mathbf{B}}
\newcommand{\mbfk}{\mathbf{K}}

\newcommand{\ca}{\mathcal{A}}
\newcommand{\cb}{\mathcal{B}}
\newcommand{\cc}{\mathcal{C}}

\newcommand{\ce}{\mathcal{E}}
\newcommand{\cf}{\mathcal{F}}
\newcommand{\cg}{\mathcal{G}}
\newcommand{\ch}{\mathcal{H}}
\newcommand{\ci}{\mathcal{I}}
\newcommand{\cj}{\mathcal{J}}
\newcommand{\ck}{\mathcal{K}}

\newcommand{\co}{\mathcal{O}}
\newcommand{\cp}{\mathcal{P}}
\newcommand{\cR}{\mathcal{R}}
\newcommand{\cs}{\mathcal{S}}
\newcommand{\ct}{\mathcal{T}}

\def\rond{\mathscr}

\newcommand{\ra}{\rond{A}}
\newcommand{\rb}{\rond{B}}
\newcommand{\rc}{\rond{C}}
\newcommand{\rd}{\rond{D}}
\newcommand{\re}{\rond{E}}
\newcommand{\rf}{\rond{F}}
\newcommand{\rwf}{\prescript{\ss{W}\!\!\!}{}\rf}
\newcommand{\rwk}{\prescript{\ss{W}\!\!\!}{}\rk}
\newcommand\xk[1]{\prescript{{#1}\!\!\!}{}\rk}
\newcommand\xf[1]{\prescript{{#1}\!\!\!}{}\rf}

\newcommand{\ri}{\rond{I}}

\newcommand{\rk}{\rond{K}}
\newcommand{\rl}{\rond{L}}
\newcommand{\rM}{\rond{M}}

\newcommand{\rw}{\rond{W}}

\newcommand{\ie}{i.e.\ }
\newcommand{\eg}{e.g.\ }
\newcommand{\conf}{cf.\ }
\newcommand{\win}{\,\wtilde{\in}\,}

\newcommand{\half}{\textstyle{\frac{1}{2}}}
\usepackage{tikz}

\def\dom{\mathrm{dom}}

\def\rmb{\mathrm{b}}
\def\rmc{\mathrm{c}}

\def\dd{\mathrm{d}}

\def\d{\mathrm{d}}

\def\rme{\mathrm{e}}
\def\e{\mathrm{e}}

\def\rmi{\mathrm{i}}

\def\rms{\mathrm{s}}

\def\veps{\varepsilon}
\def\vphi{\varphi}

\def\bm#1{\boldsymbol{#1}}

\def\bra#1{\langle{#1}|}
\def\ket#1{|{#1}\rangle}
\def\braket#1#2{\langle{#1}|{#2}\rangle}
\def\bracet#1#2{\langle{#1},{#2}\rangle}

\def\jap#1{\langle {#1} \rangle}
\def\rarrow{\rightarrow}

\def\what{\widehat}
\def\what#1{\widehat{ #1\,}}
\def\wtilde{\widetilde}

\DeclareMathOperator*{\slim}{s-lim}
\DeclareMathOperator*{\wlim}{w-lim}

\def\nin{\notin}
\def\supp{\mbox{\rm supp\! }}

\def\cbu{C_{\mathrm b}^{\mathrm u}}

\def\spec{\mathrm{Sp}}
\def\spe{\mathrm{Sp_{ess}}}

\def\nin{\notin}

\def\ooplus{\textstyle{\bigoplus}}

\def\ccup{\textstyle{\bigcup}}
\def\ccap{\textstyle{\bigcap}}

\def\ssxi{{\ss{\Xi}}}
\def\ssw{{\scriptscriptstyle{W}}}
\def\ssx{{\scriptscriptstyle{X}}}
\def\ssy{{\scriptscriptstyle{Y}}}

\def\sse{{\ss{E}}}
\def\ssf{{\ss{F}}}

\def\ssh{{\ss{H}}}

\def\tsum#1#2{{\textstyle\sum_{#1}^{#2}}}

\newcommand{\hyphen}{\textnormal{-}}

\usepackage{perpage}
\MakePerPage{footnote}
\renewcommand\thesubsection{\thesection.\arabic{subsection}}
\newtheorem{theorem}{Theorem}[section]
\newtheorem{lemma}[theorem]{Lemma}
\newtheorem{proposition}[theorem]{Proposition}
\newtheorem{corollary}[theorem]{Corollary}

\newtheorem{definition}[theorem]{Definition}
\newtheorem{remark}[theorem]{Remark}

\newtheorem{example}[theorem]{Example}

\newtheorem*{theorem*}{Theorem}
\newtheorem*{proposition*}{Proposition}

\begin{document}

\vspace*{-10mm}

\title[Field Algebra]
{On the structure of the field $\bm{C^*}$-algebra \\
  of a symplectic space and spectral analysis\\
  of the operators affiliated to it\\[2mm]
  {\tiny\today}
}

\author[V. Georgescu]{Vladimir Georgescu} \address{V. Georgescu, CY
  Cergy Paris Universit\'e, Laboratoire AGM, UMR 8088 CNRS, F-95000
  Cergy, France} \email{vladimir.georgescu@math.cnrs.fr}

\author[A. Iftimovici]{Andrei Iftimovici}
\address{A. Iftimovici, CY Cergy Paris Universit\'e, Laboratoire AGM,
  UMR 8088 CNRS, F-95000 Cergy, France}
\email{andrei.iftimovici@cyu.fr}
     
     \thanks{{\em AMS Subject classification (2010):} 46L60
         (Primary), 47A10, 47L90, 81Q10, 81S05, 81S30 (Secondary).}
     \thanks{{\em Key words and phrases:} canonical
       commutation relations, quantum field theory, $C^*$-algebra,
       spectral theory, \phantom{Key}essential spectrum, HVZ
       theorem.} 

\begin{abstract} 

  We show that the $C^*$-algebra generated by the field operators
  associated to a symplectic space $\Xi$ is graded by the semilattice
  of all finite dimensional subspaces of $\Xi$.  If $\Xi$ is finite
  dimensional we give a simple intrinsic description of the components
  of the grading, we show that the self-adjoint operators affiliated
  to the algebra have a many channel structure similar to that of
  N-body Hamiltonians, in particular their essential spectrum is
  described by a kind of HVZ theorem, and we point out a large class
  of operators affiliated to the algebra.
  
\end{abstract}

\maketitle

{\tableofcontents}

\newpage

\section{Introduction}

The field operator algebra of a symplectic space $\Xi$ was introduced
by D.\ Kastler as ``the object of main interest for field theory'' in
\cite[\S6]{Kas}.  Kastler constructs a $C^*$-algebra $\rM$ involving
the twisted convolution product and the field $C^*$-algebra $\rl$ is
mentioned as a subalgebra of interest only at the end of his paper (he
devotes just over a page to this issue).

Independently of Kastler's work, D.\ Buchholz and H.\ Grundling
introduced the \emph{resolvent CCR algebra} in \cite{BuGr07,BuGr08} as
an alternative of the Weyl CCR algebra with more convenient properties
\cite{BuGr15}.  Loosely speaking, if we work in a representation of
$\Xi$ and denote $\phi(\xi)$ the field operator at the point
$\xi\in\Xi$, then the Weyl algebra is the $C^*$-algebra generated by
the exponentials $\e^{\rmi\phi(\xi)}$ while the resolvent algebra is
the $C^*$-algebra generated by the resolvents
$R_\xi=(\rmi-\phi(\xi))^{-1}$, where $\xi$ runs over $\Xi$.  On a more
abstract level, the resolvent CCR algebra is the universal
$C^*$-algebra generated by a family of operators
$\{R_\xi\}_{\xi\in\Xi}$ satisfying certain relations suggested by the
preceding definition of these operators in a representation, \conf
\cite[Defs.\ 3.1 and 3.4]{BuGr08}. In fact the resolvent and field
algebra are canonically isomorphic as follows from \cite[Th.\
23]{Kas}.  The resolvent algebra has been studied and used in the
treatment of quantum systems with a finite or infinite number of
degrees of freedom in a series of papers by D.\ Buchholz
\cite{Bu14,Bu17,Bu18,Bu20} and D.\ Bahns and D.\ Buchholz \cite{BB20}.

If $\Xi$ is finite dimensional a certain class of $C^*$-subalgebras of
$\rM$ has been pointed out in \cite{BG2}: for each \emph{finite}
stable under sums set $\cs$ of subspaces of $\Xi$ a $C^*$-subalgebra
$\rl(\cs)\subset\rM$ is defined which has a structure similar to the
$C^*$-algebra generated by the Hamiltonians of an $N$-body system.
The self-adjoint operators affiliated to such subalgebras generalize
the usual $N$-body Hamiltonians and have a similar many channel
structure and spectral theory.  In \cite{GI2} the theory is extended
to infinite sets $\cs$, in particular a $C^*$-algebra $\rl$, called
\emph{(graded) symplectic $C^*$-algebra of $\Xi$}, is associated to
the set of all subspaces of $\Xi$.  Only after reading the paper
\cite{BuGr08} by D.~Buchholz and H.~Grundling did we realise that
$\rl$ is just Kastler's field algebra, or resolvent CCR algebra.  The
extension to infinite dimensional $\Xi$ of our construction of the
symplectic $C^*$-algebra is easy.  Thus ``field algebra'', ``resolvent
CCR algebra'', and ``(graded) symplectic algebra'' should be
considered as synonyms; we mainly use the terminology ``field
algebra'' for historical reasons.

We make a preliminary comment on our notations.  Five $C^*$-algebras
(and their subalgebras) appear in our arguments: $\rM,\rl$, $\rk,\rf$,
$\re$. But in fact $\rM$ is canonically isomorphic to $\rk$ and $\rl$
is canonically isomorphic to $\rf$. More precisely, $\rM$ and $\rl$
are the ``abstract'' Kastler and field algebra (defined in
\eqref{eq:kastgen} and \eqref{eq:deff}) and involve the twisted
convolution product associated to symplectic space $\Xi$.  The
$C^*$-algebras $\rk,\rf$ are images of $\rM,\rl$ constructed via a
(continuous) representation $W$ of $\Xi$, see Definition
\ref{df:kas-field} (these algebras depend on $\Xi$ and $W$ but we do
not indicate it in notations unless really needed). The algebra $\re$,
called \emph{extended field $C^*$-algebra}, is defined only in
representations of $\Xi$, it contains $\rf$, and is needed in the
spectral theory of particles with nonzero spin, \eg N-body Dirac and
Pauli Hamiltonians (see \S\ref{ss:end} for a short presentation and
Sec.\ \ref{s:efa} for a detailed study).

Our results on the spectral analysis of the operators affiliated to
$\rf$ require $\rf\cap K(\ch)\neq0$, where $\ch$ is the Hilbert space
on which is represented $\rf$ and $K(\ch)$ is the algebra of compact
operators on $\ch$.  Indeed, a description of the quotient
$C^*$-algebra $\rf/K(\ch)$ allows one to compute the essential
spectrum of the observables affiliated to $\rf$ and to prove the
Mourre estimate with respect to convenient conjugate operators. But
$\rf\cap K(\ch)=0$ if $\dim\Xi=\infty$, hence the case of finite
dimensional symplectic spaces is important for us. 

A summary of our more significative results concerning the field
algebra $\rf$ may be found in the subsections \ref{ss:i1}--\ref{ss:i3}
of this introduction, while in subsections \ref{ss:i0} are
\ref{ss:gc0} we recall or introduce the main notions and notations. We
mention that we are forced to use a more general notion of observable
instead of self-adjoint operator because in general images through
morphisms and strong limits of self-adjoint operators affiliated to a
$C^*$-algebra are not self-adjoint operators in the usual sense (\eg
they could be not densely defined).

\subsection{}\label{ss:i0}
\addtocontents{toc}{\SkipTocEntry}

We first recall some functional analysis terminology and notations.
If $\ch$ is a Hilbert space then $B(\ch)$ is the algebra of bounded
operators on $\ch$, $K(\ch)$ the ideal of compact operators, $U(\ch)$
the group of unitary operators, and $1$ or $1_\ch$ the identity
operator.  $\spec(T)$ and $\spe(T)$ are the spectrum and the essential
spectrum of an operator $T$.

If $\{\ca_i\}_{i\in I}$ is a family of linear subspaces of a Banach
space $\ca$ then
\begin{equation}\label{eq:csum}
  \tsum{i\in I}{\rmc}\ca_i \doteq
  \text{ norm closure of the linear sum } \tsum{i\in I}{} \ca_i.
\end{equation}
If $\ca_1,\dots,\ca_n$ are subsets of a Banach algebra
$\ca$ then
\begin{align}
  & \ca_1\ca_2\dots\ca_n \doteq \text{ linear span of the products }
    A_1A_2\dots A_n \text{ with } A_i\in\ca_i
  \label{eq:linspan}\\
  & \ca_1\cdot\ca_2\cdot\ldots\cdot\ca_n \doteq \text{ closure of }
  \ca_1\ca_2\dots\ca_n . \label{eq:clinspan}
\end{align}
By ideal of a $C^*$-algebra we mean a closed and self-adjoint ideal.
By morphism between $*$-algebras we mean $*$-morphism.  We write
$\simeq$ for ``isomorphic'' and $\cong$ for ``canonically
isomorphic''.  If $\rc$ is a $C^*$-algebra and $\cp:\rc\to\rc$ a
morphism such that $\cp^2=\cp$ we say that $\cp$ is a \emph{projection
  morphism} or \emph{morphic projection}. Giving a projection morphism
is equivalent to giving a direct sum decomposition $\rc=\ra+\ri$ with
$\ra$ a $C^*$-subalgebra and $\ri$ an ideal; then $\cp$ is the linear
projection $\rc\to\ra$ determined by the direct sum. \label{p:pmor}

We say that a self-adjoint operator $A$ with spectrum $\sigma(A)$ on a
Hilbert space $\ch$ is \emph{affiliated}%
\footnote{\ This is not the notion introduced by S.~L.~Woronowicz.} to
(or with) a $C^*$-subalgebra $\rc\subset B(\ch)$ if the next
equivalent conditions are satisfied:
\begin{equation}\label{eq:afoch}
  \theta(A)\in\rc\ \forall \theta\in C_0(\R) \Leftrightarrow
    (A-z)^{-1}\in\rc \text{ for some } z\nin\sigma(A) .
\end{equation}
An \emph{observable} affiliated to an arbitrary $C^*$-algebra $\rc$ is
a morphism $A:C_0(\R)\to\rc$. This notion is discussed in detail in
\cite[\S8.1]{ABG}.  A self-adjoint operator is identified with the
observable defined by its $C_0\hyphen$functional calculus and we
usually denote $\theta(A)$ instead of $A(\theta)$ the value of $A$ at
$\theta$ for any observable $A$.  We will use the notation
\begin{equation}\label{eq:tilde}
  A\win\rc \Leftrightarrow  A\in\rc \text{ or }
  A \text{ is an observable affiliated to } \rc.
\end{equation}
The zero morphism is an observable affiliated to $\rc$ denoted
$\infty$; this is natural because $\theta(\infty)=0$ for any
$\theta\in C_0(\R)$. Observables can be described in terms of
$\rc$-valued \emph{self-adjoint resolvents}%
\footnote{\ The usual terminology is ``pseudo-resolvent'' but this
  seems excessive in our setting.}  the resolvent of $A$ being the map
$R_A\colon\C\setminus\R\to\rc$ with $R_A(z)=r_z(A)$ where
$r_z(\lambda)=(\lambda-z)^{-1}$.  If $\rc$ is realized on a Hilbert
space $\ch$ then observables may be identified with self-adjoint
operators acting in closed subspaces of $\ch$ \cite[\S8.1.2]{ABG}.
For example, the Hamiltonians of N-body systems with hard core
interactions are observables affiliated to the $C^*$-algebra generated
by the usual N-body Hamiltonians \cite{BGS} but are not self-adjoint
operators on $\ch$. And $\infty$ is the only operator with domain
$\{0\}$.

If $\rc,\rd$ are $C^*$-algebras, $\cp:\rc\to\rd$ a morphism, and $A$
an observable affiliated to $\rc$, then $\cp(A)\doteq\cp\circ A$ is an
observable affiliated to $\rd$.  If $\rc,\rd$ are realised on Hilbert
spaces $\ch,\ck$ and $A$ is the observable associated to a
self-adjoint operator on $\ch$ then quite often the observable
$\cp(A)$ is not associated to a (densely defined) self-adjoint
operator on $\ck$. \emph{Since objects like $\cp(A)$ appear naturally
  in our arguments, we are forced to work with the more general notion
  of observable instead of self-adjoint operator.}

An observable $A$ affiliated to $\rc$ is \emph{strictly affiliated} to
$\rc$ if the linear span of the operators $\theta(A)T$ with
$\theta\in C_0(\R)$ and $T\in\rc$ is dense in $\rc$.  Then $\cp(A)$
is a (densely defined) self-adjoint operator in any non-degenerate
representations $\cp$ of $\rc$ \cite[Pr.\ A.7]{DG}.

In the case of observables the strong limits are interpreted as
follows: if $\{A_n\}_{n\in\N}$ is a sequence of observables affiliated
to $B(\ch)$ and $\slim_nA_n(\theta)\doteq A(\theta)$ exists for any
$\theta\in C_0(\R)$ then $A$ is an observable affiliated to $B(\ch)$
written $A=\slim_nA_n$. Each $A_n$ could be a densely defined
self-adjoint operator without $A$ being so (we may have $A=\infty$).

The $C^*$-algebra $C^*(\ca)$ \emph{generated by a set $\ca$ of
  observables} affiliated to a $C^*$-algebra $\rc$ is the smallest
$C^*$-subalgebra of $\rc$ to which are affiliated all the $A\in\ca$.
If $\ca=\{A_i\mid i\in I\}$ we write $C^*(\ca)=C^*(A_i\mid i\in I)$.
Clearly $C^*(A)\equiv C^*(\{A\})=\{\theta(A)\mid\theta\in C_0(\R)\}$.
In terms of the resolvents $R_A$, for any $z_0$ not real $C^*(\ca)$ is
the $C^*$-subalgebra generated by the operators $R_A(z_0)$ with
$A\in\ca$.  If $A$ is an observable affiliated to $\rc$ then $C^*(A)$
is the closed linear span of the operators $R_A(\rmi\lambda)$ with
$\lambda\in\R\setminus\{0\}$.

If $A_1,\dots, A_n$ are observables with
$C^*(A_i)\cdot C^*(A_j)=C^*(A_j)\cdot C^*(A_i)$ and
$\ca_I\doteq C^*(A_{i_1})\cdot C^*(A_{i_2})\cdot\ldots\cdot
C^*(A_{i_m})$ for $I=\{i_1,\dots,i_m\}$ not empty subset of
$\{1,\dots, n\}$, then $\ca_I$ is a $C^*$-algebra and
$C^*(A_1,\dots,A_n)$ is the closure of $\sum_{I}\ca_I$.

$C^*$-algebras graded by semilattices have been introduced in
\cite{BG1} as a tool in the spectral theory of N-body Hamiltonians and
further studied in \cite{ABG,BG2,DG0,DG1,DG}. An exhaustive study of
these algebras may be found in Athina Mageira's thesis
\cite{M1,M2}. In this paper we consider only semilattices consisting
of finite dimensional subspaces of real vector spaces.

If $\Xi$ is an arbitrary real vector space the Grassmannian
$\mbg(\Xi)$ is the set of its finite dimensional subspaces, the
projective space $\mbp(\Xi)$ the set of its one dimensional subspaces,
and $\mbh(\Xi)$ the set of its hyperplanes.  $\mbg(\Xi)$ is equipped
with the order relation given by inclusion: $E\leq F$ means
$E\subset F$. Then $\mbg(\Xi)$ is a \emph{lattice}: for any pair of
its elements $E,F$ their upper bound $E\vee F=E+F$ and lower bound
$E\wedge F=E\cap F$ exist in $\mbg(\Xi)$.  We say that
$\cs\subset\mbg(\Xi)$ is a \emph{subsemilattice} if
$E,F\in\cs\Rightarrow E+F\in\cs$.  If $X\subset \Xi$ is a vector
subspace then $\mbg(X)$ and
$\mbg_{\ss{\supset X}}(\Xi)\doteq\{E\in\mbg(\Xi)\mid E\supset X\}$ are
subsemilattices and also a sublattices of $\mbg(\Xi)$.

Let $\cs\subset\mbg(\Xi)$ a subsemilattice. A $C^*$-algebra $\rc$ is
\emph{$\cs$-graded}, or \emph{graded by $\cs$}, if a family
$\{\rc(E)\}_{E\in\cs}$ of $C^*$-subalgebras of $\rc$ is given such
that the linear sum $\mathring\rc\doteq\sum_{E\in\cs} \rc(E)$ is
direct and dense in $\rc$ and
\begin{equation}\label{eq:grading1}
  \rc(E)\rc(F)\subset\rc(E+F) \quad\forall\ E,F\in\cs.
\end{equation}
The algebras $\rc(E)$ are called \emph{components} of $\rc$. We write
$\rc=\sum_{E\in\cs}^\rmc\rc(E)$ meaning that $\rc$ is the closure of
the linear sum $\mathring\rc$.  The sum $\sum_{E\in\cs} \rc(E)$ being
direct, for any $E\in\cs$ each $T\in \mathring\rc$ has a unique
component $T(E)\in \rc(E)$ such that $T=\sum_E T(E)$.
 If $\ct$ is a subsemilattice with $\cs\subset\ct$ then
an $\cs$-graded $C^*$-algebra is naturally $\ct$-graded. 

We sometimes use similarly defined inf-graded $C^*$-algebras. A subset
$\cs\subset\mbg(\Xi)$ is called \emph{stable under intersections}, or
\emph{$\cap$-stable}, if $E,F\in\cs\Rightarrow E\cap F\in\cs$. Then a
$C^*$-algebra $\rc$ is \emph{inf-graded by $\cs$} if a family
$\{\rc(E)\}_{E\in\cs}$ of $C^*$-subalgebras of $\rc$ is given such
that the sum $\mathring\rc\doteq\sum_{E\in\cs}\rc(E)$ is direct and
dense in $\rc$ and
$\rc(E)\rc(F)\subset\rc(E\cap F)\ \forall\ E,F\in\cs$.

\subsection{Grassmann C*-algebra}\label{ss:gc0}
\addtocontents{toc}{\SkipTocEntry}

The \emph{Grassmann $C^*$-algebra} $\cg_\ssx$ of a real finite
dimensional vector space $X$ is a simple example of inf-graded
$C^*$-algebra which may be seen as an abelian version of the field
algebra. Indeed, besides the inf-grading by $\mbg(\Xi)$, it is also
the $C^*$-algebra generated by the linear forms on $X$ which in this
context play the r\^ole of field operators. On the other hand,
Grassmann $C^*$-algebras also play an important r\^ole in the
description of the structure of the field algebra and in a generalized
N-body problem.

We denote $C_\rmb(X)$ the $C^*$-algebra of bounded continuous complex
functions on $X$ and $\cbu(X)$, $C_0(X)$ and $C_\rmc(X)$ the
subalgebras of uniformly continuous functions, functions which tend to
zero at infinity, and functions with compact support.

If $Y\in\mbg(X)$ then we embed $C_0(X/Y)\subset\cbu(X)$ via
$\varphi\mapsto\varphi\circ\pi_Y$ ($\pi_Y$ is the natural surjection
of $X$ onto the the quotient space $X/Y$). In other terms, $C_0(X/Y)$
is the set of continuous functions $\varphi:X\to\C$ with
$\varphi(x+y)=\varphi(x)$ for all $x\in X,y\in Y$ and such that
$\varphi(x)\to0$ when the distance from $x$ to $Y$ (for some norm on
$X$) tends to infinity.

The family of $C^*$-subalgebras $C_0(X/Y)$ with $Y\in\mbg(X)$ is
linearly independent and
\begin{equation}\label{eq:c0xy}
C_0(X/Y)\cdot C_0(X/Z)=C_0(X/(Y\cap Z)) \quad \forall Y,Z\in\mbg(X).
\end{equation}
Then $\mathring\cg_\ssx\doteq\tsum{Y\in\mbg(X)}{}C_0(X/Y)$ is a unital
$*$-subalgebra of $\cbu(X)$ and
\begin{equation}\label{eq:cgx1}
  \cg_\ssx\doteq \text{ closure of } \mathring\cg_\ssx
  =\tsum{Y\in\mbg(X)}{\rmc}C_0(X/Y)\subset\cbu(X)
\end{equation}
is a $C^*$-subalgebra of $\cbu(X)$, called \emph{(abelian) Grassmann
  $C^*$-algebra of $X$}. Clearly $\cg_\ssx$ is a $C^*$-algebra
inf-graded by the lattice $\mbg(X)$ with $C_0(X/Y)$ as components, in
particular the notation $\cg_\ssx(Y)=C_0(X/Y)$, which we sometimes
use, is natural.

There is an alternative description of $\cg_\ssx$, similar to
\eqref{eq:fieldBG1}, as the $C^*$-algebra generated by a simple and
natural set of observables affiliated to $\cbu(X)$. Note first that
any uniformly continuous function $u:X\to\R$ may be identified with an
observable affiliated to the $C^*$-algebra $\cbu(X)$, namely the
morphism $C_0(\R)\ni\theta\mapsto\theta\circ u\in\cbu(X)$; in the
natural representation of $\cbu(X)$ on $L^2(X)$ this observable is the
operator of multiplication by $u$. In particular, any linear form
$\varphi:X\to\R$ is an observable affiliated to $\cbu(X)$, hence the
dual vector space $X^*$ is a space of such observables.

If $Y\in\mbg(X)$ let
$Y^\perp=\{\vphi\in X^*\mid \varphi(x)=0\}\in\mbg(X^*)$.

\begin{proposition}\label{pr:phialgA}
  We have $\cg_\ssx = C^*(\varphi\mid\varphi\in X^*)$.  If
  $Y\in\mbg(X)$ and $\varphi_1,\dots,\varphi_n\in X^*$ then
  $Y=\ker\varphi_1\cap\dots\cap\ker\varphi_n$ if and only if
  $\varphi_1,\dots,\varphi_n$ generate $Y^\perp$ and then
  \begin{equation}\label{eq:genset}
    C_0(X/Y)=C^*(\varphi_1)\cdot C^*(\varphi_1) \cdot\ldots\cdot
    C^*(\varphi_n). 
  \end{equation}
\end{proposition}

The proof is an exercise in linear algebra.  The $C^*$-algebra of the
additive group $X^*$ is $C_0(X)$ and $C_0(X)\subset\cg_\ssx$, so the
group $C^*$-algebra is a component of $\cg_\ssx$. With a terminology \`a
la Buchholz-Grundling $\cg_\ssx$ would be the \emph{resolvent group
  $C^*$-algebra} of $X^*$ \cite{BuGr13}.

\begin{remark}\label{re:gralgint}{\rm

    The subalgebra $\cg_\ssx\subset\cbu(X)$ is stable under
    translations by elements of $X$ hence the crossed product
    $\cg_\ssx\rtimes X$ is a well defined $C^*$-algebra
    (see the comment after Remark \ref{re:czerot}). We will see in
    \S\ref{ss:nbb} that $\cg_\ssx\rtimes X$ is naturally embedded in
    the field algebra associated to the symplectic space $T^*X=X\oplus
    X^*$ and is the $C^*$-algebra generated by the maximal
    class of N-body type Hamiltonians on $X$
    (Theorem \ref{th:nbham} and Corollary \ref{co::nbham}).   

  }\end{remark}

\begin{example}\label{ex:2dim}{\rm

If $X=\R^2$ let $C(\theta)$ be the set of functions
$(\alpha,\beta)\mapsto u(\alpha\cos\theta+\beta\sin\theta)$ with $u\in
C_0(\R)$. Then $\{C(\theta)\}_{0\leq\theta<\pi}$ is a linearly
independent family of $C^*$-subalgebras of $\cbu(X)$ with
$C(\theta_1)\cdot C(\theta_2)=C_0(X)$ if $\theta_1\neq\theta_2$ and
$\cg_\ssx=\C+\sum_{0\leq\theta<\pi}^\rmc C(\theta)+C_0(X)$.   
    
}\end{example}

\subsection{}\label{ss:i1} 
\addtocontents{toc}{\SkipTocEntry}

A \emph{symplectic space} is a real vector space $\Xi$ equipped with a
symplectic form, \ie a bilinear anti-symmetric map $\sigma:\Xi^2\to\R$
which is non-degenerate in the following sense:
$\sigma(\xi,\eta)=0\ \forall \eta\in \Xi \Rightarrow \xi=0$.  If
$E\subset\Xi$ we set
$E^\sigma\doteq\{\xi\in\Xi \mid \sigma(\xi,\eta)=0 \ \forall \eta\in
E\}$; clearly $E\subset F \Rightarrow F^\sigma\subset E^\sigma$.  We
set $(E^\sigma)^\sigma=E^{\sigma\sigma}$.  If $E,F$ are finite
dimensional subspaces then $E^{\sigma\sigma}=E$ and
$E\subset F^\sigma\Leftrightarrow F\subset E^\sigma$.  We have
$\cap_iE_i^\sigma=(\sum_iE_i)^\sigma$ for any family of subspaces
$E_i$.  A subspace $E$ is \emph{isotropic} if $E\subset E^\sigma$\!,
\emph{coisotropic} if $E^\sigma\subset E$, \emph{Lagrangian} if
$E=E^\sigma$, \emph{symplectic} if $\sigma$ is non-degenerate on it.
$E\in\mbg(\Xi)$ is symplectic if and only if $E\cap E^\sigma=0$. If
$\mbg_\rms(\Xi)$ is the set of finite dimensional symplectic subspaces
\begin{equation}\label{eq:fil}
  E\in\mbg(\Xi) \Rightarrow \exists F\in\mbg_\rms(\Xi) \text{ such
    that }   E\subset F.
\end{equation}
We refer to \cite{DeG2} for the CCR theory associated to $\Xi$. A
\emph{representation} of $\Xi$ on a Hilbert space $\ch$ is a map
$W:\Xi\to U(\ch)$ such that
\begin{equation}\label{eq:ccr1I}
  W(\xi+\eta)= \e^{\frac{\rmi}{2}\sigma(\xi,\eta)} W(\xi)W(\eta) \ 
  \forall\xi,\eta\in\Xi \quad\text{and}\quad
  \wlim_{t\to 0}W(t\xi)=1 \ (t\in\R).
\end{equation}
The second condition is equivalent to the strong continuity of the
restriction of $W$ to finite dimensional subspaces hence
$\forall\xi\in\Xi$ we may define the field operator
$\phi(\xi)\equiv\phi_{\ss{W}}(\xi)$ as the self-adjoint operator such
that $W(t\xi)=\e^{\rmi t\phi(\xi)}$ $\forall t\in\R$. We set
$R_\xi(z)=(\phi(\xi)-z)^{-1}$.

If $E\in\mbg(\Xi)$ then $M(E)$ is the space of bounded Borel measures
on $E$ and $L^1(E)$ the subset of absolutely continuous measures (with
$L^1(0)=\C\delta_0$ if $E=0$).

\begin{definition}\label{df:kas-field}
  The \emph{Kastler $C^*$-algebra $\rk $ of $\Xi$ in the
    representation $W$} is the norm closure in $B(\ch)$ of the set of
  operators $W(\mu) = \int_E W(\xi) \mu(\d\xi)$ with
  $E\in\mbg(\Xi)$, $\mu\in M(E)$.
   The \emph{field $C^*$-algebra $\rf$ of $\Xi$ in the representation
    $W$} is the norm closure of the set of operators $W(\mu)$ with
  $E\in\mbg(\Xi)$ and $\mu\in L^1(E)$. 
\end{definition}

Note that the definition of $\rf$ given above is neither that of
Buchholz-Grundling nor that of Kastler; see Proposition
\ref{pr:fieldBG} and Remark \ref{re:kdefofrl} for the equivalence of
the definitions.

If $\xi\in\Xi$ then $W(\delta_\xi)=W(\xi)$ where
$\delta_\xi$ is the Dirac measure at $\xi$, so $W(\xi)\in\rk$; clearly
\begin{equation}\label{eq:1stcondW}
 \lim_{\xi\in E,\xi\to0}\|W(\xi)^*TW(\xi)-T\|=0 \quad \text{for all }
 T\in\rk \text{ and } E\in\mbg(\Xi).
\end{equation}
$\rf\subset\rk$ is a $C^*$-algebra containing the identity operator.
Since $\Xi$ is fixed we generally do not include it in notations but
if this is necessary we write $\rk(\Xi)$ and $\rf^\ssxi$.  The
algebras $\rk,\rf$ associated to different $W$ are canonically
isomorphic but if the specification of $W$ or $\Xi$ is necessary we
use the notations $\rwk$, $\rwf$ or $\rwk(\Xi)$, $\rwf^\ssxi$.  For
example, if $\Xi$ is finite dimensional then by the Stone-Von Neumann
theorem one has $\ch=\ch_0\otimes \ch_1$ and $W=W_0\otimes1_{\ch_1}$
for some irreducible representation $W_0$ of $\Xi$ on $\ch_0$ hence we
have $\rwk={\xk{\ssw_{\!0}}}{}\otimes1_\ck$ and
$\rwf={\xf{\ssw_{\!0}}}{}\otimes 1_{\ch_1}$.

Here is a description \`a la Buchholz and Grundling \cite{BuGr08} of
$\rf$ (see also \cite[Th.\ 23]{Kas}):
 
\begin{proposition}\label{pr:fieldBG}
  For any $z\in\C\setminus\R$ we have
\begin{align}
  \rf &=C^*(\phi(\xi)\mid \xi\in \Xi) =
  \text{$C^*$-algebra generated by the $\phi(\xi)$ with } \xi\in\Xi
  \label{eq:fieldBG1} \\
    &=\text{$C^*$-algebra generated by the operators } R_\xi(z)
    \text{ with } \xi\in\Xi.    \label{eq:fieldBG2}
\end{align}
\end{proposition}

\begin{definition}\label{df:rdefE}
  If $E\in\mbg(\Xi)$ and $\xi_1,\dots,\xi_n$ is a generating set for
  $E\in\mbg(\Xi)$ then
\begin{align}
  \rf(E) &\doteq \text{ norm closure of the set of operators }
  W(\mu) \text{ with } \mu\in L^1(E)  \label{eq:rfphi1I}\\
  &=C^*(\phi(\xi_1))\cdot C^*(\phi(\xi_2)) \cdot\ldots
  \cdot C^*(\phi(\xi_n))  \label{eq:rfphi2I} \\
  & = \text{ norm closed linear span of products }
    R_{\xi_1}(\rmi\lambda_1)\dots R_{\xi_n}(\rmi\lambda_n) 
    \nonumber\\   & \phantom{aaa}
     \text{ with } \lambda_1,\dots,\lambda_n\in\R\setminus\{0\}.
    \label{eq:rfphi3I}
\end{align}
$\rf(E)$ is a $C^*$-subalgebra of $\rf$ and $\rf(0)=\C$.
\end{definition}

Clearly
$ C^*(\phi(\xi))\cdot C^*(\phi(\eta))=C^*(\phi(\eta))\cdot
C^*(\phi(\xi))$ $\forall\xi,\eta\in\Xi$ so \eqref{eq:rfphi2I} is a
$C^*$-algebra.

We use the terminology of \S\ref{ss:i0} concerning graded
$C^*$-algebras, but the presentation here is rather
self-contained. See \S\ref{ss:fcaw} for the proofs of Theorems
\ref{th:Fgrad} and \ref{th:idec}.

\begin{theorem}\label{th:Fgrad}
  The set of $C^*$-subalgebras $\rf(E)$ of $\rf$ has the following
  properties:
  \begin{align}
  & E,F\in\mbg(\Xi) \Rightarrow \rf(E)\cdot\rf(F)=\rf(E+F),
  \label{eq:product}\\
  & \mathring\rf\doteq\tsum{E\in\mbg(\Xi)}{}\rf(E)
    \text{ is a linear direct sum and is dense in } \rf,
    \label{eq:ldir}\\
  & \text{if }\cs\subset\mbg(\Xi) \text{ is finite then }
    \tsum{E\in\cs}{}\rf(E) \text{ is norm closed}. \label{eq:finclo1}
  \end{align}
  In particular, the $C^*$-algebra $\rf$ is graded by the semilattice
  $\mbg(\Xi)$ with components $\rf(E)$.
\end{theorem}

\begin{remark}\label{re:emph}{\rm Note that $\rf(E)\cap\rf(F)=0$ if
    $E\subsetneq F$.  }\end{remark}

By \eqref{eq:product} and \eqref{eq:ldir} $\mathring\rf$ is a dense
unital $*$-subalgebra of $\rf$ and we may write 
\begin{equation}\label{eq:defff}
  \rf=\tsum{E\in\mbg(\Xi)}{\rmc}\rf(E).
\end{equation}
Since the sum in \eqref{eq:ldir} is direct, each $T\in \mathring\rf$
has, for each $E\in\mbg(\Xi)$, a unique component $T(E)\in \rf(E)$
such that $T=\sum_E T(E)$.

\begin{remark}\label{re:lestructI}
  {\rm We will see in Theorem \ref{th:structure} that the algebras
    $\rf(E)$ have a rather simple structure: they are isomorphic to
    tensor products $\ra\otimes\rb$ with $\ra$ an abelian
    $C^*$-algebra and $\rb$ isomorphic to the algebra of compact
    operators on a separable Hilbert space.
}
\end{remark}

For any $\cs\subset\mbg(\Xi)$ we set
\begin{equation}\label{eq:fcs}
\mathring\rf(\cs) \doteq \tsum{E\in\cs}{} \rf(E) \quad\text{and}\quad
  \rf(\cs) \doteq \tsum{E\in\cs}{\rmc} \rf(E).  
\end{equation}
By \eqref{eq:finclo1}, if $\cs$ is finite then
$\rf(\cs)=\mathring\rf(\cs)$.  If $\cs$ is a subsemilattice of
$\mbg(\Xi)$, meaning $E,F\in\cs\Rightarrow E+F\in\cs$, then $\rf(\cs)$
is naturally an $\cs$-graded $C^*$-algebra (Appendix \ref{s:graslat}).

In particular, any subspace $E\subset\Xi$ determines three
$C^*$-subalgebras of $\rf$:
\begin{align} \label{eq:deffE} \rf_\sse \doteq\tsum{F\subset E}{\rmc}
  \rf(F), \quad \rf'_\sse \doteq\tsum{F\not\subset E}{\rmc}\rf(F),
  \quad \rf_{\ss{\supset E}}\doteq \tsum{F\supset E}{\rmc} \rf(F).
\end{align}
Clearly $\rf_\sse$ is a unital $C^*$-subalgebra and $\rf'_\sse$ and
$\rf_{\ss{\supset E}}$ are ideals of $\rf$.  From \eqref{eq:rfphi2I}
we get a description of $\rf_\sse$ in terms of field operators:
\begin{equation}\label{eq:fephi}
   \rf_\sse=C^*(\phi(\xi)\mid \xi\in E).
\end{equation}

\begin{theorem}\label{th:idec}
  The $C^*$-algebra $\rf_\sse$ and the ideal $\rf'_\sse$ satisfy
  \begin{equation}\label{eq:PJI1}
    \rf = \rf_\sse + \rf'_\sse
    \quad\text{ and }\quad
    \rf_\sse \cap \rf'_\sse = 0.
  \end{equation}
  The projection $\cp_\sse \colon \rf\to\rf_\sse$ determined by this
  direct sum decomposition is a morphism and it is the unique
  continuous linear map $\cp_\sse\colon\rf\to\rf$ such that
  \begin{equation}\label{eq:mprojI1}
    T=\tsum{F}{}T(F)\in\mathring\rf \Rightarrow
    \cp_\sse T=\tsum{F\subset E}{}T(F). 
  \end{equation}
  For any subspaces $E,F$ we have 
  \begin{equation}\label{eq:ecapfI1}
    \rf_{\ss{E\cap F}}=\rf_\sse \cap\rf_\ssf
    \quad\text{and}\quad
    \cp_{\ss{E\cap F}}=\cp_\sse\cp_\ssf=\cp_\ssf\cp_\sse.
  \end{equation}
\end{theorem}

With the terminology introduced in \S\ref{ss:i0}, \emph{$\cp_\sse$ is
  a projection morphism of $\rf$ onto $\rf_\sse$}.

\subsection{}\label{ss:end}
\addtocontents{toc}{\SkipTocEntry}

For $\Xi$ finite dimensional and $W$ irreducible $\rf(E)$ has a simple
explicit description:

\begin{theorem}\label{th:mainthfd}
  If $\Xi$ is finite dimensional and $W$ is irreducible then $\rf(E)$
  is the set of $T\in B(\ch)$ such that:
  \begin{equation}\label{eq:maineq}
    \lim_{\xi\to0}\|[W(\xi),T]\| = 0,\ \ 
    \lim_{\xi\in E,\, \xi\to0}\|(W(\xi)-1)T\| = 0,\ \ 
    [W(\xi),T] = 0\, \forall \xi \in E^\sigma.
  \end{equation}
\end{theorem}

This is one of our main results (proved in \S\ref{ss:idc}).  The first
condition in \eqref{eq:maineq} can be written in the following form,
that makes sens even if $\dim\Xi=\infty$:
\begin{equation}\label{eq:maineq1}
\xi\mapsto W(\xi)TW(\xi)^* \text{ is norm continuous on finite
  dimensional subspaces of } \Xi.
\end{equation}
If $\Xi$ and $W$ are arbitrary then the conditions \eqref{eq:maineq},
the first one being written as in \eqref{eq:maineq1}, define an
interesting $C^*$-algebra $\re(E)$ of operators on $\ch$ which
contains $\rf(E)$, namely the $E$-component of the \emph{extended
  field algebra} $\re$ which is graded by $\mbg(\Xi)$.  The field
$C^*$-algebra $\rf$ contains only functions of the fields
$\phi(\xi)$. If the representation $W$ is not irreducible then many
other physically interesting observables are not affiliated to $\rf$,
\eg the Hamiltonians involving spin interactions. The $C^*$-algebra
$\re$ contains $\rf$ and fixes this flaw. We summarize below some
facts concerning $\re$ which will be studied in Section \ref{s:efa}.

\begin{theorem}\label{th:rcgradI}
  $\re(E)$ is a $C^*$-subalgebra of $B(\ch)$ such that 
  $\rf(E)\subset\re(E)$.  The family of $C^*$-algebras
  $\{\re(E)\}_{E\in\mbg(\Xi)}$ is linearly independent and 
  $\re(E)\cdot\re(F)\subset\re(E+F)$.
\end{theorem}

\begin{definition}\label{df:extrcI}
  The \emph{extended field $C^*$-algebra} associated to the
  representation $W$ is
  \begin{equation}\label{eq:rcgradI}
    \re \doteq \tsum{E\in\mbg(\Xi)}{\rmc} \re(E)
  \end{equation}
  and is a $\mbg(\Xi)$-graded $C^*$-algebra of bounded operators on
  $\ch$.
\end{definition}

We mention three more facts concerning the relation between $\rf$ and
$\re$:
\begin{compactenum}
\item $\re(0)=\Com(\Xi)$ where $\Com(\Xi)=\Com^\ssw(\Xi)$ is the
  commutant of the representation $W$, \ie the set of $T\in B(\ch)$
  such that $[T,W(\xi)]=0\,\forall\xi$.
\item
If $W$ is of finite multiplicity then $\re(E)=\rf(E)\cdot\Com(\Xi)$
and $\re=\rf\cdot\Com(\Xi)$.
\item
Assume $\Xi$ finite dimensional. Then
$\re(\Xi)=\rf(\Xi)\cdot\Com(\Xi)$ and $W$ is of finite multiplicity if
and only if $\re(\Xi)=K(\ch)$.
\end{compactenum}

\subsection{}\label{ss:comments}
\addtocontents{toc}{\SkipTocEntry}

The $\mbg(\Xi)$-graded structure of $\rf$ described in Theorem
\ref{th:Fgrad} is the main tool in the study of $\rf$ and its
applications.  Note that it allows one to construct operators in $\rf$
in a straightforward way: for each $E\in\mbg(\Xi)$ choose
$T(E)\in\rf(E)$ such that $\sum_{E}T(E)\doteq T$ is norm convergent;
then $T\in\rf$. Self-adjoint operators affiliated to $\rf$ may be
constructed by a similar procedure, see for example Theorem
\ref{th:afgrad} for a general abstract result, we give concrete
results later: we will see that in N-body type situations this
construction is surprisingly efficient. Note also that, by Theorem
\ref{th:mainthfd}, if $\Xi$ is finite dimensional and $W$ is
irreducible then the conditions ensuring the appartenance
$T(E)\in\rf(E)$ are quite explicit. 

Theorem \ref{th:idec}, which is a straightforward consequence of
Theorem \ref{th:Fgrad}, is important in the spectral analysis of the
operators in $\rf$ or observables affiliated to $\rf$, see Theorem
\ref{th:yeshvz} and Sections \ref{s:fdhvz} and \ref{s:efa} where we
prove extensions of the HVZ theorem describing the essential spectrum
of N-body Hamiltonians. We will see that if $\Xi$ is infinite
dimensional or is finite dimensional but $W$ is of infinite
multiplicity then $\rf\cap K(\ch)=0$, hence if $T\in\rf$ then its
essential spectrum coincides with its spectrum, so we have nothing to
say in these situations. But the relation $\rf\cap K(\ch)=0$ has much
worse consequences: it implies that many physically interesting
Hamiltonians are not affiliated to $\rf$ if $\dim\Xi=\infty$.  For
this reason in this paper we do not treat examples of infinite
dimensional symplectic spaces and of operators affiliated to the
corresponding field algebra. But let us make a comment on this case to
explain what is missing in $\rf$.

We consider \label{p:qft} the Fock space situation and describe a
result from \cite{G}.  Let $\Xi$ be a complex infinite dimensional
Hilbert space and $\ch=\Gamma(\Xi)$ the symmetric Fock space
associated to it. We keep the notation $\Xi$ for the underlying real
vector space of $\Xi$ equipped with the symplectic structure defined
by $\sigma(\xi,\eta)=\Im\braket{\xi}{\eta}$. Then $\rf$ is a
$C^*$-algebra of operators on $\ch$ which does not contain compact
operators and the usual quantum field Hamiltonians are not affiliated
to it. The problem comes from the fact that $\Gamma(A)\nin\rf$ if $A$
is a bounded operator on the one particle Hilbert space $\Xi$. A
solution is to extend $\rf$ by adding the necessary \emph{free kinetic
  energies}. More precisely, if $\co$ is an abelian $C^*$-algebra on
the Hilbert space $\Xi$ whose strong closure does not contain compact
operators then
\begin{equation}\label{eq:Ofield}
\Phi(\co)\doteq C^*(\phi(\xi)\Gamma(A) \mid \xi\in \Xi,
A\in\co,\|A\|<1) 
\end{equation}
is a $C^*$-algebra of operators on $\ch$ which contains the compacts
and whose quotient with respect to the ideal of compact operators is
canonically embedded in $\co\otimes\Phi(\co)$ which allows one to
describe the essential spectrum of the operators affiliated to
$\Phi(\co)$. The Hamiltonians of the $P(\varphi)_2$ models with a
spatial cutoff are affiliated to such algebras. The algebra
$\ra\doteq\Phi(\C)$ has a remarkable property: $K(\ch)\subset\ra$ and
$\ra/K(\ch)\cong\ra$.

\subsection{}\label{ss:find}
\addtocontents{toc}{\SkipTocEntry}

\emph{From now on until the end of the introduction we assume $\Xi$
  finite dimensional}.  Two finite dimensional symplectic spaces are
symplectically isomorphic hence their field algebras are isomorphic
and their realizations in irreducible representations are isomorphic.

We may describe $\rf_\sse$ and its commutant in $\rf$ independently of
the graded structure of $\rf$.

\begin{theorem}\label{th:comutant-i}
  For any subspace $E\subset\Xi$ we have \\[1mm]
  {\rm(1)} $\rf_\sse = \{T\in\rf \mid [T,W(\xi)]=0\ \forall \xi\in
  E^\sigma\}$, \\[1mm]
  {\rm (2)}
  $\rf_{\ss{E^\sigma}}=\{T\in\rf \mid [S,T]=0\ \forall S\in\rf_\sse\}$.
\end{theorem}

\begin{corollary}\label{co:comutant-i}
If $X\in\mbg(\Xi)$ is Lagrangian then $\rf_X$ is a maximal abelian
subalgebra of $\rf$, \ie if $T\in\rf$ then
$[S,T]=0$ $\forall S\in\rf_\ssx$ if and only if $T\in\rf_\ssx$. 
\end{corollary}

The \emph{spectrum} of an observable $T$ affiliated to a $C^*$-algebra
$\rc$ is the set $\spec(T)$ of real $\lambda$ such that
$\theta(T)\neq0$ if $\theta\in C_0(\R)$ and $\theta(\lambda)\neq0$.
Note that $\infty$ is the only observable with empty spectrum.  If
$\ch$ is a Hilbert space and $\rc\subset B(\ch)$, the \emph{essential
  spectrum} of $T$ is the set $\spe(T)$ of $\lambda$ such that
$\theta(T)\nin K(\ch)$ if $\theta(\lambda)\neq0$. We have
\begin{equation}\label{eq:inclusion}
  E\subset F \Rightarrow
  \spec(\cp_\sse T)\subset \spec(\cp_\ssf T) \quad\forall T\in\rf.
\end{equation}
Indeed, if $E\subset F$ then $\cp_\sse=\cp_\sse\cp_\ssf$ and
$\cp_\sse$ is a morphism.  The next theorem and its corollary follow
from Theorem \ref{th:idec} via Atkinson's theorem, \conf Theorems
\ref{th:gahvz}, \ref{th:gahvzcs}, \ref{th:RKR}.

\begin{theorem}\label{th:yeshvz}
  If $W$ of finite multiplicity then
\begin{equation}\label{eq:hvz0}
  \spe(T) = \ccup_{H\in\mbh(\Xi)}\spec(\cp_\ssh T) \quad
  \forall\ T\win\rf.
\end{equation}
Moreover, for any $H\in\mbh(\Xi)$ and any nonzero $\xi\in H^\sigma$ we
have
  \begin{equation}\label{eq:limpminf}
    \cp_\ssh T=\slim_{r\rarrow\infty} W(r\xi)^*T W(r\xi) .
\end{equation}
\end{theorem}

The definition of $\cp_\ssh T$ and of the limit in \eqref{eq:limpminf}
in the case of observables is given in \S\ref{ss:i0}.  The relation
\eqref{eq:hvz0} can be improved if $T$ belongs to a graded subalgebra
of $\rf$:

\begin{corollary}\label{co:yeshvz}
  Let $\cs\subset\mbg(\Xi)$ stable under sums with $\Xi\in\cs$ and
  $\cs_{\max}$ the set of maximal elements of $\cs\setminus\{\Xi\}$.
  Then
\begin{equation}\label{eq:hvz00}
  \spe(T) = \ccup_{E\in\cs_{\max}}\spec(\cp_\sse T) \quad
  \forall\ T\win\rf(\cs).
\end{equation}
\end{corollary}

Indeed, we have $\supset$ in \eqref{eq:hvz00} by
\eqref{eq:inclusion}. On the other hand, if $H$ is a hyperplane let
$E\in\cs$ of maximal dimension with $E\subset H$. Then $E$ is the
greatest element of $\cs$ included in $H$ hence
$\cp_\sse T=\cp_\ssh T$ so $\spec(\cp_E T)=\spec(\cp_HT)$ for all
$T\in\rf(\cs)$.

\begin{remark}\label{re:atk}{\rm The main points of the proof of
    Corollary \ref{co:yeshvz} are the relation
    $\rf(\Xi)=\rf(\cs)\cap K(\ch)$ and the fact that the Theorem
    \ref{th:idec} gives us an explicit description of the quotient
    $\rf(\cs)/\rf(\Xi)$ as a $C^*$-subalgebra of
    $\bigoplus_{E\in\cs_{\max}}\rf_\sse$ and of the canonical morphism
    $\cp:\rf(\cs)\to\rf(\cs)/\rf(\Xi)$. From this one gets
    \eqref{eq:hvz00} which is a general version of the HVZ theorem,
    and also the Mourre estimate with respect to certain conjugate
    operators by an extension of \cite[Th.\ 8.4.3]{ABG}.  Note that
    $\rf(\cs)/\rf(\Xi)$ is realised on the Hilbert space
    $\ch_\cs\doteq\bigoplus_{E\in\cs_{\max}}\!\ch$ and there are many
    self-adjoint operators $T$ affiliated to $\rf(\cs)$ such that
    $\cp(T)$ is not a densely defined operator on $\ch_\cs$.
    
}\end{remark}

\subsection{}\label{ss:i22}
\addtocontents{toc}{\SkipTocEntry}

We now turn to the question of self-adjoint operators are affiliated
with the field algebra $\rf$.  We consider here only the phase space
$\Xi=T^*\R^n$, for a general framework and proofs we refer to Section
\ref{s:LDsab}. We will see that $\rf$ is rather small and quite simple
and natural operators are not affiliated to it, but for example the
class of Hamiltonians with an N-body type structure affiliated to
$\rf$ is surprisingly large. Note that two finite dimensional
symplectic spaces of the same dimension are symplectically isomorphic
hence their field algebras are isomorphic, so there is no loss of
generality in the choice $\Xi=T^*\R^n$, but this is not natural for
some physical systems, \eg non-relativistic N-body systems with
translation invariant potentials require different Euclidean spaces
\cite{ABG,DeG1}.

Let us consider the Euclidean space $X=\R^n$ with scalar product
$\bracet{x}{y}=\sum_{i=1}^n x_iy_i$. The dual space $X^*$ of $X$ is
naturally identified with $X$ but it is convenient to keep the
notation $X^*$ for it. We think of $X$ as the configuration space of a
system whose phase space is $\Xi\equiv T^*X=X\oplus X^*$ with the
symplectic form
\begin{equation}\label{eq:standardLD}
\sigma(\xi,\eta)=\bracet{y}{k}-\bracet{x}{l}
\quad\text{if } \xi=(x,k), \eta=(y,l).
\end{equation}
The Schr\"odinger representation \cite[\S4.2.1]{DeG2} of $\Xi$ acts in
$\ch= L^2(X)$ by the rule
\begin{equation}
  (W(\xi)f)(y) = \e^{i\langle y+\frac{x}{2},k\rangle}f(y+x) \quad
  \text{for}\quad \xi=x+k\in\Xi.
\end{equation}
The Fourier transform is
$(F u)(k)=(2\pi)^{-n/2}\int_X \e^{-\rmi\bracet{x}{k}}u(x)\d x$.  The
position $q$ and momentum $p$ ($X$-valued) observables are defined as
follows: if $\varphi:X\to\C$ is a Borel function then $\varphi(q)$ is
the operator of multiplication by $\varphi$ in $\ch$ and
$\varphi(p)=F^*\varphi(q)F$.  Let $q_j$ be the operator of
multiplication by the j-th variable and $p_j=-\rmi\partial_j$ with
$\partial_j$ the derivative with respect to the j-th variable. Then
$\bracet{x}{p}=\sum_j x_jp_j$ and $\bracet{q}{k}=\sum_j k_jq_j$ if
$x,k\in X$ and the \emph{field operator at the point
  $\xi=(x,k)\in\Xi$} is $\phi(\xi)\doteq \bracet{q}{k}+\bracet{x}{p}$.
These operators are self-adjoint if the definitions are conveniently
interpreted and the field algebra $\rf$ is the $C^*$-algebra generated
by them. We have
\begin{equation}
W(\xi)  = \e^{\rmi\phi(\xi)}
= \e^{\frac{\rmi}{2}\bracet{x}{k}}\e^{\rmi\bracet{q}{k}}
  \e^{\rmi \bracet{x}{p}}
= \e^{-\frac{\rmi}{2}\bracet{x}{k}}\e^{\rmi\bracet{x}{p}}
  \e^{\rmi\bracet{q}{k}}.
\end{equation}
where $(\e^{\rmi\bracet{q}{k}}f)(y)=\e^{\rmi\bracet{y}{k}}f(y)$ and
$(\e^{\rmi\bracet{x}{p}}f)(y)=f(x+y)$.  By Definition
\ref{df:kas-field} \emph{$\ch$ is equipped with two remarkable
  $C^*$-algebras: the Kastler algebra $\rk$ and the field algebra
  $\rf$}.  Since $F^{-1}W(x,k)F=W(k,-x)$ we have
\begin{equation}\label{eq:foukl}
F^{-1}\rk F=\rk \quad\text{and}\quad F^{-1}\rf F=\rf
\end{equation}
For $T\in B(\ch)$ we have
$ W(\xi)^*TW(\xi)= \e^{-\rmi \bracet{x}{p}}\e^{-\rmi\bracet{q}{k}} T
\e^{\rmi\bracet{q}{k}}\e^{\rmi \bracet{x}{p}}$ hence the relation
\eqref{eq:1stcondW}, which is satisfied by any $T\in\rk$, is
equivalent to
\begin{equation}\label{eq:equicont0}
  \lim_{k\to0}\|\e^{-\rmi\bracet{q}{k}}T\e^{\rmi\bracet{q}{k}}-T\|=0
  \quad\text{and}\quad
  \lim_{x\to0}\|\e^{-\rmi\bracet{x}{p}}T\e^{\rmi\bracet{x}{p}}-T\|=0 .
\end{equation}
For example, if $\varphi\in L^\infty(X)$ then
\begin{equation}\label{eq:cbu}
  \varphi(q)\in\rk\Leftrightarrow \varphi(p)\in\rk
  \Rightarrow \varphi \in\cbu(X) .
\end{equation}
But the analogous result for the field algebra seems quite remarkable
(Proposition \ref{pr:fctrf}): 

\begin{theorem}\label{th:cbuxx}
  $\varphi(q)\in\rf \Leftrightarrow  \varphi(p)\in\rf
  \Leftrightarrow \varphi\in \cg_{\ssx}$.
\end{theorem}

$\cg_\ssx$ is defined in \eqref{eq:cgx1}. If $X=\R$ this means
$\varphi\in C_0(\R)+\C$, hence if $\varphi\in L^\infty(\R)$ then
$\vphi(q)\in\rf$ if and only if $\vphi$ is continuous and has finite
and equal limits at $\pm\infty$.

We now give examples of simple operators affiliated or not to the
field algebra.

\begin{example}\label{ex:xx2xy}{\rm If $X=\R^2$ the operators of
    multiplication by $x^2$, $x^2+y^2$, or $x-y$ are clearly
    affiliated to $\rf$.  But that of multiplication by $x^2-y^2$ is
    not affiliated to $\rk$ because if $\theta\in C_\rmc(\R)$ then
    $\theta(x^2-y^2)$ is not uniformly continuous.}
\end{example}

\begin{example}\label{ex:anisotropnaf}
  {\rm The limit in \eqref{eq:limpminf} is the same as $r\to+\infty$
    or $r\to-\infty$, so an anisotropic behaviour is not allowed by
    $\rf$.  It follows that certain physically relevant Hamiltonians
    are not affiliated to $\rf$, for example, the one dimensional
    anisotropic Schr\"odinger operators. More precisely, let $X=\R$
    and $H=p^2+v(q)$ with $v$ real locally integrable such that 
    \begin{equation}\label{eq:lim+-}
      \lim_{x\to\pm\infty}\int_{|y-x|<1}|v(y)-a_\pm|\d x=0 \text{ for some
        numbers } a_\pm.
    \end{equation}
    Then $(H+\rmi)^{-1}\in\rf$ if and only if $a_+=a_-$ because}
  \[
    \slim_{r\to\pm\infty} \e^{\rmi rp}(H+\rmi)^{-1}\e^{-\rmi
      rp}=(p^2+a_\pm+\rmi)^{-1} .
  \]
  {\rm Similarly, the anisotropic algebras
    $C(\overline\R)\cdot C_0(\R^*)$ and
    $C(\overline\R)\cdot C(\overline{\R^*})$ from \cite{GI5}, although
    natural algebras of quantum Hamiltonians, are not included in
    $\rf$. The N-body Hamiltonians studied in \cite{GN} are not
    affiliated to the field algebra associated to the natural $X$.}
\end{example}

\begin{example}\label{ex:riemann}
  {\rm This is similar to Example \ref{ex:anisotropnaf} but we perturb
    $p^2$ by a second order operator. Let $v\geq1$ a bounded Borel
    function on $\R$ satisfying \eqref{eq:lim+-}. Then  $H=pv(q)p$,
    with the Sobolev space $\ch^1(\R)$ as form domain, is affiliated
    to $\rf$ if and only if $a_+=a_-$. 
  }\end{example}

\begin{example}\label{ex:stark}
  {\rm For $X=\R$ the Stark Hamiltonian
    $H=p^2+q=\e^{\rmi p^3/3}q \e^{-\rmi p^3/3}$ is not affiliated to
    $\rf$: indeed, if $T=(H+\rmi)^{-1}$ then the second condition in
    \eqref{eq:equicont0} is satisfied but not the first one, \conf
    \cite[\S 4.2.7]{G1}. Also $p^2-q^2$ is not affiliated to $\rf$
    by \cite[Prop.\ 6.3]{BuGr08}.}
\end{example}

\begin{example}\label{ex:magn}{\rm

    If $X=\R^3$ and $b$ is a constant magnetic field with  magnetic
    vector potential $a(x)=\frac{1}{2}b\times x$ then the magnetic
    Hamiltonian $H_0=(p-a)^2$ is affiliated to $\rf$. This is an easy
    consequence of Theorem \ref{th:mainthfd} because $H_0=\sum_{j=1}^3
    \phi(\eta_j)^2$ with $\eta_j\in T^*X$ (Remark
    \ref{re:cross}). Theorem \ref{th:affg} allows one to treat
    $H=(p-a)^2+v$ with rather general $a$ and $v$. 
 
}\end{example}

\begin{example}\label{ex:dilat}
  {\rm An example of a different nature is the generator of the
    dilation group: $\omega=pq+qp=2pq+\rmi$ is not affiliated to $\rf$
    (again $X=\R$). Indeed, in Proposition \ref{pr:dilnaf} we show
    that $\wlim_{|r|\to\infty}W(r\xi)^*(\omega+\rmi)^{-1}W(r\xi)=0$
    for all $\xi\in\Xi$, $\xi\neq0$ and $(\omega+\rmi)^{-1}$ is not
    compact, hence $(\omega+\rmi)^{-1}\nin\rf$ by Theorem
    \ref{th:yeshvz}. Note that $p^2-q^2=p'q'+q'p'$ where
    $p'=(p-q)/\sqrt2, q'=(p+q)/\sqrt2$ satisfy $[p',q']=-\rmi$ }
\end{example}

\subsection{}\label{ss:i3}
\addtocontents{toc}{\SkipTocEntry}

We now consider N-body Hamiltonians in the framework of
\S\ref{ss:i22}.  Our purpose is to show that the class of such
Hamiltonians affiliated to $\rf$ is very large.

We say that a function $h:X\to\R$ is \emph{divergent} if
$\lim_{k\to\infty}h(k)=+\infty$. Below we fix the kinetic energy
function $h$ but we allow kinetic energy operators of the form
$h(p-k)$ with arbitrary $k\in X$ because the origin of the momentum
space should not play a r\^ole.

\begin{proposition}\label{pr:nbham1}
  Let $h:X\to\R$ be continuous and divergent and
  $\ct_0\subset\mbg(X)$. Denote $\ct$ the set of intersections of
  subspaces from $\ct_0$ and $\ct^\sigma=\{Y^\sigma\mid Y\in\ct\}$,
  which is a subsemilattice of $\mbg(\Xi)$.  Then $\rf(\ct^\sigma)$ is
  the $C^*$-algebra generated by the self-adjoint operators
  $H=h(p-k)+v(q)$ with $k\in X$ and real
  $v\in\sum_{Y\in\ct_0}C_0(X/Y)$.
\end{proposition}

This is a consequence of Theorem \ref{th:nbham}. The proposition says
that the algebras $\rf(\cs)$ with $\cs$ a subsemilattice of subspaces
of $\Xi$ such that $\min\cs=X$ are generated by Hamiltonians having an
N-body type structure. A computation page \pageref{p:stnb}  clarifies
this assertion. 

Let $\cs\subset\mbg(\Xi)$ be an arbitrary subsemilattice with
$\min\cs=X$ and $h:X\to\R$ continuous and divergent. Then $h(p)$ is a
kinetic energy operator affiliated to $\rf(X)=C_0(X^*)$ and our goal
is to build self-adjoint operators $H=h(p)+V$ affiliated to
$\rf(\cs)$. The case of bounded operators $V$ is easy but not without
interest if $\cs$ is infinite because we do not assume that $V$ is a
sum of components. By Proposition \ref{pr:afgrad} and Corollary
\ref{co:yeshvz} we have:

\begin{proposition}\label{pr:vbdd}
  If $V\in\rf(\cs)$ is symmetric then $H= h(p)+V$ is affiliated to
  $\rf(\cs)$ and $\cp_\sse H=h(p)+\cp_\sse V$ for all $E\in\cs$.  If
  $\Xi\in\cs$ then
 \begin{equation}\label{eq:hvzNbI}
   \spe(H) = \ccup_{E\in\cs_{\max}}\spec(\cp_\sse H) .
 \end{equation}
\end{proposition}

\begin{remark}\label{re:infcs}{\rm

The maximal choice $\cs=\mbg_{\ss{\supset X}}$ is particularly
interesting, well beyond the N-body problem. We have $\cs=\ct^\sigma$
with $\ct\subset\mbg(X)$ stable under intersections and the N-body
case corresponds to finite $\ct$: then $H=h(p)+\sum_{Y\in\ct} v_\ssy$
with $v_\ssy\in C_0(X/Y)$ in the simplest case, \conf \eqref{eq:NB01}.
But if $\ct=\mbg(X)$ then $H=h(p)+V$ and $V$ is not just a sum of
$v_\ssy$. In fact,  extending from finite $\ct$ to $\mbg(X)$ is like
going from trigonometric polynomials to uniformly (Bohr) almost
periodic functions. Instead of periodic functions, \ie functions
invariant under translations by elements of subgroups $a\Z$, we use
functions invariant under translations by elements of vector subgroups
(subspaces) of $X$. Thus the self-adjoint operators affiliated to
$\rf_{\ss{\supset X}}$ could be called ``almost N-body Hamiltonians'',
meaning that they can be approximated, in some sense, with elementary
N-body type Hamiltonians.

}\end{remark}

In order to treat unbounded $V$ we consider here only functions $h$
such that the form domain of $h(p)$ is a Sobolev space (Theorem
\ref{th:afgradII} is a more general but less explicit statement).
$\ch^s\equiv\ch^s(X)$ is the Sobolev space of order $s\in\R$ defined
by the norm $\|u\|_s=\|\jap{p}^su\|$ where $\jap{x}=(1+|x|^2)^{1/2}$.
We have $\ch^0 =\ch$ and if $s \geq 0$ then
$\ch^s\subset \ch \subset \ch^{-s}= (\ch^s)^*$.  If $s>0$ is we write
$h(x)\sim |x|^{2s}$ if $c'|k|^{2s} \leq h(k) \leq c''|k|^{2s}$ for
some constants $c', c''$ and all large $k$. If $h$ is bounded from
below this is equivalent to $D(|h(p)|^{1/2}=\ch^s$. Form sums are
defined in Appendix \S\ref{ss:perturb}; for the proof of the next
theorem see Theorem \ref{th:opafrf}.

\begin{theorem}\label{th:opafrfI}
  Let $\cs\subset\mbg(\Xi)$ a subsemilattice with $\min\cs=X$ and
  $h:X\to\R$ continuous with $h(x)\sim |x|^{2s}$ for some $s>0$.  Let
  $V:\ch^s\to\ch^{-s}$ symmetric such that $V\geq -\mu h(p)-\nu$ with
  $\mu<1,\nu\geq0$ and
  $\varphi(p)V\jap{p}^{-s} \in \rf(\cs)\ \forall\varphi\in
  C_\rmc(\R)$.  Then the form sum $H=h(p)+V$ is a self-adjoint
  operator strictly affiliated to $\rf(\cs)$.
\end{theorem}

The perturbations $V$ allowed by Theorem \ref{th:opafrfI} are much
more general than those usually considered in the standard N-body
problem even when $\cs$ (hence $\ct$) is finite and even in the
two-body problem (\ie $\ct=\{0,X\}$). We explain this below.

If $s \geq 0$ we denote $\mbfb^s(X) \doteq B(\ch^s,\ch^{-s})$ and say
that \emph{$T\in \mbfb^s(X)$ is small at infinity} if it verifies the
following equivalent conditions:
\[
  T:\ch^{s}\to\ch^{-t} \text{ is compact if } t>s \Leftrightarrow
  \varphi(p)T:\ch^{s}\to\ch \text{ is compact if }
  \varphi\in C_\rmc(X).
\]
Let $\mbfb_0^s(X)$ be the closed subset of $\mbfb^s(X)$ consisting of
small at infinity operators. For example, $\varphi(q)\in\mbfb^0(X)$ if
and only if $\varphi\in L^\infty(X)$ and $\varphi(q)\in\mbfb_0^0(X)$
means that we also have
$\lim_{a\to\infty}\int_{|x-a|}|\varphi(x)|\d x=0$ (which explains the
conditions \eqref{eq:lim+-}).

If $Y\in\mbg(X)$ and $Y'$ is a complementary subspace of $Y$ in $X$
then $X=Y\oplus Y'$ and
\[
L^2(X)\cong L^2(Y)\otimes L^2(Y')=L^2(Y;L^2(Y')).
\]
Let $F_\ssy$ be the Fourier transform in $L^2(Y)$ and keep the
notation $F_\ssy$ for $F_\ssy\otimes1$ acting in
$L^2(X)=L^2(Y)\otimes L^2(Y')$.  Let $q_\ssy,p_\ssy$ be the position
and momentum observables associated to $Y$, so
$\varphi(p_\ssy)=F_\ssy^{-1}\varphi(q_\ssy)F_\ssy$ if $\varphi:Y\to\C$
and $C_0(Y^*)=F_\ssy^{-1}C_0(Y)F_\ssy$ (\conf \S\ref{ss:i22} and
\S\ref{ss:HLdec}).  If $\Phi:Y\to\mbfb(Y')$ is continuous and we
denote $\Phi(q_\ssy)$ the operator of multiplication by $\Phi$ in
$L^2(X)=L^2(Y;L^2(Y'))$, then
$\Phi(p_\ssy)\doteq F_\ssy^{-1} \Phi(q_\ssy) F_\ssy$.

This definition extends to functions on $Y$ with values operators on
Sobolev spaces on $Y'$.  For example, if $V_\ssy:Y\to\mbfb^s(Y')$ is
continuos and $V_\ssy(y):\ch^s(Y')\to\ch^{-s}(Y')$ is symmetric and
satisfies $\pm V_\ssy(y)\leq c_\ssy\jap{|y|+|p_{Y'}|}^{2s}$ for some
constant $c_\ssy$ and all $y$, then clearly
$\pm V_\ssy(p_\ssy)\leq C_\ssy\jap{p}^{2s}$ for some number $C_\ssy$,
hence $V_\ssy(p_\ssy)\in\mbfb^s(X)$.

Let $\ct\subset\mbg(X)$ a \emph{finite} $\cap$-stable subset with
$0,X\in\ct$. Denote $\cs=\ct^\sigma$ and $\cs_{\max}$ the set of
maximal elements of $\cs\setminus\{\Xi\}$.

\begin{theorem}\label{th:afgradV}
  For any $Y\in\ct$ let $Y'$ a complementary subspace of $Y$ in $X$
  and $V_\ssy$ a norm continuos function $Y\to\mbfb_0^s(Y')$ such that
  $V_\ssy(y):\ch^s(Y')\to\ch^{-s}(Y')$ is symmetric and satisfies
  $\pm V_\ssy(y)\leq c_\ssy\jap{|y|+|p_{Y'}|}^{2s}$ for a constant
  $c_\ssy$ and all $y\in Y$. Let $V_0=0$.  For $E=Y^\sigma\in\cs$
  assume that the symmetric operator
  $V(E)=V_\ssy(p_\ssy)\in \mbfb^s(X)$ satisfies:
  \[
    V(E)\geq-\mu_\sse h(p)-\nu_\sse \text{ with }
    \mu_\sse,\nu_\sse\geq0 \text{ and } \tsum{E\in\cs}{}\mu_\sse<1.
  \]
  Then the form sums $H=h(p)+\sum_EV(E)$ and
  $H_\sse=h(p)+\sum_{F\leq E}V(F)$ are bounded from below self-adjoint
  operators, with form domain $\ch^s$, strictly affiliated to
  $\rf(\cs)$, and $\cp_\sse H=H_\sse\ \forall E\in\cs$.  Moreover, we
  have 
    \begin{equation}\label{eq:hvzNbII}
      \spe(H) = \ccup_{E\in\cs_{\max}}\spec(\cp_\sse H) .
    \end{equation}
\end{theorem}

We mention that the relations \eqref{eq:hvzNbI} and \eqref{eq:hvzNbII}
are HVZ type theorems for $H$.

\begin{remark}\label{re:formal}{\rm

Let us make a formal comment on the structure of the ``potentials''
$V_\ssy(p_\ssy)$ to clarify the gain in generality compared
to potentials $V_\ssy$ in the standard N-body problem. The usual
$V_\ssy$ is a function of the internal variable $q_{\ss{Y'}}$ (of
class $L^u$ for the natural number $u$) while $V_\ssy(p_\ssy)$ is an
operator depending on $q_{\ss{Y'}}$  and on the total momentum
$p=(p_\ssy, p_{\ss{Y'}})$.  
    
}\end{remark}

\begin{example}\label{ex:simple}{\rm

    The simplest case is $\ct=\{0,X\}$, hence $\cs=\{X,\Xi\}$ and
    $\cf(\cs)=C_0(X^*)+\mbfk(X)$. Then $V(X)=0$ hence
    $V=V(\Xi)$. Since $0^\sigma=\Xi$ we have $Y=0$ hence $Y'=X$ so
    $V:\ch^s\to\ch^{-s}$ must be a symmetric small at infinity
    operator such that $V\geq-\mu h(p)-\nu$ with $\mu<1$ (the operator
    $V$ may have the same order $2s$ as $h(p)$).  Then $H=h(p)+V$ is
    strictly affiliated to $\cf(\cs)$ and
    $\spe(H)=h(X)=[\min h,\infty[$.
    
   }\end{example}

\section{Field algebra}\label{s:afg}
\protect\setcounter{equation}{0}

\subsection{Kastler algebra} \label{ss:akas}

We recall here some basic facts about Kastler's algebra \cite{Kas}.

\subsubsection{Measure spaces}\label{sss:noterm}

If $X$ is a finite dimensional real vector space we denote $M(X)$ the
space of bounded Borel measures on $X$ identified with the Banach
space dual of $C_0(X)$ and $L^1(X)$ the subspace of absolutely
continuous measures. If $X=0\equiv\{0\}$ then
$M(X)=L^1(X)=\C\delta_0\equiv\C$ where $\delta_0$ is the Dirac measure
at zero. If $Y\subset X$ is a subspace then $M(Y)$ is identified with
the closed subspace of $M(X)$ consisting of the measures with support
in $Y$, so $M(Y)\subset M(X)$ isometrically, and $L^1(Y)\subset M(X)$
consists of measures with support in $Y$ and whose restriction to $Y$
is absolutely continuous. Since the elements of $L^1(X)$ are thought
as measures not functions their norms do not depend of the choice of a
Lebesgue measure on $X$.  $\mbs(X)$ is the space of Schwartz test
functions with dual the space $\mbs'(X)$ of tempered distributions,
hence $L^1(X)\subset M(X)\subset\mbs'(X)$ continuously.  We set
$\bracet{u}{v}=v(u)$ if $u\in\mbs(X),v\in\mbs'(X)$.

\subsubsection{Definition}\label{sss:kasdef}

If $\Xi$ is a symplectic space then $\{M(E)\}_{E\in\mbg(\Xi)}$ is a
filtered family of Banach spaces with $M(E)\subset M(F)$ isometrically
if $E\subset F$ so $M(\Xi)=\cup_{E\in\mbg(\Xi)}M(E)$ is a normed
vector space; we denote $\|\cdot\|_1$ its norm.  If $E,F\in\mbg(\Xi)$
and $E\neq F$ then
\begin{equation}\label{eq:FG}
  M(E)\cap M(F)=M(E\cap F) \quad\text{and}\quad
  L^1(E)\cap L^1(F)=0 .
\end{equation}
According to Kastler, the \emph{twisted convolution algebra} of $\Xi$
is the normed space $M(\Xi)$ equipped with the following unital
$*$-algebra structure: if $\mu,\nu\in M(\Xi)$ then
$\exists\,E\in\mbg(\Xi)$ such that $\mu,\nu\in M(E)$ and then the
twisted convolution $\mu\oast\nu\in M(E)$ is defined by
\begin{equation}\label{eq:tconv}
  \int f(\xi)(\mu \oast \nu)(\d\xi)=
  \iint \e^{-\frac{\rmi}{2}\sigma(\xi,\eta)}
f(\xi+\eta)\mu(\d\xi)\nu(\d\eta) \quad \forall f\in C_0(E),
\end{equation}
the unit is the Dirac measure at zero $\delta_0$, and the adjoint
measure $\mu^*$ is defined by
\begin{equation}\label{eq:tadjI}
\int f(\xi)\mu^*(\d\xi)=\int f(-\xi)\bar{\mu}(\d\xi)
\quad \forall f\in C_0(E).
\end{equation}
If $\Xi$ is finite dimensional then Kastler \cite[Th.\ 15]{Kas} proves
that \emph{the $*$-algebra $M(\Xi)$ has, modulo equivalence, only one
  faithful irreducible representation and any faithful representation
  of $M(\Xi)$ is a multiple of this one; in particular, $M(\Xi)$ has a
  unique $C^*$-norm $\|\cdot\|$ and}
$\|\mu\|\leq\|\mu\|_1=\int_\Xi|\mu|$. The $C^*$-algebra completion
$\rM(\Xi)$ of $M(\Xi)$ will be called \emph{Kastler $C^*$-algebra} of
$\Xi$. Moreover, by \cite[Th.\ 8]{LM2} we have a continuous embedding
\begin{equation}\label{eq:distrib}
  \rM(\Xi)\subset\mbs'(\Xi) .
\end{equation}
If $\Xi$ is infinite dimensional we have \eqref{eq:fil} and Kastler's
theorem implies
\begin{equation}\label{eq:EFkast}
  E,F\in\mbg_\rms(\Xi) \text{ and } E\subset F \Longrightarrow
  \rM(E)\subset\rM(F) \text{ isometrically} 
\end{equation}
hence one may define its \emph{Kastler $C^*$-algebra} as
\begin{equation}\label{eq:kastgen}
  \rM(\Xi)= \text{ inductive limit of the $C$*-algebras }
  \rM(E) \text{ with } E\in\mbg_\rms(\Xi) .
\end{equation}
From now on, since $\Xi$ is fixed, we simplify the notation and set
$\rM=\rM(\Xi)$.  We keep the notations $\mu,\nu$ for the elements of
$\rM$ and $\mu\oast \nu$ for their product.  Then for any
$E\in\mbg(\Xi)$ we have $E\subset F$ for some $F\in\mbg_\rms(\Xi)$
hence $M(E)\subset M(F)\subset\rM$ and
\begin{equation}\label{eq:rmeXi}
  \rM(E)\doteq\text{ closure of } M(E) \text{
    in } \rM. 
\end{equation}
is a $C^*$-subalgebra of $\rM$ and if $E$ is symplectic $\rM(E)$ is
canonically isomorphic to the Kastler algebra of $E$, so depends only
on the restriction of $\sigma$ to $E$. We have (see page
\pageref{p:E+F})
\begin{equation}\label{eq:kgrad}
E,F\in\mbg(\Xi)\Rightarrow \rM(E)\oast\rM(F)\subset\rM(E+F). 
\end{equation}
For any $\xi\in\Xi$ let $\delta_\xi$ be the Dirac measure at $\xi$
thought as element $\delta_\xi\in M(E)\subset\rM$ for any
$E\in\mbg(\Xi)$ such that $\xi\in E$. If $\xi,\eta\in\Xi$ we clearly
have
\begin{equation}\label{eq:ccrm}
  \delta_\xi^*=\delta_{-\xi} \quad\text{and}\quad
  \delta_\xi\oast\delta_\eta = \e^{-\frac{\rmi}{2}\sigma(\xi,\eta)}
  \delta_{\xi+\eta} .
\end{equation}
Thus $\delta_\xi$ is a unitary element of $\rM$ hence
$\|\delta_\xi\|=1$.  Let
$\xi_\sigma\doteq\sigma(\cdot,\xi)\colon\Xi\to\R$ and $\tau_\xi$ the
translation by $\xi$ defined by $(\tau_\xi
f)(\eta)=f(\eta-\xi)$. Clearly, if $E\in\mbg(\Xi)$ and $\mu\in M(E)$
\begin{equation}\label{eq:products}
   \delta_\xi \oast \mu = \e^{\frac{\rmi}{2}\xi_\sigma}
   \tau_\xi\mu, \quad 
   \mu \oast \delta_{-\xi}= \e^{\frac{\rmi}{2}\xi_\sigma}
   \tau_{-\xi}\mu, \quad 
   \delta_\xi \oast \mu \oast \delta_{-\xi}= \e^{\rmi \xi_\sigma}\mu  .
 \end{equation}  
 If $f:\Xi\to\R$ we write $\lim_{\xi\to0}f(\xi)=0$ if this holds for
 the restriction of $f$ to any finite dimensional subspace of $\Xi$.
 Then by a density and continuity argument we easily get
\begin{equation}\label{eq:1stcondI}
  \lim_{\xi\to0}\|\delta_\xi^*\oast\mu\oast\delta_{\xi}-\mu\| =
  \lim_{\xi\to0}\|\delta_\xi\oast\mu-\mu\oast\delta_\xi\| =0
  \quad \forall \mu\in\rM.
\end{equation}

\subsubsection{Representations}\label{sss:Krep}

If we embed $\Xi\subset\rM(\Xi)$ by using the map
$\xi\mapsto\delta_\xi$ then the representations of $\Xi$ naturally
extend to representations of the Kastler algebra $\rM(\Xi)$.
Proposition \ref{pr:wmorph} is a simple consequence of Stone-Von
Neumann's and Kastler's theorems.

\begin{remark}\label{re:stvon}{\rm We recall the statement Stone-Von
    Neumann theorem \cite{DeG2,Fol} because we will need several times
    the notations introduced here: any finite dimensional $\Xi$ has an
    irreducible representation $W_0:\Xi\to U(\ch_0)$ and if
    $W:\Xi\to U(\ch)$ is an arbitrary representation then there is a
    Hilbert space $\ck$ and a unitary operator
    $V:\ch\to\ch_0\otimes\ck$ such that
    $VW(\xi) V^*=W_0(\xi)\otimes1_\ck$ $\forall\ \xi\in\Xi$.  The
    \emph{multiplicity of $W$} is the dimension of $\ck$ and the
    \emph{commutant of $W$} is the set
    $\Com^\ssw(\Xi)=\{T\in B(\ch)\mid [T,W(\xi)]=0\,\forall\xi\}$.
    Clearly $V\Com^\ssw(\Xi)V^*=1_{\ch_0}\otimes B(\ck)$.
  }\end{remark}

\begin{proposition}\label{pr:wmorph}
  If $W:\Xi\to U(\ch)$ is a representation of the symplectic space
  $\Xi$ then there is a unique $C^*$-algebra representation
  $W:\rM(\Xi)\to B(\ch)$ such that
\begin{equation}\label{eq:wme}
  W(\mu) = \int_E W(\xi) \mu(\d\xi) 
  \quad  \forall E\in\mbg(\Xi) \text{ and } \forall\mu\in M(E) .
\end{equation}
This representation is faithful. 
\end{proposition}

\begin{proof}

  First assume $\Xi$ finite dimensional and define
  $W:M(\Xi)\to B(\ch)$ by \eqref{eq:wme} with $E=\Xi$.  We have
  $W(\mu\oast\nu)=W(\mu) W(\nu)$, $W(\mu^*)=W(\mu)^*$ and
  $W(\delta_\xi)=W(\xi)$ hence $W:M(\Xi)\to B(\ch)$ is a unital
  morphism such that $\|W(\mu)\|\leq\|\mu\|_1$.  To prove that this
  morphism is injective we use the Stone-Von Neumann theorem with the
  notations used above.  From $VW(\xi)V^*=W_0(\xi)\otimes1_\ck$ we get
  $VW(\mu)V^*=W_0(\mu)\otimes1_\ck$ hence
  $\|W(\mu)\|=\|W_0(\mu)\|\ \forall\mu\in M(\Xi)$. Hence it suffices
  to find one representation of $\Xi$ whose associated morphism is
  injective; this is the case for the regular representation
  \cite[Th.\ 5]{Kas}.

  Since $W:M(\Xi)\to B(\ch)$ is a faithful representation the relation
  $\|\mu\|_0=\|W(\mu)\|$ defines a $C^*$-norm on $M(\Xi)$ hence by
  Kastler's theorem we see that $\|W(\mu)\|=\|\mu\|$ and this clearly
  implies the statement of the theorem if $\Xi$ is finite dimensional.

  If $\Xi$ is infinite dimensional and $F\in\mbg_\rms(\Xi)$ the
  restriction $W|F$ is a representation of $F$ so it extends to an
  isometric representation $W_F\colon\rM(F)\to B(\ch)$.  If
  $F,G,H\in\mbg_\rms(\Xi)$ and $F+G\subset H$ then $W_H|\rM(F)=W_F$
  and $W_H|\rM(G)=W_G$ so $W$ induces a $C^*$-algebra representation
  $\rM(\Xi)\to B(\ch)$ that we also denote $W$ which clearly satisfies
  \eqref{eq:wme}. Since the restriction of $W$ to any $\rM(F)$ is
  isometric and $\rM(\Xi)$ is the closure of the union of such
  $\rM(F)$, it follows that the $W\colon \rM(\Xi)\to B(\ch)$ is
  isometric.
\end{proof}

\begin{definition}\label{df:wkast}
  The \emph{Kastler $C^*$-algebra in the representation $W$} is the
  norm closure $\rwk(\Xi)$ of the set of operators $W(\mu)$ with
  $\mu\in M(\Xi)$. The set $\rwk(\Xi)$ is a unital $C^*$-subalgebra
  of $B(\ch)$ and $\mu\mapsto W(\mu)$ extends to an isomorphism
  $W\colon\rM(\Xi)\to\rwk(\Xi)$.
\end{definition}

\begin{remark}\label{re:rwk}{\rm By Proposition \ref{pr:wmorph}, if
    $W_1,W_2$ are representations of $\Xi$ then the algebras
    $\xk{\ssw_{\!1}}(\Xi)$ and $\xk{\ssw_{\!2}}(\Xi)$ are canonically
    isomorphic in the following sense: there is a unique morphism
    $\xk{\ssw_1}(\Xi)\to\xk{\ssw_2}(\Xi)$ which sends $W_1(\xi)$ into
    $W_2(\xi)$ for all $\xi\in\Xi$. For this reason}
    we drop the superscript $W$ and write $\rk(\Xi)$ or just $\rk$
    unless $W$ plays a role.
\end{remark}

\begin{remark}\label{re:wmorK}

  {\rm If $\dim\Xi<\infty$ then $\rwk(\Xi)\cong
    \xk{\ssw_0}\otimes 1_{\ck}$ with the notations of Remark
    \ref{re:stvon}.} 

\end{remark}

\subsubsection{C*-norm on \texorpdfstring{$L^1$}{L1}}\label{sss:cnL1}

If $E\in\mbg(\Xi)$ is not symplectic the $C^*$-algebra $\rM(E)$ is
defined in \eqref{eq:rmeXi} but the $C^*$-norm on $M(E)$ is not
unique, it depends on the embedding of $E$ in $\Xi$. But the situation
is better if we replace $M(E)$ by $L^1(E)$: there is a unique
$C^*$-norm on $L^1(E)$ without non-degeneracy condition on
$\sigma$. This is interesting since it may be used to construct the
field algebra without referring to Kastler algebra (\conf comment
after \eqref{eq:rldf2}).

\begin{theorem}\label{th:semisymp}
  Let $\Xi$ be a finite dimensional real vector space equipped with a
  Lebesgue measure and a bilinear anti-symmetric form
  $\sigma:\Xi\times \Xi\to\R$. Equip $L^1(\Xi)$ with the Banach
  $*$-algebra structure defined by the involution
  $f^*(\xi)=\bar{f}(-\xi)$ and the twisted convolution
\begin{equation}\label{eq:L1twist}
  (f\oast g)(\xi)=\int \e^{-\frac{\rmi}{2}\sigma(\xi,\eta)}
  f(\xi-\eta)g(\eta) \dd\eta.
\end{equation}
Then there is a unique $C^*$-norm on $L^1(\Xi)$ (also independent of
the chosen measure $\d\eta$).
\end{theorem}

\begin{proof}

  If $\sigma$ is non-degenerate the assertion follows from Stone-Von
  Neumann theorem and if $\sigma=0$ then this holds for arbitrary
  locally compact abelian groups, \conf \cite[p.\ 224]{Boi}.

  We will need the following theorem of B.A. Barnes \cite[Thm.\
  5.5]{Bar}.  Recall that a Banach $*$-algebra $\ca$ is postliminal if
  its enveloping $C^*$-algebra $\tilde\ca$ is postliminal, \ie the
  range of each irreducible representation of $\tilde\ca$ contains a
  nonzero compact operator (then it will contain all compact
  operators) \cite[\S4.1.10 and Theorem 9.1]{Dix}. The theorem says
  that if $\ca,\cb$ are postliminal Banach $*$-algebras and have
  unique $C^*$-norms, then the algebraic tensor product
  $\ca\otimes\cb$ has a unique $C^*$-norm. Note that the convolution
  algebra of an amenable group can have several distinct $C^*$-norms
  \cite[p.\ 1]{Bar}.

  Let
  $\Xi^\sigma=\{\xi\in \Xi\mid \sigma(\xi,\eta)=0\ \forall\eta\in
  \Xi\}$ and $\Theta$ a subspace supplementary to $\Xi^\sigma$ in
  $\Xi$, so $\Xi=\Xi^\sigma\oplus \Theta$. Note that $\sigma=0$ on
  $\Xi^\sigma$. The restriction of $\sigma$ to $\Theta$ is
  non-degenerate because if $\xi\in \Theta$ and $\sigma(\xi,\eta)=0$
  for all $\eta\in \Theta$ then $\sigma(\xi,\eta)=0$ for all
  $\eta\in \Xi$ hence $\xi\in \Xi^\sigma$ so $\xi=0$. Thus $\Theta$ is
  a symplectic space when equipped with $\sigma|_{\Theta^2}$.

  We have $L^1(\Xi)=L^1(\Xi^\sigma)\hat\otimes L^1(\Theta)$ where
  $\hat\otimes$ denote the projective tensor product. It is clear that
  this equality holds at the Banach $*$-algebra level if we equip
  $L^1(\Xi)$ with the twisted convolution product \eqref{eq:L1twist},
  $L^1(\Xi^\sigma)$ with the usual convolution product, and
  $L^1(\Theta)$ with the twisted convolution associated to its
  symplectic form $\sigma|_{\Theta^2}$.

  We apply now Barnes' theorem to the algebraic tensor product
  $L^1(\Xi)=L^1(\Xi^\sigma)\otimes L^1(\Theta)$. The $C^*$-envelop of
  $L^1(\Xi^\sigma)$ is the abelian algebra of $C_0$ functions on the
  dual of $\Xi^\sigma$, hence is postliminal. Moreover, as we noticed
  at the beginning of the proof, there is only one $C^*$-norm on
  $L^1(\Xi^\sigma)$. On the other hand, the Stone-Von Neumann theorem
  applied to the symplectic space $\Theta$ implies that the
  $C^*$-envelop of $L^1(\Theta)$ is the $C^*$-algebra of compact
  operators on the Hilbert space of the unique (modulo unitary
  equivalence) irreducible representation of $\Theta$, hence
  $L^1(\Theta)$ is postliminal and has a unique $C^*$-norm.  Thus
  there is a unique $C^*$-norm on $L^1(\Xi^\sigma)\otimes L^1(\Theta)$
  hence on its projective completion $L^1(\Xi)$.
\end{proof}

\subsection{Components of the field algebra}\label{ss:components}

Let $\Xi$ be a symplectic space.  If $E\in\mbg(\Xi)$ then $L^1(E)$ is
an ideal of $M(E)$ so its closure in $\rM(E)$ is a $C^*$-subalgebra
and an ideal of $\rM(E)$ that we denote
\begin{equation}\label{eq:rldf1}
  \rl(E) \doteq \text{ closure of the $*$-algebra } L^1(E) \text{ in }
  \rM(E) .
\end{equation}
In particular $\rl(0)=\C$.  By Theorem \ref{th:semisymp} we have
\begin{equation}\label{eq:rldf2}
  \rl(E)=\text{ completion of } L^1(E) \text{ for its unique
    $C^*$-norm} 
\end{equation}
so we could use Mageira's reconstruction theorem \cite[Th.\
2.2]{M1} to define the field algebra independently of the Kastler
algebra. However, viewing it as a subalgebra of $\rM(\Xi)$ could be
useful for further developments, \eg for the treatment of anisotropic
operators mentioned in Example \ref{ex:anisotropnaf}. In this
subsection we study the structure of the algebras $\rl(E)$. 

If $\dim\Xi<\infty$ we have a simple intrinsic description of
$\rl(E)$:

\begin{theorem}\label{th:intrinsicE}
  If $\Xi$ is finite dimensional then $\rl(E)$ is the set of
  $\mu\in\rM$ such that
  \begin{equation}\label{eq:intrinsicE}
\lim_{\xi\in E,\xi\to0}\|\delta_\xi\oast\mu-\mu\|=0
\quad\text{and}\quad
\delta_\xi\oast\mu=\mu\oast\delta_\xi \text{ if } \xi \in E^\sigma.
  \end{equation}
\end{theorem}

This is a corollary of Theorem \ref{th:mainthfd} which is a much
deeper fact because it does not involve a condition like $\mu\in\rM$
where $\rM$ is a ``complicated'' algebra.

We call \emph{center} of $E\in\mbg(\Xi)$ the space
$E_\rmc=E\cap E^\sigma=\{\xi\in E\mid \sigma(\xi,\eta)=0\
\forall\eta\in E\}$. Clearly $E/E_\rmc$ is naturally a symplectic
space, so the $C^*$-algebra $\rl(E/E_\rmc)$ is well defined. $E$ is
symplectic if and only if $E_\rmc=0$, $E$ is isotropic if and only if
$E_\rmc=E$, and $E$ is coisotropic if and only if $E_\rmc=E^\sigma$.

If $E^*$ is the dual of $E$ then
$\Xi\ni\xi\mapsto\xi_\sigma|_{E}\in E^*$ is surjective with kernel
$E^\sigma$ hence we may identify $E^*=\Xi/E^\sigma$. Thus if $\Xi$ is
finite dimensional $C_0(E^*)\subset \cbu(\Xi)\subset\mbs'(\Xi)$.

A $C^*$-algebra $\ca$ is called \emph{elementary} if it is isomorphic
to the $C^*$-algebra of compact operators on a separable Hilbert
space. \label{p:elementary} The $C^*$-tensor products at (3) and (4)
below are unambiguously defined, both factors being nuclear algebras.

\begin{theorem}\label{th:structure}
  Let $E$ be a finite dimensional subspace of $\Xi$.\\[1mm]
  {\rm(1)} $E$ is isotropic if and only $\rl(E)$ is abelian and
  then $\rl(E){\cong} C_0(E^*)$.\\[1mm]
  {\rm(2)} $E$ is symplectic if and only if $\rl(E)$ is
  elementary.\\[1mm]
  {\rm(3)} $\rl(E)\simeq\rl(E_\rmc)\otimes\rl(E/E_\rmc)$.
  \\[1mm]
  {\rm(4)} If $E$ is coisotropic then
  $\rl(E)\simeq\rl(E^\sigma)\otimes\rl(E/E^\sigma)$. \\[1mm]
  {\rm(5)} $\rl(E)$ is unital if and only if $E=0$ and then $\rl(0)=\C$. 
  \end{theorem}

\begin{proof}  
  $\rl(E)$ is abelian if and only $\mu\oast\nu=\nu\oast\mu$
  $\forall\mu,\nu\in L^1(E)$ which by \eqref{eq:tconv} means
\begin{equation*}
\int \e^{-\frac{\rmi}{2}\sigma(\xi,\eta)} \mu(\xi-\eta)\nu(\eta) \d\eta=
\int \e^{\frac{\rmi}{2}\sigma(\xi,\eta)} \mu(\xi-\eta) \nu(\eta)\d\eta
\quad \forall \mu,\nu\in C_\rmc(E), \ \forall \xi\in E.
\end{equation*}
Integrating over $\xi$ we then get
\[
  \iint \left(\e^{-\frac{\rmi}{2}\sigma(\xi,\eta)} -
    \e^{\frac{\rmi}{2}\sigma(\xi,\eta)} \right) \mu(\xi)\nu(\eta)
  \d\xi\d\eta =0 \quad \forall \mu,\nu\in C_\rmc(E) .
\]
hence
$\e^{-\frac{\rmi}{2}\sigma(\xi,\eta)} -
\e^{\frac{\rmi}{2}\sigma(\xi,\eta)}=0\ \forall \xi,\eta\in E$ which
clearly is equivalent to $\sigma|_{E^2}=0$.

We now describe the canonical isomorphism $\rl(E)\cong C_0(E^*)$ for
isotropic $E$.  Then for $\mu,\nu\in L^1(E)$ the twisted convolution
product $\mu\oast\nu$ coincides with the ordinary convolution product
$\mu\star\nu$. Thus $L^1(E)$ is the usual Banach $*$-algebra
associated to the convolution product on $E$ whose enveloping
$C^*$-algebra is identified with $C_0(E^*)$ via a Fourier
transformation. We still need to show that the norm induced by
$\rM(\Xi)$ on $L^1(E)$ coincides with the norm induced by the
enveloping $C^*$-algebra. But this follows from the fact that for an
arbitrary locally compact abelian group $E$ there is only one
$C^*$-norm on the convolution algebra $L^1(E)$ \cite[p.\ 224]{Boi}.
This proves (1).

If $E$ is symplectic then by definition $\rl(E)$ is the completion of
the twisted convolution algebra $L^1(E)$ in $\rM(E)$.  If $W$ is an
irreducible representation of $E$ then from Proposition
\ref{pr:wmorph} and the Stone-Von Neumann theorem it follows that
$\rl(E)$ is isomorphic to the $C^*$-algebra generated by the operators
$W(\mu) = \int_E W(\xi) \mu(\d\xi)$ with $\mu\in L^1(E)$, which is the
$C^*$-algebra of compact operators on the Hilbert space of the
representation (obvious in the Schr\"odinger representation).
Reciprocally, if $E$ is not symplectic then $E_\rmc\neq0$ hence
$\rl(E)$ is not elementary by (3).

Assertion (3) follows from Theorem \ref{th:semisymp} and its proof. If
$F$ is a subspace supplementary to $E_\rmc$ in $E$, so
$E=E_\rmc\oplus F$, then $L^1(E)=L^1(E_\rmc)\hat\otimes L^1(F)$
projective tensor product of the usual convolution algebra
$L^1(E_\rmc)$ with the twisted convolution algebra $L^1(F)$ defined by
the symplectic structure of $F$. By the uniqueness of the $C^*$-norms
of these algebras and \cite[Thm.\ 5.5]{Bar} we get
$\rl(E)=\rl(E_\rmc)\otimes\rl(F)$ which is the uniquely defined
$C^*$-tensor product. This proves (3).  Finally, (4) is a particular
case of (3) and (5) is true because by the preceding description the
algebras $\rl(E)$ with $E\neq0$ cannot have unit.
\end{proof}

\subsection{Field algebra}\label{ss:FA}

The set of $C^*$-subalgebras $\rl(E)$ of $\rM$ has some remarkable
properties summarised in the next theorem.  If $\ca_1,\dots,\ca_n$ are
subsets of $\rM$ then $\ca_1\oast\dots\oast\ca_n$ is the linear span
of the products $\mu_1\oast\dots\oast\mu_n$ with $\mu_i\in\ca_i$ and
$\ca_1\dot\oast\dots\dot\oast\ca_n$ its closure.

\begin{theorem}\label{th:remarc}
{\rm(1)} The family of linear subspaces $\rl(E)$ of $\rM$ is linearly
independent. \\[1mm]
{\rm(2)} If $E,F\in\mbg(\Xi)$ and $\mu\in\rl(E),\nu\in\rl(F)$ then
$\mu\oast\nu\in\rl(E+F)$. \\[1mm]
{\rm(3)} $\rl(E)\dot\oast\rl(F)=\rl(E+F)$. \\[1mm]
{\rm(4)} If $\cs\subset\mbg(\Xi)$ is finite then
$\sum_{E\in\cs}\rl(E)$ is norm closed.
\end{theorem}

From (1) and (2) it follows that the linear subspace generated by the
$\rl(E)$, denoted
\begin{equation}\label{eq:rl0}
  \mathring\rl\doteq\tsum{E\in\mbg(\Xi)}{}\rl(E),
\end{equation}
is a unital $*$-subalgebra of $\rM$ hence its closure is a unital
$C^*$-subalgebra.  The sum in \eqref{eq:rl0} is direct hence if
$\mu\in\mathring\rl$ its component in $\rl(E)$ is uniquely determined;
we denote it $\mu(E)$, so that $\mu=\sum_E\mu(E)$.

\begin{definition}\label{df:deff}
  The \emph{field $C^*$-algebra} $\rl\equiv\rl^{\ssxi}$ of $\Xi$ is the
  closure of $\mathring\rl$:
  \begin{equation}\label{eq:deff}
    \rl\doteq\tsum{E\in\mbg(\Xi)}{\rmc}\rl(E).
  \end{equation}
\end{definition}

Parts (1) and (2) of Theorem \ref{th:remarc} and the preceding
definition say that \emph{$\rl$ is equipped with a $\mbg(\Xi)$-graded
  $C^*$-algebra structure with components $\rl(E)$}.  

In the rest of this section we prove Theorem \ref{th:remarc}. Note
that
\begin{equation}\label{eq:Kcommut}
\mu\in\rM(E),\,\nu\in\rM(F) \text{ with } F\subset E^\sigma
\Rightarrow  \mu\oast\nu=\nu\oast\mu.
\end{equation}
Indeed, it suffices to show that $\mu\oast\nu=\nu\oast\mu$ if
$\mu\in M(E)$ and $\nu\in M(F)$. If $G\in\mbg_\rms(\Xi)$ with
$E+F\subset G$ and $f\in C_0(G)$ then by \eqref{eq:tconv} 
\[
  \int f(\xi)(\mu \oast \nu)(\d\xi)= \int_{E\times F}
  \e^{\frac{\rmi}{2}\sigma(\xi,\eta)} f(\xi+\eta)\mu(\d\xi)\nu(\d\eta)
  = \int_{E\times F} f(\xi+\eta)\mu(\d\xi)\nu(\d\eta)
\]
because $\sigma(\xi,\eta)=0$ if $\xi\in E,\eta\in F$ and we have the
same formula for $\nu\oast\mu$.

\begin{lemma}\label{lm:liminf}
  If $\Xi$ is finite dimensional,
  $\mu=\sum_{E\in\mbg(\Xi)}\mu(E)\in\mathring\rl$, and $\xi\in\Xi$,
  then
\begin{equation}\label{eq:liminf}
  \lim_{r\to\infty} \delta_{r\xi} \oast \mu \oast \delta_{-r\xi}
  =\tsum{\xi\in F^\sigma}{} \mu(F) \quad
  \text{in } \mbs'(\Xi) .
\end{equation}
\end{lemma}

\begin{proof}
  We have
  \begin{equation}\label{eq:mudelta}
    \delta_\xi \oast \mu \oast \delta_{-\xi}=
    \e^{\rmi\xi_\sigma}\mu \quad \forall \mu\in\rM(\Xi) . 
  \end{equation}
  Indeed, \eqref{eq:products} ensures this for $\mu\in M(\Xi)$ and
  then the relation extends to any $\mu\in\mbs'(\Xi)$ by continuity
  and density of $M(\xi)$ in $\mbs'(\Xi)$; then we use the embedding
  $\rM(\Xi)\subset\mbs'(\Xi)$.

  Now let $E$ be a subspace of $\Xi$ and $\mu\in\rl(E)$. From
  \eqref{eq:Kcommut} and \eqref{eq:ccrm} we get
  \begin{equation}\label{eq:fix}
    \delta_\xi \oast \mu \oast \delta_{-\xi}=\mu \quad\text{if }
    \xi\in E^\sigma .
  \end{equation}
  On the other hand, as weak limit in $\mbs'(\Xi)$ we have
  \begin{equation}\label{eq:limmod}
    \lim_{r\to\infty} \delta_{r\xi} \oast \mu \oast \delta_{-r\xi}=0
    \quad\text{if } \xi\nin E^\sigma.
  \end{equation}
  For the proof, note that, by the last relation in
  \eqref{eq:products} and a density and continuity argument, for any
  $\mu,\nu\in\rM(\Xi)$ we have
  \[
    \|\delta_{r\xi} \oast \mu \oast \delta_{-r\xi} -\delta_{r\xi}
    \oast \nu \oast \delta_{-r\xi}\|\leq\|\mu-\nu\|
  \]
  hence, since $\rl(E)$ is the norm closure of $L^1(E)$ in $\rM(\Xi)$,
  it suffices to prove \eqref{eq:limmod} for $\mu\in L^1(E)$.  If
  $\xi\nin E^\sigma$ then $\xi_\sigma|_E$ is a nonzero linear form on
  $E$ hence by \eqref{eq:mudelta} if $\theta\in\mbs(\Xi)$ 
  \begin{equation}\label{eq:limmod1}
    \bracet{\theta}{\delta_{r\xi} \oast \mu \oast \delta_{-r\xi}}
    =\bracet{\theta}{\e^{\rmi r\xi_\sigma}\mu}
    =\int_E\e^{\rmi r\sigma(\xi,\eta)}\theta(\eta)\mu(\eta) \d \eta
  \end{equation}
  tends to zero as $r\to\infty$ by the Riemann-Lebesgue lemma. This
  proves \eqref{eq:liminf}.
\end{proof}

We prove (1) of the Theorem \ref{th:remarc}, \ie if $\mu(E)\in\rl(E)$
$\forall E\in\cs=$ finite, then
\begin{equation}\label{eq:dirsum}
  \tsum{E\in\cs}{}\mu(E)=0 \Longrightarrow \mu(E)=0\ \forall E\in\cs .
\end{equation}
Since there is a finite dimensional symplectic subspace of $\Xi$
containing all the subspaces $E\in\cs$, it suffices to show this under
the assumption that $\Xi$ is finite dimensional.  For any $\xi\in\Xi$
we will have
\[
  \tsum{E\in\cs}{} \lim_{r\to\infty} \delta_{r\xi} \oast \mu(E) \oast
  \delta_{-r\xi}=0.
\]
Fix $F\in\cs$ and consider the spaces $E\in\cs$ such that
$E\not\subset F$. For such $E$ we have $F^\sigma\not\subset E^\sigma$
hence $E^\sigma\cap F^\sigma$ is a strict subspace of $F^\sigma$, so
we may choose $\xi\in F^\sigma$ which does not belong to any of these
subspaces, hence $\xi\nin E^\sigma$ if $E\not\subset F$. On the other
hand, if $E\subset F$ then $\xi\in F^\sigma\subset E^\sigma$. Thus
from \eqref{eq:fix} and \eqref{eq:limmod} we get
\[
  \tsum{E\subset F}{}\mu(E)=0 \quad \forall F\in\cs .
\]
If $F$ is minimal in $\cs$ we get $\mu_F=0$. Then if $\cs_1$ is the
set of $E\in\cs$ which are not minimal, we get
$\sum_{E\in\cs_1}\mu(E)=0$. By repeating the above argument for
$\cs_1$ we get $\mu_F=0$ for all $F$ minimal in $\cs_1$, etc.  This
proves part (1) of the theorem

For the proof of (2) of Theorem \ref{th:remarc} we may assume $\Xi$
finite dimensional, \conf \eqref{eq:fil}.  If $\mu\in M(E)$ and
$\nu\in M(F)$ then $\supp(\mu\oast\nu)\subset E+F$ because if
$f\in C_0(\Xi)$ has support disjoint from $E+F$ then the right hand
side of \eqref{eq:tconv} is equal to zero. Thus
$\rM(E)\dot\oast\rM(F)\subset\rM(E+F)$, in particular we have
\eqref{eq:kgrad}.  Alternatively, we may use
$|\mu\oast\nu|\leq|\mu|\star|\nu|$, where $\star$ is the ordinary
convolution operation of measures, and the following known fact: if
$\mu,\nu$ are positive bounded measures on $\Xi$ then
$\supp(\mu\star\nu)$ is included in the closure of
$\supp\mu+\supp\nu$. \label{p:E+F}

Now we prove that $\mu\oast\nu\in L^1(E+F)$ if
$\mu\in L^1(E),\nu\in L^1(F)$.  By the preceding comments, it suffices
to prove that $\mu\star\nu\in L^1(E+F)$ if $\mu,\nu$ are positive
measures in $L^1(E)$ and $L^1(F)$ respectively.  Denote
$\mu\otimes\nu$ the product measure on $E\oplus F$. This is clearly an
absolutely continuous positive bounded measure on $E\oplus F$ and if
we denote $S:E\oplus F\to E+F$ the sum operation
$S(\xi,\eta)=\xi+\eta$ then $\mu\star\nu$ is the bounded positive
measure on $E+F$ defined by
$(\mu\star\nu)(A)=(\mu\otimes\nu)(S^{-1}(A))$ for any Borel set
$A\subset E+F$. We have to show that this measure is absolutely
continuous. But $S$ is a linear surjective map, hence if
$N\subset E+F$ is of measure zero then $S^{-1}(N)$ is of measure
zero. Indeed, if $G$ is a subspace of $E\oplus F$ supplementary to
$\ker S$ then $S:G\to E+F$ is a linear bijective map, so
$M=S^{-1}(N)\cap G$ is of measure zero and $S^{-1}(N)=G\oplus\ker S$
is also of measure zero by Fubini theorem.

Thus we have $\rl(E)\oast\rl(F)\subset\rl(E+F)$ and for the proof of
(3) of Theorem \ref{th:remarc} it remains to show that 
$\rl(E)\oast\rl(F)$ is dense in $\rl(E+F)$. For this it suffices to
show that the only function $f\in L^\infty(E+F)$ such that
$\int f(\xi)(\mu\oast\nu)(\d\xi)=0$ for all $\mu\in L^1(E)$ and
$\nu\in L^1(F)$ is $f=0$. More explicitly, this condition means
\begin{equation*}
  \int_E\int_F \e^{-\frac{\rmi}{2}\sigma(\xi,\eta)}
  f(\xi+\eta)\mu(\xi)\nu(\eta)\d\xi\d\eta =0 \quad\text{for all }
  \mu\in L^1(E),\nu\in L^1(F) ,
\end{equation*}
where $\d\xi,\d\eta$ are the Euclidean measures associated to some
scalar products on $E,F$ respectively. For $\xi_0\in E$,
$\eta_0\in F$, and $r>0$ real let us take above $\mu$ equal to the
characteristic function of the ball $|\xi-\xi_0|<r$ in $E$ divided by
its volume $B_E(r)$ and $\nu$ the similarly defined function on
$F$. Then we get
\begin{equation*}
  \frac{1}{B_E(r) B_F(r)}\int_{|\xi-\xi_0|<r}\int_{|\eta-\eta_0|<r}
  \e^{-\frac{\rmi}{2}\sigma(\xi,\eta)}
  f(\xi+\eta)\d\xi\d\eta =0.
\end{equation*}
By the Lebesgue differentiation theorem, the limit as $r\to0$ of the
left hand side above is equal to
$\e^{-\frac{\rmi}{2}\sigma(\xi_0,\eta_0)} f(\xi_0+\eta_0)$ for almost
every $\xi_0\in E,\eta_0\in F$. Thus $f=0$.

This finishes the proof of (3).  By (1), (2) and \eqref{eq:deff} $\rl$
is a $\mbg(\Xi)$-graded $C^*$-algebra hence (4) is true by (1) of
Lemma \ref{lm:subsemi}. This finishes the proof of the theorem.

\begin{remark}\label{re:kdefofrl}{\rm From Theorems \ref{th:remarc}
    and \ref{th:structure} it follows that for each $L\in\mbp(\Xi)$ we
    have an abelian $C^*$-subalgebra $\rl(L)$ of $\rl$ with
    $\rl(L)\cong C_0(L)\simeq C_0(\R)$ such that the family of
    subspaces $\{\rl(L)\}_{L\in\mbp(\Xi)}$ is linearly independent and
    generates $\rl$. In other terms, the linear sum
    $\ri=\sum_{L\in\mbp(\Xi)}\rl(L)\subset\rl$ is direct and the
    $C^*$-algebra generated by $\ri$ is $\rl$. Thus $\rl$ is the field
    $C^*$-algebra introduced by Kastler in \cite[\S6]{Kas}.
  }\end{remark}

\begin{remark}\label{re:lestruct}

  {\rm By Theorem \ref{th:structure} the algebras $\rl(E)$ have a
    rather simple structure: they are isomorphic to $C^*$-algebras of
    the form $\ra\otimes\rb$ with $\ra$ an abelian $C^*$-algebra and
    $\rb$ an elementary $C^*$-algebra. This and \cite[Prop.\ 4.2]{M1}
    imply for example that \emph{$\rl$ is a nuclear
      $C^*$-algebra}. This has been proved before in \cite[Th.\
    3.8]{Bu14} by other techniques.}

\end{remark}

\subsection{Projection morphisms} \label{ss:projmor}

For any subset $\cs\subset\mbg(\Xi)$ we set
 \begin{equation}\label{eq:sumcs}
 \mathring\rl(\cs)\doteq\tsum{E\in\cs}{}\rl(E)\quad\text{and}\quad
 \rl(\cs)=\tsum{E\in\cs}{\rmc}\rl(E)=\text{ closure of }
 \mathring\rl(\cs) .
\end{equation}
By Theorem \ref{th:remarc} if $\cs$ is a subsemilattice of
$\mbg(\Xi)$ then $\mathring\rl(\cs)$ is a $*$-subalgebra and
$\rl(\cs)$ an $\cs$-graded $C^*$-subalgebra of $\rl$.  If $\cs$ is
finite then $\rl(\cs)=\mathring\rl(\cs)$ is a $C^*$-subalgebra of
$\rl$ if and only if $\cs$ is a subsemilattice of $\mbg(\Xi)$, by
Lemma \ref{lm:subsemi} and Theorem \ref{th:remarc}-(3).

The linear direct sum decomposition
$\mathring\rl=\sum_{E\in\mbg(\Xi)}\rl(E)$ gives us a linear projection
$\mathring\cp(\cs)\colon \mathring\rl\to \mathring\rl(\cs)$ for any
$\cs$.  Then
$\mathring\cp(E)=\mathring\cp(\{E\}): \mathring\rl\to\rl(E)$ is the
linear projection determined by the direct sum decomposition and
$\mathring\cp(\cs)=\sum_{E\in\cs}\mathring\cp(E)$.  If
$\mathring\cp(\cs)$ extends to a continuous map $\rl\to\rl(\cs)$
we denote $\cp(\cs)$ the extension, this will be a projection
of $\rl$ onto the subspace $\rl(\cs)$.

Any subspace $E\subset\Xi$ determines three subsemilattices of
interest defined as in \eqref{eq:sas} and the corresponding
$C^*$-subalgebras will be denoted
$\rl_\sse,\rl'_\sse,\rl_{\ss{\supset E}}$. Thus
\begin{equation}\label{eq:deffle}
  \rl_\sse\doteq\tsum{\ss{F\subset E}}{\rmc} \rl(F),\quad
  \rl'_\sse \doteq{\textstyle\sum^\rmc_{F\ss{\not\subset E}}}\rl(F),
  \quad \rl_{\ss{\supset E}} = \tsum{F\supset E}{\rmc} \rl(F).
\end{equation}
Clearly $\rl_\sse $ is a unital $C^*$-subalgebra of $\rl$ graded by
$\mbg(E)$ and $\rl'_\sse$ and $\rl_{\ss{\supset E}}$ are ideals graded
by the obvious subsemilattices. From Theorem \ref{th:gmain} and
Proposition \ref{pr:ecapf} we get:
\begin{compactenum}
\item $\rl = \rl_\sse + \rl'_\sse$ and $\rl_\sse \cap \rl'_\sse = 0$.
\item The projection $\cp_\sse\colon \rl\to\rl_\sse$ determined by
  the preceding direct sum decomposition is a morphism such that
  $\mu=\tsum{\ssf}{}\mu(F)\in\mathring\rl \Rightarrow
  \cp_\sse\mu=\tsum{\ss{F\subset E}}{}\mu(F)$.
\item  $\rl_{\ss{E\cap F}}=\rl_\sse\cap\rl_\ssf \quad\text{and}\quad
  \cp_{\ss{E\cap F}}=\cp_\sse\cp_\ssf=\cp_\ssf\cp_\sse$. 
\end{compactenum}  
Thus \emph{$\cp_E$ is a projection morphism of $\rl$ onto its
  subalgebra $\rl_E$}.

\begin{remark}\label{re:ideals}{\rm In particular, $\rl$ has many
    ideals, fact established in \cite{Bu14} by different
    techniques. $\rl$ has many other ideals, \eg $\rl(\cj)$ is an
    ideal if $\cj\subset\mbg(\Xi)$ satisfies $E\in\cj, F\supset
    E\Rightarrow F\in\cj$.  If 
    $\cj_k=\{E\in\mbg(\Xi) \mid \dim E\geq k\}$ we get a sequence of
    ideals $\rl_{(k)}\doteq\rl(\cj_k)$ such that $\rl_{(0)}=\rl$ and
    $\rl_{(k)}\supset\rl_{(k+1)}$.  }\end{remark}

\begin{remark}\label{re:zero}{\rm If $E=0$ then $\rl_0=\rl(0)=\C$ and
    $\cp_0\colon\rl\to\C$ is a projection morphism and a trace because
    $\cp_0(\mu\oast\nu)=\cp_0(\mu)\cp_0(\nu)=\cp_0(\nu\oast\mu)$.
  }\end{remark}

Let
$\rl_E^\com=\{\nu\in\rl \mid \nu\oast\mu=\mu\oast\nu\ \forall
\mu\in\rl_E\}$. From Theorem \ref{th:comutant} we get:

\begin{proposition}\label{pr:comut}
  If $\Xi$ is finite dimensional then for any subspace $E\subset\Xi$:
  \begin{equation}\label{eq:comutantL}
    \rl_E = \{\mu\in\rl \mid \mu\oast\delta_\xi=\delta_\xi\oast\mu\
    \forall \xi\in E^\sigma\} \quad\text{and}\quad
    \rl_E^\com=\rl_{E^\sigma}.
  \end{equation}
\end{proposition}

We mention a consequence of Lemma \ref{lm:liminf} which says that, if
$\Xi$ is finite dimensional, \emph{the projection morphism $\cp_E$
  associated to a hyperplane $E$ may be thought as a ``translation at
  infinity'' in a direction $\sigma$-orthogonal to $E$}. In
Proposition \ref{pr:Fxisigma} we prove a similar result in a
representation $\rf$ of $\rl$ for $\Xi$ of any dimension. If
$\xi\in\Xi$ we set $\xi^\sigma\doteq(\R\xi)^\sigma$, hence we have
$\xi\in F^\sigma\Leftrightarrow F\subset\xi^\sigma$.  Any hyperplane
$E\subset\Xi$ is of the form $E=\xi^\sigma$ with
$\xi\in E^\sigma\!\setminus\!\{0\}$.

\begin{proposition}\label{pr:xisigma}
  If $\Xi$ is finite dimensional, $\xi\in\Xi\setminus\!\{0\}$, and
  $E=\xi^\sigma$, then $\forall \mu\in\rl$
  \begin{equation}\label{eq:liminfc}
    \cp_{E}\mu=\lim_{r\to\infty} \delta_{r\xi} \oast \mu \oast \delta_{-r\xi}
    \quad\text{in } \mbs'(\Xi) .
  \end{equation}
\end{proposition}

We now extend this interpretation of $\cp_E$ to the case when $E$ is
not a hyperplane. Note that by Lemma \ref{lm:subsemi}-(3) any
$\mu\in\rl$ belongs to some $\rl(\cs)$ with $\cs$ a countable
subsemilattice.

\begin{proposition}\label{pr:limcsc}
  Assume $\Xi$ finite dimensional and let $E\in\mbg(\Xi)$ and
  $\cs\subset\mbg(\Xi)$ a countable subsemilattice.  Then there is
  $\xi\in E^\sigma$ such that $\xi\nin F^\sigma$ if
  $F\in\cs,F\not\subset E$, and for any such $\xi$ and any
  $\mu\in\rl(\cs)$
  \begin{equation}\label{eq:liminfcs}
    \cp_{E}\mu=\lim_{r\to\infty} \delta_{r\xi} \oast \mu \oast
    \delta_{-r\xi} 
    \quad\text{in } \mbs'(\Xi) .
  \end{equation}
\end{proposition}

\begin{proof}
  If $F\not\subset E$ then $E^\sigma\not\subset F^\sigma$ hence
  $E^\sigma\cap F^\sigma$ is a strict subspace of $E^\sigma$ which
  cannot be a countable union of strict subspaces so $\xi$ exists and
  is not zero. Then it suffices to consider $\mu\in \mathring\rl(\cs)$
  hence $\mu=\sum_{F\in\cs}\mu(F)$ with only a finite number of
  nonzero terms hence $\cp_E\mu=\sum_{F\subset E}\mu(F)$. On the other
  hand
  \begin{align*}
    &\delta_{r\xi} \oast \mu \oast \delta_{-r\xi}
    =\tsum{F\in\cs}{} \,\delta_{r\xi}\oast \mu(F) \oast
      \delta_{-r\xi}\\ &=\tsum{F\subset E}{} \,\mu(F) 
      +\tsum{F\not\subset E}{}\,
      \delta_{r\xi}\oast \mu(F) \oast \delta_{-r\xi}
  \end{align*}
and if $r\to\infty$ then each term in the last sum tends to zero by
Lemma \ref{lm:liminf}. 
\end{proof}

\subsection{Field observables}\label{ss:fops}

Our purpose here is to give a description of the field algebra in
terms of fields thus making the link with the Buchholz-Grundling
approach.  We define the field operators at an abstract level as
observables affiliated to $\rl$ and express the components $\rl(E)$
and the subalgebras $\rl_E$ in terms of the fields.

\subsubsection{Fourier algebra and observables}\label{sss:morphext}

Observables affiliated to a $C^*$-algebra $\rc$ may be described by
using $*$-subalgebras $\ca$ of $C_0(\R)$ such that any morphism
$\ca\to\rc$ extends uniquely to a morphism on $C_0(\R)$. For example,
the algebra of continuous rational functions of degree $<0$ has this
property and gives the description of observables in terms of
self-adjoint resolvents. Here we show that the Fourier algebra allows
one to define observables as generators of unitary groups in the
multiplier algebra of $\rc$.

More generally, if $X$ is a locally compact abelian group, we
construct $X$-valued observables affiliated to $\rc$, \ie morphisms
$\phi:\ca(X)\to\rc$, by using the Fourier algebra of $X$. Let $X^*$ be
the space of characters of $X$ equipped with a Haar measure
$\dd\chi$. We define the Fourier transformation by the condition
$\theta(x)=\int_{X^*} \chi(x)\hat\theta(\chi)\dd\chi$.  The set of
Fourier transforms of integrable functions on $X^*$ is a dense stable
under conjugation and translations subalgebra of $C_0(X)$ which, when
equipped with the norm $\|\theta\|=\|\hat\theta\|_{L^1(X^*)}$, becomes
a Banach $*$-algebra $\ca(X)$ continuously embedded into $C_0(X)$
called \emph{Fourier algebra of $X$}.

\begin{proposition}\label{pr:fourieralg}
  If $\rc$ is a $C^*$-algebra and $\phi:\ca(X)\to\rc$ is a
  morphism then $\phi$ has a unique extension to a morphism
  $C_0(X)\to\rc$.
\end{proposition}

\begin{proof}
  Since any homomorphism from a Banach algebra into a $C^*$-algebra is
  continuous \cite[Prop.\ 4.2]{Dal}, the morphism $\phi:\ca(X)\to\rc$
  is continuous.  We may assume that $\rc$ is commutative. If $Y$ is
  the spectrum of $\rc$ then $Y$ is a locally compact space,
  $\rc\simeq C_0(Y)$, so $\phi:\ca(X)\to C_0(Y)$ is a continuous
  morphism.  If $y\in Y$ then the map $\theta\mapsto \phi(\theta)(y)$
  is clearly a character of $\ca(X)$.  By \cite[Cor.\ 23.7]{HR} the
  spectrum of $\ca(X)$ can be identified with $X$ with the help of the
  evaluation characters $\theta\mapsto\theta(x)$, hence there is a
  unique $x=u(y)\in X$ such that $\phi(\theta)(y)=\theta(u(y))$. Thus
  we get a function $u:Y\to X$ such that $\phi(\theta)=\theta\circ u$
  hence $\|\phi(\theta)\|\leq \sup_x|\theta(x)|=\|\theta\|_{C_0(X)}$.
\end{proof}  

Assume that $\ri$ is an essential ideal of a unital $C^*$-algebra
$\rc$ and equip $\rc$ with the \emph{strict topology} defined by the
family of seminorms $\|S\|_J=\|SJ\|+\|JS\|$ where $J$ runs over $\ri$.
If $U=\{U_t\}_{t\in\R}$ is a one parameter group of unitary elements
of $\rc$ then the continuity of the map $t\mapsto U_t\in\rc$ in the
strict topology is equivalent to the norm continuity of
$t\mapsto U_tJ$ for any $J\in\ri$; if this is satisfied we say that
$U$ is \emph{strictly continuous}. Then, by Proposition
\ref{pr:fourieralg}, we may define an observable $A$ affiliated to
$\rc$ by requiring $\theta(A)=\int\hat\theta(t) U_t\dd t$ if
$\theta\in\ca(\R)$ the integral being taken in the strict topology.
We say that $A$ is the \emph{infinitesimal generator of $U$} and we
write $U_t=\rme^{\rmi tA}$.  We reformulate these remarks as a
proposition and give an alternative proof in terms of double
centralizers.

\begin{proposition}\label{pr:affideal}
  Assume that $\ri$ is an essential ideal of the unital $C^*$-algebra
  $\rc$ and $\{U_t\}_{t\in\R}$ is a group of unitary elements of $\rc$
  such that $t\mapsto U_tJ$ is norm continuous for any $J\in\ri$. Then
  $\{U_t\}_{t\in\R}$ has an infinitesimal generator affiliated to
  $\rc$.
\end{proposition}

\begin{proof}

  For each $\theta\in\ca(\R)$ we define continuous maps
  $L_\theta,R_\theta:\ri\to\ri$ by
  \[
    L_\theta(J) =\int\hat\theta(t) U_tJ\dd t \quad\text{and}\quad
    R_\theta(J) =\int\hat\theta(t) JU_t\dd t.
  \]
  Clearly the pair $(L_\theta,R_\theta)$ is a double centralizer of
  the $C^*$-algebra $\ri$ \cite[p.\ 38]{Mur}, \eg
  $J_1L_\theta(J_2)=R_\theta(J_1)J_2$, hence it is an element of the
  multiplier algebra of $\ri$ which is $\rc$ because $\ri$ is an
  essential ideal of $\rc$.  Thus there is an element
  $\theta(A)\in\rc$, formally equal to $\int\check\theta(t) U_t\dd t$,
  such that $L_\theta(J)=\theta(A)J$ and $R_\theta(J)=J\theta(A)$.  It
  is easy to check that $\theta\mapsto\theta(A)$ is a morphism
  \cite[p.\ 39]{Mur} hence by Proposition \ref{pr:fourieralg} this map
  defines a self-adjoint operator $A$ affiliated to $\rc$.
\end{proof}

\subsubsection{Field observables}\label{sss:fops}

Now we go back to our framework. We need two simple facts.

\begin{lemma}\label{lm:essid}
  If $E\in\mbg_\rms(\Xi)$ then $\rl(E)$ is an essential ideal of
  $\rM(E)$.
\end{lemma}

If $E\in\mbg_\rms(\Xi)$ then the \emph{strict topology on $\rM(E)$} is
that associated to $\rl(E)$.

\begin{lemma}\label{lm:afo}
  If $\xi\in E\in\mbg_\rms(\Xi)$ then the family
  $\{\delta_{t\xi}\}_{t\in\R}$ is a strictly continuous one parameter
  group of unitary elements in $\rM(E)$.
\end{lemma}

Both lemmas are easily proven.  Thanks to Lemma \ref{lm:afo} and
Proposition \ref{pr:affideal} we may now define the fields as
observables affiliated to the Kastler algebra $\rM(\Xi)$.

\begin{definition}\label{df:abstrfop}
  The field observable $\phi(\xi)$ at the point $\xi\in\Xi$ is the
  infinitesimal generator of the one parameter group
  $\{\delta_{t\xi}\}_{t\in\R}$ of unitary elements in $\rM(\Xi)$.
\end{definition}

$\phi(\xi)$ is affiliated to any $C^*$-algebra $\rM(E)$ such that
$\xi\in E\in\mbg_\rms(\Xi)$ and a priori it could depend on $E$, so
should be denoted $\phi_E(\xi)$. But if $E\subset F\in\mbg_\rms(\Xi)$
then clearly $\phi_F(\xi)=\phi_E(\xi)$ and then by \eqref{eq:fil}
$\phi_E(\xi)$ is independent of $E$ so we may denote it $\phi(\xi)$.
This is the observable affiliated to $\rM(\Xi)$ defined by
$\theta(\phi(\xi))=\int_\R \hat\theta(t) \delta_{t\xi} \d t$ for all
$\theta\in\ca(\R)$.

\begin{theorem}\label{th:fieldgen}
{\rm(1)}  The $C^*$-algebra generated by $\phi(\xi)$ is
$C^*(\phi(\xi))=\rl(\R\xi)$. \\[1mm]
{\rm(2)} If $\{\xi_1,\dots,\xi_n\}$ is a generating set for the
subspace $E\in\mbg(\Xi)$ then
\begin{equation} \label{eq:fieldgen1} \rl(E) =
  C^*(\phi(\xi_1))\dot\oast C^*(\phi(\xi_2)) \dot\oast\ldots
  \dot\oast C^*(\phi(\xi_n)) .
  \end{equation}
  {\rm(3)} If $E\subset \Xi$ is an arbitrary subspace then
  \begin{equation} \label{eq:fieldgen2}
    \rl_E =C^*(\phi(\xi)\mid \xi\in E) .
  \end{equation}
\end{theorem}

\begin{proof}
  The assertion (1) is trivial if $\xi=0$, so let $\xi\neq0$. If
  $\theta\in\ca(\R)$ then
  $\theta(\phi(\xi))=\int \hat\theta(t) \delta_{t\xi} \d t$ is a
  measure on the line $\R\xi$ acting as follows:
  $\int f \theta(\phi(\xi))=\int_\R\hat\theta(t)f(t\xi)\d t$ if
  $f\in C_0(\R\xi)$.  Thus $\theta(\phi(\xi))$ is absolutely
  continuous, \ie belongs to $L^1(\R\xi)\subset\rl(\R\xi)$. It is
  clear that the $C^*$-algebra generated by $\phi(\xi)$ is the closure
  in $\rl(\R\xi)$ of the set of elements $\theta(\phi(\xi))$ with
  $\theta\in\ca(\R)$, or this set is clearly dense in $L^1(\R\xi)$
  which in turn is dense in $\rl(\R\xi)$. This finishes the proof of
  (1). Then (1) and (3) of Theorem \ref{th:remarc} imply (2). Finally,
  (3) is an easy consequence of the definition \eqref{eq:deffle}.
\end{proof}

We adopt \eqref{eq:fieldgen2} as definition of $\rl_E$ \emph{for any
  subset} $E\subset\Xi$.  Clearly $ E\subset F$ $\Rightarrow$
$\rl_E\subset \rl_F$. We have $\rl_\emptyset=0$, $\rl_0=\C$ (because
$\phi(0)=0$), and if $E=\{\xi\}$ with $\xi\neq0$ then
$\rl_E=C^*(\phi(\xi))$ which is not unital.  Below we say that two
vectors are collinear if they generate the same subspace; so $0$ is
not colinear with any nonzero vector.

\begin{proposition}\label{pr:gen123}
  Let $E\subset\Xi$ a \emph{finite set} and $\ce$ the set of
  \emph{linear subspaces} generated by the subsets of $E$; then
  $\rl_E=\sum_{F\in\ce}\rl(F)$. If $\xi\in\Xi$ then $\phi(\xi)$ is
  affiliated to $\rl_E$ if and only if $\xi$ is collinear with some
  $\eta\in E$. 
\end{proposition}

\begin{proof}
  By (4) of Theorem \ref{th:remarc} the sum $\sum_{F\in\ce}\rl(F)$
  is a closed linear subspace of $\rl$. If $F',F''\in\ce$ are
  generated by the subsets $f',f''$ of $E$ then $F'+F''$ is generated
  by the subset $f'\cup f''$ hence belongs to $\ce$ and
  $\sum_{F\in\ce}\rl(F)$ is a $C^*$-algebra by (3) of Theorem
  \ref{th:remarc}.  $\rl_E$ is the $C^*$-algebra generated by the
  operators $\phi(\xi)$ with $\xi\in E$ and if
  $\xi_1,\dots,\xi_n\in E$ generate the subspace $F$ and
  $u_1,\dots,u_n\in C_0(\R)$ then
  $u_1(\phi(\xi_1))\oast\dots\oast u_n(\phi(\xi_n))$ belongs to
  $\rl(F)$ by \eqref{eq:fieldgen1} hence
  $\rl_E\subset\sum_{F\in\ce}\rl(F)$. We have equality here because
  $\rl(F)\subset\rl_E$ by the same argument.  The last assertion
  follows from Theorem \ref{th:remarc}-(1).
\end{proof}

\begin{remark}\label{re:naff}
  {\rm In particular, if $\xi,\eta\in\Xi$ are linearly independent and
    $\zeta\in\Xi$ is not collinear with any of them, then the operator
    $\phi(\zeta)$ is not affiliated to
    $C^*(\phi(\xi),\phi(\eta))$.
    Note also that if $\dim E>1$ then
    $\rl_E$ is not generated by a finite set of operators $\phi(\xi)$
    hence the definitions (1.3) and (3.32) in \cite{GI19} are
    wrong, but this does not play any role in later arguments there,
    it suffices to change the quoted relations in
    $\Phi_E=C^*(\phi(\xi)\mid \xi\in E)$. }
\end{remark}

\subsection{Field algebra in a representation} \label{ss:fcaw}

Here we describe Hilbert space realizations of the abstract field
algebra $\rl$ obtained via representations $W$ of $\Xi$. We define
them by transporting $\rl$ with the help of the $C^*$-isomorphism
$W\colon\rM(\Xi)\to\rwk(\Xi)$.

In Definition \ref{df:abstrfop} we introduced the ``abstract'' field
operator $\phi(\xi)$ as an observable affiliated to $\rM(\Xi)$. On the
other hand, in \S\ref{ss:i1} we defined the field operator
$\phi_\ssw(\xi)$ as a self-adjoint operator on the Hilbert space
$\ch$. These two operators are related by the algebraic representation
$W$ of Proposition \ref{pr:wmorph}: indeed, if we take
$\mu=\delta_{t\xi}$ in \eqref{eq:wme} we get
$W(\delta_{t\xi})=W(t\xi)$ which implies
$W(\phi(\xi))=\phi_\ssw(\xi)$. Thus the image of the ``abstract''
field operator $\phi(\xi)$ through the $C^*$-algebra representation
$W$ is the field operator $\phi_\ssw(\xi)$. If there is no risk of
confusion, we identify these operators
$\phi_\ssw(\xi)\equiv\phi(\xi)$.

In the next definition we use Proposition \ref{pr:wmorph} and Theorem
\ref{th:fieldgen} which gives a description \`a la Buchholz-Grundling
of the field algebra and its components in the representation $W$.

\begin{definition}\label{df:wkfield}
  Let $W$ be a representation of $\Xi$ on the Hilbert space
  $\ch$. \\[1mm]
  {\rm(1)} The \emph{field $C^*$-algebra $\rf$ 
    of $\Xi$ in the representation $W$} is the norm closure of the set
  of operators $W(\mu)$ with $\mu\in L^1(E)$ for some $E\in\mbg(\Xi)$.
  We have $\rf= C^*(\phi(\xi)\mid \xi\in\Xi)$.\\[1mm]
  {\rm(2)} If $E\subset\Xi$ is any linear subspace then
  $\rf_\sse$ 
  is the norm closure of the set of operators $W(\mu)$ with
  $\mu\in L^1(F)$ for some $F\in\mbg(E)$. We have
  $\rf_\sse =C^*(\phi(\xi)\mid \xi\in E)$. \\[1mm]
  {\rm(3)} If $E\in\mbg(\Xi)$ then the norm closure of the set of
  operators $W(\mu)$ with $\mu\in L^1(E)$ is a $C^*$-algebra
  $\rf(E)$.
  If $\{\xi_1,\dots,\xi_n\}$ is a generating set for $E$ then
  \begin{equation} \label{eq:b-g-e2}
    \rf(E)=C^*(\phi(\xi_1))\cdot C^*(\phi(\xi_2)) \cdot\ldots\cdot
    C^*(\phi(\xi_n)) .
  \end{equation}
\end{definition}

\eqref{eq:b-g-e2} follows from Theorem \ref{th:fieldgen} and implies
the last assertions of (1) and (2). As in the purely algebraic
framework we call the algebras $\rf(E)$ \emph{components} of $\rf$.
Clearly
\begin{equation}\label{eq:rfE}
  \rf_\sse=\tsum{F\subset E}{\rmc} \rf(F) \quad
  \text{for any linear subspace } E\subset\Xi. 
\end{equation}

\begin{remark}\label{re:just}{\rm The spaces $\rf,\rf_\sse,\rf(E)$
    from Definition \ref{df:wkfield} depend of $\Xi$ and $W$ but in
    general we do not specify this explicitly unless this is necessary
    and then we use the notations
    $\rwf^\ssxi, \rwf^\ssxi_\sse, \rwf^\ssxi(E)$ or just
    $\rf^\ssxi, \rf^\ssxi_\sse, \rf^\ssxi(E)$. The last three simpler
    notations are justified by Proposition \ref{pr:independence}.  If
    $\Xi$ is clear from the context we set $\rf=\rf^\ssxi$,
    etc.}\end{remark}

\begin{remark}\label{re:sympsubsp}

  {\rm The restriction of $W$ to a symplectic subspace $E\subset\Xi$
    is a representation of $E$ on $\ch$ still denoted $W$. If $F$ is a
    symplectic subspace and $E\subset F$ then
    $\rwf^\sse\subset \rwf^\ssf$ and \eqref{eq:fil} implies
    $\rwf^\ssxi=\overline{\cup}_{E\in\mbg_s(\Xi)}\rwf^\sse$, where
    $\overline{\cup}$ denotes the closure of the union.}
\end{remark}

Below, the relation $\Phi(\phi_1(\xi))=\phi_2(\xi)$ means
$\Phi\big(u(\phi_1(\xi))\big)=u(\phi_2(\xi))$ $\forall u\in C_0(\R)$
and is equivalent to $\Phi(R_1(\xi))=R_2(\xi)$ with
$R_k(\xi)=(\phi_k(\xi)-\rmi)^{-1}$.

\begin{proposition}\label{pr:independence}
  The algebras $\prescript{\ssw_{\!1}\!\!}{}\rf^\ssxi$ and
  $\prescript{\ssw_{\!2}\!\!}{}\rf^\ssxi$ associated to two
  representations $W_1,W_2$ of $\Xi$ are canonically isomorphic in the
  following sense: if $\phi_k(\xi)$ are the field operators in the
  representation $W_k$, then there is a unique morphism
  $\Phi\colon\prescript{\ssw_{\!1}\!\!}{}\rf^\ssxi \to
  \prescript{\ssw_{\!2}\!\!}{}\rf^\ssxi$ such that
  $\Phi(\phi_1(\xi))=\phi_2(\xi)$ for all $\xi\in\Xi$; this morphism
  is an isomorphism.
\end{proposition}

\begin{proof}
  Assume first $\Xi$ finite dimensional and let $W_0:\Xi\to U(\ch_0)$
  an irreducible representation and $W:\Xi\to U(\ch)$ an arbitrary
  representation. By the Stone-Von Neumann theorem and with the
  notations of Remark \ref{re:stvon} we have
\[
  Vu_1(\phi(\xi_1))\dots u_n(\phi(\xi_n))V^*
  =\big(u_1(\phi_0(\xi_1))\dots u_n(\phi_0(\xi_n))\big)\otimes1_\ck
\]
for any $\xi_1,\dots,\xi_n\in\Xi$ and $u_1,\dots,u_n\in C_0(\R)$.  The
map $A\mapsto V^{-1}(A\otimes1_\ck)V$ is a morphism
$B(\ch_0)\to B(\ch)$ and, by the preceding relation, its restriction
to $\prescript{\ssw_{\!0}\!\!}{}\rf^\ssxi$ is an isomorphism
$\Phi\colon \prescript{\ssw_{\!0}\!\!}{}\rf^\ssxi\to \rwf^\ssxi$ which
satisfies $\Phi (\phi_0(\xi))=\phi(\xi)$ for all $\xi\in\Xi$ and is
uniquely determined by this relation.  If $\Xi$ is infinite
dimensional for each $E\in\mbg_s(\Xi)$ we have a canonical isomorphism
$\Phi_E\colon\prescript{\ssw_{\!1}\!\!}{}\rf^\sse \to
\prescript{\ssw_{\!2}\!\!}{}\rf^\sse$ and by uniqueness
$\Phi_F|\rf_E^{\ssw_1}=\Phi_E$ if $E,F\in\mbg_s(\Xi)$ and
$E\subset F$. Then we use Remark \ref{re:sympsubsp}.
\end{proof}

\begin{remark}\label{re:isonot}

  {\rm If $\dim\Xi<\infty$ then by the Stone-Von Neumann theorem
    $\rwf^\ssxi\cong \prescript{\ssw_{\!0}\!\!}{}\rf^\ssxi\otimes
    1_{\ck}$ and $\rwf^\ssxi(\Xi)\cong K(\ch_0)\otimes 1_{\ck}$
    with the notations of Remark \ref{re:stvon}.

  }\end{remark}

In the next proposition we describe some simple properties of the
algebras $\rf(E)$.

\begin{proposition}\label{pr:properties}
  Let $E,F$ be finite dimensional subspaces of $\Xi$.
\begin{compactenum}
\item[{\rm(a)}] $\rf(E)$ is a non-degenerate $C^*$-subalgebra of
  $B(\ch)$.
\item[{\rm(b)}] $\rf(0)=\C$ and this is the only $\rf(E)$ which
  is unital.
\item[{\rm(c)}] $\rf(E)$ is abelian if and only if $E$ is
  isotropic and then $\rf(E)\cong C_0(E^*)$.  
\item[{\rm(d)}] If $S\in\rf(E)$ and $\xi\in E^\sigma$ then
  $SW(\xi)=W(\xi)S$.
\item[{\rm(e)}] If $S\in\rf(E)$ and $T\in\rf_{E^\sigma}$ then $ST=TS$.
\end{compactenum}
\end{proposition}

\begin{proof}
  Let $\lambda_E$ be a Lebesgue measure on $E$ and $\rho$ an
  integrable function on $E$ with $\int_{E}\rho\lambda_E=1$. Then for
  $\xi\in E$ and $\veps >0$ consider the function
  $\rho^\varepsilon_\xi(\eta) =
  \varepsilon^{-n}\rho((\eta-\xi)/\veps)$, where $n=\dim E$, and let
  $\mu^\varepsilon_\xi(\dd\eta) =
  \rho^\varepsilon_\xi(\eta)\lambda_E(\dd\eta)$. Then we have
    \begin{equation*}
    W(\mu^\varepsilon_\xi)=
    \int_E W(\eta)
    \varepsilon^{-n}\rho((\eta-\xi)/\veps)\lambda_E(\dd\eta) = 
    \int_E W(\xi+\varepsilon\eta) \rho(\eta)\lambda_E(\dd\eta)
  \end{equation*}
  hence
  \begin{equation}\label{eq:approx}
    \slim_{\varepsilon\to0}W(\mu^\varepsilon_\xi)=W(\xi)
  \end{equation}
  For example $\slim_{\varepsilon\to0}W(\mu_0^\varepsilon)=1$, which
  implies property (a). The first assertion of (b) is obvious and the
  second follows from $\rf(E)\cap\rf(F)=0$ if $E\neq F$. If $\rf(E)$
  is abelian then $W(\mu)W(\nu)=W(\nu)W(\mu)$ for all
  $\mu,\nu\in L^1(E)$. By taking $\mu=\mu^\varepsilon_\xi$ and
  $\nu=\mu^\varepsilon_\eta$ with $\xi,\eta\in E$ and making
  $\varepsilon\to0$ we get $W(\xi)W(\eta)=W(\eta)W(\xi)$ and thus
  $(\e^{-\rmi\sigma(\xi,\eta)}-1)W(\eta)W(\xi)=0$ hence
  $\e^{-\rmi\sigma(\xi,\eta)}=1$ for all $\xi,\eta\in E$, so $E$ is
  isotropic. The converse assertion is obvious and the isomorphism
  with $C_0(E^*)$ is discussed in Theorem \ref{th:structure}. The
  assertions (d) and (e) are obvious. 
\end{proof}

Clearly Theorem \ref{th:remarc} implies Theorem \ref{th:Fgrad} and
Theorem \ref{th:idec} is a consequence of Theorem \ref{th:gmain} and
Proposition \ref{pr:ecapf}.
We now show that certain projections $\cp_\sse$ may be thought as
``translations at infinity''.

\begin{lemma}\label{lm:liminfinity}
  If $F\in\mbg(\Xi)$, $\xi\nin F^\sigma$, and $T\in\rf(F)$ then
  $\slim_{r\to\infty} T W(r\xi)=0$
\end{lemma}

\begin{proof}
  We have to prove $\lim_{r\to\infty} \|T W(r\xi)f\|=0$
  $\forall f\in\ch$.  We may assume $\Xi$ finite dimensional, if not
  we replace it by a finite dimensional symplectic subspace of $\Xi$
  containing $F$ and $\xi$.  Then it suffices to consider $T=W(\mu)$
  with $\mu\in L^1(F)$.  If $\nu=\mu^*\oast\mu$ then
  \[
    \|W(\mu)W(r\xi)f\|^2 =
    \braket{f}{W(r\xi)^*W(\mu)^*W(\mu)W(r\xi)f}
    = \braket{f}{W(r\xi)^*W(\nu)W(r\xi)f}
\]
and from \eqref{eq:products} we have
\[
  W(r\xi)^*W(\nu)W(r\xi)=W(-r\xi)W(\nu)W(r\xi)
  =W(\delta_{-r\xi}\oast\nu\oast\delta_{r\xi})
  =W\big(\e^{\rmi\sigma(r\xi,\cdot)}\nu \big)
\]
with $\nu\in L^1(E)$, hence
\[
  \|W(\mu)W(\xi)f\|^2
  = \braket{f}{W\big(\e^{\rmi\sigma(r\xi,\cdot)}\nu \big)f}
  =\int_E \e^{\rmi r\sigma(\xi,\eta)} \braket{f}{W(\eta)f} \nu(\dd\eta) .
\]
The function $\theta\doteq\sigma(\xi,\cdot)|F$ is a nonzero linear
form on $F$ and the integral above is of the form
$\int_F\e^{\rmi r\theta(\eta)} u(\eta)\dd\eta$ with $u\in L^1(F)$ and
if $r\to\infty$ such an integral tends to zero by the Riemann-Lebesgue
lemma, which finishes the proof
\end{proof}

\begin{proposition}\label{pr:Fxisigma}
  Let $\xi\in\Xi\setminus\{0\}$ and $E=\xi^\sigma$. Then for any
  $T\in\rf$
  \begin{equation}\label{eq:Fliminfc}
    \cp_\sse T=\slim_{r\to\infty} W(r\xi)^*T W(r\xi) .
  \end{equation}
\end{proposition}

\begin{proof}
  $\cp_\sse$ is continuous so it suffices to prove
  \eqref{eq:Fliminfc} for $T\in\mathring\rf$ and for this it suffices
  to consider the case $T\in\rf(F)$ for some $F\in\mbg(\Xi)$. By (d)
  of Proposition \ref{pr:properties} for $F\subset E$ we have
  $TW(r\xi)=W(r\xi)T$ $\forall r\in\R$ hence
  $\slim_{r\to\infty} W(r\xi)^*T W(r\xi)=S$. If
  $F\not\subset E=\xi^\sigma$ then $\xi\nin F^\sigma$ hence from Lemma
  \ref{lm:liminfinity} we get
  $\slim_{r\to\infty} W(r\xi)^*T W(r\xi)=0$.
\end{proof}

The following proposition is proved by an argument similar to that of
Proposition \ref{pr:limcsc}; the last assertion is a consequence of
Lemma \ref{lm:subsemi}-(3).

\begin{proposition}\label{pr:limcsc-b}
  Let $\Xi$ be finite dimensional, $E\in\mbg(\Xi)$, and
  $\cs\subset\mbg(\Xi)$ a countable subsemilattice.  Then 
  $\exists\,\xi\in E^\sigma$ such that $\xi\nin F^\sigma$ if
  $F\in\cs,F\not\subset E$, and for any such $\xi$
  \begin{equation}\label{eq:liminfcs-b}
    \cp_\sse T=\slim_{r\to\infty} W^*(r\xi) T W(r\xi) \quad
    \forall T\in\rf(\cs) .
  \end{equation}
  In particular, for any $T\in\rf$ and $E\in\mbg(\Xi)$ there is
  $\xi\in E^\sigma$ such that
    \begin{equation}\label{eq:liminfcs-c}
    \cp_\sse T=\slim_{r\to\infty} W^*(r\xi) T W(r\xi) .
  \end{equation}
\end{proposition}

The next result improves for finite dimensional $\Xi$ the statements
(d) and (e) of Proposition \ref{pr:properties}. If $\ca\subset\rf$
then $\ca^\com\doteq\{A\in\rf \mid [A,B]=0\ \forall B\in\ca\}$ is its
commutant in $\rf$.

\begin{theorem}\label{th:comutant}
  If $\Xi$ is finite dimensional then for any subspace $E\subset\Xi$
  we have
\begin{equation}\label{eq:comutant}
  \rf_{\ss{E^\sigma}}= \{T\in\rf \mid [T,W(\xi)]=0\ \forall \xi\in E\}
  = \{T\in\rf \mid [T,S]=0\ \forall S\in \rf_\sse\}.
\end{equation}
In particular $\rf_\sse^\com\equiv(\rf_\sse)^\com=\rf_{\ss{E^\sigma}}$
and
\begin{equation}\label{eq:comutant2}
  \rf_\sse= \{T\in\rf \mid [T,W(\xi)]=0\ \forall \xi\in E^\sigma\}.
\end{equation}
\end{theorem}

\begin{proof}
  If $T\in\rf_\sse$ then $[T,W(\xi)]=0\ \forall \xi\in E^\sigma$ by
  Proposition \ref{pr:properties}-(d). Reciprocally, assume that $T$
  has the last property. Then by the last assertion of Proposition
  \ref{pr:limcsc-b} we get $\cp_\sse T=T$ which is equivalent to
  $T\in\rf_\sse$. This proves the first equality in
  \eqref{eq:comutant}.  To prove the second, note that by Definition
  \ref{df:wkfield}-(2) $\rf_\sse$ is the $C^*$-algebras generated by
  the $u(\phi(\xi))$ with $u\in C_0(\R)$ and $\xi\in E$ hence
  $T\in\rf_\sse^\com$ if and only if $[T,u(\phi(\xi))]=0$ for all such
  $u$ and $\xi$ and this in turn is equivalent to
  $[T,\e^{\rmi\phi(\xi)}]=0$ for all $\xi\in E$.
\end{proof}

\begin{corollary}\label{co:comutant}
  If $X$ is a Lagrangian subspace of $\Xi$ then $\rf_\ssx^\com=\rf_\ssx$.
\end{corollary}

An explicit description of the algebra $\rf_X$ is given in Proposition
\ref{pr:rfyz}.

We mention a fact which allows one to consider crossed products
of graded $C^*$-subalgebras $\rf(\cs)$ by the action of finite
dimensional subspaces of $\Xi$.

\begin{proposition}\label{pr:normcont} 
  {\rm(1)} If $\cs\subset\mbg(\Xi)$ then
  $W(\xi)^*\rf(\cs)W(\xi)=\rf(\cs)$  $\forall\xi\in\Xi$. \\[1mm]
  {\rm(2)} $\forall\xi\in\Xi$ the map $T\mapsto W^*(\xi)TW(\xi)$ is an
  automorphism of $\rf$ and for any $T\in\rf$ the map
  $\xi\mapsto W^*(\xi)TW(\xi)\in\rf$ is norm continuous on finite
  dimensional subspaces of $\Xi$.
\end{proposition}

\begin{proof}
  We have $W(\xi)^*W(\eta)W(\xi)= \e^{\rmi\sigma(\xi,\eta)}W(\eta)$
  hence if $E\in\mbg(\Xi)$ and $\mu\in L^1(E)$
\begin{equation*}
  W(\xi)^* W(\mu) W(\xi) = \int_E W(\xi)^*W(\eta)W(\xi)
  \mu(\eta)\d_E\eta
  = \int_E W(\eta) \e^{\rmi\sigma(\xi,\eta)}\mu(\eta)\d_E\eta
\end{equation*}
which clearly implies (1) from which it follows that
$W(\xi)^*\rf W(\xi)=\rf$ hence the first part of (2).  Then the
relation above implies
\[
  \|W(\xi)^* W(\mu) W(\xi)-W(\zeta)^* W(\mu) W(\zeta)\|\leq
  \int_E
  \left|\e^{\rmi\sigma(\xi,\eta)}-\e^{\rmi\sigma(\zeta,\eta)}\right
  | |\mu(\eta)|\d_E\eta
\]
hence $\xi\mapsto W(\xi)^* W(\mu) W(\xi)\in B(\ch)$ is norm continuous
on finite dimensional subspaces of $\Xi$ which implies the last part
of (2) by a density argument. 
\end{proof}

\section{Finite dimensional symplectic spaces }\label{s:fdhvz}
\protect\setcounter{equation}{0}
\renewcommand\thesubsubsection{\thesubsection.\arabic{subsubsection}}

\subsection{Intrinsic description of
  \texorpdfstring{$\rf(E)$}{F(E)}}\label{ss:idc}

This subsection is devoted to the proof Theorem \ref{th:mainthfd}.  We
use the first characterization of $\rf(E)$ from Definition
\ref{df:wkfield}, so $\rf(E)$ is the norm closure in $B(\ch)$ of the
set of operators $W(f)=\int_E W(\xi) f(\xi)\d\xi$ with $f\in L^1(E)$.

We denote $\re(E)$ the set of operators $T$ satisfying the three
conditions \eqref{eq:maineq}. Clearly $\re(E)$ is a $C^*$-algebra and
we have to prove $\rf(E)=\re(E)$.  This is easy in two cases. First,
if $E={ 0}$, then $E^\sigma =\Xi$ hence $\re(0) =\C$ because $W$ is
irreducible, so this case is trivial. The second case is $E=\Xi$: then
$E^\sigma ={ 0}$ so $\re(\Xi)=K(\ch)=\rf(\Xi)$ by the Kolmogorov-Riesz
and Stone-Von Neumann theorems (or see Theorem \ref{th:RKR}).

It is easy to show that $\rf(E)\subset\re(E)$ (\conf Proposition
\ref{pr:tatar}) so Theorem \ref{th:mainthfd} is equivalent to
\begin{equation}\label{eq:re:2}
  \text{the set of operators } W(\mu) \text{ with } \mu \in L^1(E)
  \text{ is dense in }\re(E).
\end{equation}
The rest of this section is devoted to the proof of this fact for
$E\neq0,\Xi$.

\ssubsection{}\label{sss:312}
\addtocontents{toc}{\SkipTocEntry}

For this we need to go beyond measures and define the operators
$W(\mu)$ for $\mu$ temperate distributions on $\Xi$. This Weyl
pseudo-differential calculus requires some supplementary formalism, we
refer to \cite{Fol} for details.

Let $\mbs(\Xi)$ be the space of Schwartz test functions on $\Xi$ and
$\mbs'(\Xi)$ its dual, the space of tempered distributions.  If
$\mu\in\mbs'(\Xi)$ and $f\in\mbs(\Xi)$ the value $\mu(f)$ of $\mu$ at
$f$ is denoted $\bracet{f}{\mu}$ but we also write this as
$\int_\Xi f(\xi)\mu(\xi)\dd\xi$, which is often a convenient abuse of
notation. We set $\braket{f}{\mu}=\bracet{\bar{f}}{\mu}$.  We have
continuous linear embeddings
$\mbs(\Xi)\subset L^2(\Xi)\subset\mbs'(\Xi)$. By \eqref{eq:distrib} we
also have $\rf\subset\rk\subset\mbs'(\Xi)$.

Although the twisted convolution product defined on $M(\Xi)$ cannot be
extended to $\mbs'(\Xi)$, it is possible to define $\mu\oast\nu$ for
$\mu,\nu \in \mbs'(\Xi)$ if one of them has compact support (as in the
case of the usual convolution).  And the relations \eqref{eq:products}
remain valid for $\mu\in\mbs'(\Xi)$.

The map $\xi\mapsto\xi_\sigma=\sigma(\cdot,\xi)$ is a linear
isomorphism $\Xi\to \Xi^*$ hence the Fourier transform of a
distribution on $\Xi$ is naturally identified with a distribution on
$\Xi$. More precisely, $\Xi$ is equipped with the symplectic measure
$\d\xi$ and the Fourier measure $\dbar\xi=(2\pi)^{-n}\d\xi$; then the
\emph{symplectic Fourier transform} of $\mu\in L^1(\Xi)$ is
\begin{equation}\label{eq:FTS}
  (\cf_{\!\sigma}\mu)(\xi)\equiv\what\mu(\xi) \doteq
  \int_\Xi \e^{\rmi\sigma(\xi,\eta)}\mu(\eta)\dbar\eta .
\end{equation}
The map $\cf_{\!\sigma}:\mbs(\Xi)\to\mbs(\Xi)$ satisfies
$\cf_{\!\sigma}^2=1$ hence is a homeomorphism and extends to a
self-adjoint unitary operator in $L^2(\Xi)$ and to a homeomorphism in
$\mbs'(\Xi)$.

For any linear subspace $E\subset\Xi$ we choose and fix a Lebesgue
measure $\lambda_E$ or $\d_E\xi$ on $E$.  We think of $\lambda_E$ as a
distribution on $\Xi$, namely
\[
  \bracet{f}{\lambda_E}=\int_E f\lambda_E=\int_E f(\xi)\d_E\xi.
\]  
According to the conventions from \S\ref{sss:noterm} we embed
isometrically $L^1(E)\subset M(E)$ hence if $\mu$ is an absolutely
continuous measure on $E$ of the form $\mu(\d\xi)=f(\xi)\d_E\xi$ then
we may use both notations $W(\mu)$ and $W(f)$.

We could continue to work with an arbitrary irreducible representation
of $\Xi$ but it is simpler, and there is no loss of generality, if we
assume that $W$ is the Schr\"odinger representation associated to a
Lagrangian decomposition $\Xi=X\oplus X^*$ (we use the notations of
\S\ref{ss:HLdec}).  Clearly then if $f,g\in\mbs(X)$ then the map
$\xi\mapsto \braket{f}{W(\xi)g}$ belongs to ${\mbs}(\Xi)$ so if
$\mu \in \mbs'(\Xi)$ we can define $W(\mu)=\int W\mu$ as a continuous
linear operator $\mbs(X)\rarrow\mbs'(X)$ by setting
\begin{equation}\label{eq:Wmu}
  \braket{f}{W(\mu)g}
  = \int_{\Xi}\braket{f}{W(\xi)g}\mu(\xi)\dd\xi \quad \forall
  f,g\in\mbs(X).  
\end{equation}
We get a map ${\mbs}'(\Xi)\ni\mu\mapsto W(\mu)\in B(\mbs(X),\mbs'(X))$
which is bijective \cite[Th.\ 1.30]{Fol}.  We denote $T\mapsto T^{\#}$
its inverse map, hence $T=W(\mu)\Leftrightarrow\mu = T^{\#}$.  Since
$\ch=L^2(X)$ we have $B(\ch)\subset B(\mbs(X),\mbs'(X))$ hence to
each $T\in B(\ch)$ is associated $T^{\#} \in {\mbs}'(\Xi)$.

We mention that if $\nu\in\mbs'(\Xi)$ is compactly supported then
$W(\nu)$ leaves $\mbs(X)$ invariant and extends to a continuous
operator on $\mbs'(X)$. And then the relation
$W(\mu\oast\nu)=W(\mu)W(\nu)$ remains valid for any
$\mu\in\mbs'(\Xi)$.

\ssubsection{}\label{sss:313}
\addtocontents{toc}{\SkipTocEntry}

We now start the proof of \eqref{eq:re:2}. First note the following
fact.

\begin{lemma}\label{l:benji}
  If $T\in B(\ch)$ and $[W(\xi),T] = 0$ for all $\xi\in E^\sigma$ then
  $\supp T^{\#} \subset E$.
\end{lemma}

\begin{proof}
  Denote $T^{\#}\equiv\mu\in\mbs'(\Xi)$. Then $T=W(\mu)$ and by
  \eqref{eq:products} we have
\begin{align*}
T &=  W(\xi)TW(\xi)^*= W(\xi)W(\mu)W(\xi)^* =W(\xi)W(\mu)W(-\xi)\\
& = W(\delta_\xi)W(\mu)W(\delta_{-\xi})
= W(\delta_\xi\oast\mu\oast\delta_{-\xi})
= W(\e^{\rmi\xi_\sigma}\mu).
\end{align*}
Since $W:{\mbs}'(\Xi)\rarrow B({\mbs}(X),{\mbs}'(X))$ is bijective,
we get $(\e^{\rmi\xi_\sigma} -1)\mu$ $ = 0$ $\forall\xi \in E^\sigma$.
Let now $\eta^\circ\nin E$. Then there is a $\xi^\circ \in E^\sigma$
such that $\e^{\rmi\sigma(\eta^\circ,\xi^\circ)} \neq 1$. So there is
a compact neighborhood $V$ of $\eta^\circ$ such that
$\e^{\rmi\sigma(\eta,\xi^\circ)}-1 \neq 0$ for all $\eta \in V$. If
$\theta \in C_{\rmc}^\infty(V)$ then
$\varphi\doteq\theta/(\e^{\rmi\xi^\circ_\sigma}-1)\in
C_{\rmc}^\infty(V)$ so by the previous computation 
$\theta\mu=\vphi(\e^{\rmi\xi^\circ_\sigma}-1)\mu=0$.  It follows that 
$\eta^\circ\nin\supp \mu$.
\end{proof}

We construct now a regularization of $T\in B(\ch)$.
Let $\theta \in L^1(\Xi)$ with $\int_\Xi \theta(\xi)\dbar \xi=1$,
denote $\theta_\veps(\xi)=\veps^{-2n}\theta(\xi/\veps)$ with
$2n=\dim\Xi$, and define for $\veps>0$
\begin{equation}                     \label{e:Teps}
T_\varepsilon\doteq\int_\Xi W(\xi)\,T\,W(\xi)^*\,\theta_\veps(\xi)
\dbar\xi =\int_\Xi W(\veps\xi)\,T\,W(\veps\xi)^*\,\theta(\xi)\dbar\xi.
\end{equation}

\begin{lemma}\label{l:Teps}
{\rm(i)}
If\; $\lim_{\xi\rarrow0}\|\,[W(\xi),T]\,\|=0$ then
$\lim_{\veps\rarrow 0}\|T_\veps -T\|=0$.\\[1mm]
{\rm (ii)}
If\; $T\in {\re}(E)$ then $T_\veps\in {\re}(E)$ for all
$\veps>0$.\\[1mm]
{\rm (iii)} If\; $\theta\in\mbs(\Xi)$, then 
$T^{\#}_\varepsilon\equiv
(T_\varepsilon)^{\#}=\what{\theta}_{\!\veps}\, T^{\#}$. 
\end{lemma} 

\begin{proof}
Clearly
\[
\|T_\veps -T\| \leq
\int_\Xi \|W(\veps\xi)\,T\,W(\veps\xi)^*-T\|\,|\theta(\xi)|\dbar\xi
\]
which tends to zero as $\veps\rarrow 0$ under the condition of (i).
For (ii), notice first that $T_\veps$ satisfies (i) of Definition
\ref{df:intrinsic} (for all $T\in B(\ch)$).  Moreover, since
$T\in {\re}(E)$, for all $\eta \in E^\sigma$
\begin{eqnarray*}
W(\eta)W(\xi)\,T\,W(\xi)^*
& = &
\e^{i\sigma(\xi,\eta)}W(\xi)\,T\,W(\eta)W(-\xi)\\
& = &
\e^{i(\sigma(\eta,\xi)+\sigma(\eta,-\xi))}W(\xi)\,T\,W(\xi)^*W(\eta)
\end{eqnarray*}
thus (ii) of Definition \ref{df:intrinsic} is valid for $T_\veps$
too. Finally, if $\eta \in E$ then
\begin{align*}
\|(W(\eta)-1)T_\veps\|
& \leq
\int_\Xi |1-\e^{\rmi\sigma(\eta,\xi)}|\,\|T\|\,|\theta_\veps(\xi)|\dbar\xi \\
& + \|(W(\eta)-1)T\|\int_\Xi |\theta_\veps(\xi)|\dbar\xi
\end{align*}
which tends to zero if $\eta\rarrow 0$. Thus $T_\veps\in {\re}(E)$.
 
Assume now $\theta \in {\mbs}(\Xi)$ and let us compute
$T^{\#}_\veps$. If $T^\#=\mu$ and $f,g \in\mbs(X)$, since the
operators $W(\xi)$ leave $\mbs(X)$ invariant, one has (as in the
proof of Lemma \ref{l:benji}):
\begin{align*}
\braket{f}{T_\veps g} &=
\int_\Xi \braket{f}{W(\xi)W(\mu)W(\xi)^*g} \theta_\varepsilon(\xi)\dbar\xi 
= \int_\Xi \braket{f}{W(\e^{\rmi\xi_\sigma}\mu)g}
\theta_\varepsilon(\xi)\dbar\xi\\
&= \left\langle f \Big| W\left(\int_\Xi
\e^{\rmi\xi_\sigma}\theta_\veps(\xi)\dbar\xi \, \mu\right)g\right\rangle
=\braket{f}{W\big(\what{\theta}_{\!\veps}\mu \big)g} .   
\end{align*}
For the proof of the third equality above, note that it suffices to
show it for a set of $\mu$ which is dense in ${\mbs}'(\Xi)$, e.g.\
for $\mu\in {\mbs}(\Xi)$; in this case the equality is easy to
justify.  Anyway, since $W$ is bijective we get
$T^{\#}_\veps=\what{\theta}_{\!\veps}\mu$.
\end{proof}

Let us choose $\theta$ such that
$\what{\theta}\in C_{\rmc}^\infty(\Xi)$. Then if $T\in\re(E)$, due to
Lemma \ref{l:benji} and (ii) and (iii) of Lemma \ref{l:Teps}, the
distribution $T^{\#}_\varepsilon$ has compact support included in $E$.
And by (i) of Lemma \ref{l:Teps} we have
$\lim_{\varepsilon\to0}T_\varepsilon=T$ in norm. Thus for the proof of
\eqref{eq:re:2} it suffices to show that any operator of the form
$T_\varepsilon$ is norm limit of operators $W(\mu)$ with
$\mu\in L^1(E)$.

\ssubsection{}\label{sss:314}
\addtocontents{toc}{\SkipTocEntry}

Thus from now on we assume that $T\in {\re}(E)$ and $T^{\#}$ is a
distribution on $\Xi$ whose support is a compact subset of $E$. But we
may simplify the problem still further as follows.

Let $\rho\in C^\infty_{\rmc}(E)$ with $\rho\geq0$ and
$\int_{E}\rho(\eta)\d_E\eta=1$; for $\veps >0$ set
$\rho_\veps (\eta) = \veps^{-m}\rho(\eta/\veps)$ where $m=\dim
E$. Then $W(\rho_\veps)\in {\re}(E)$ by Proposition \ref{pr:tatar}
hence $W(\rho_\veps)T\in {\re}(E)$, and it is easy to check that $T$
is the norm limit of the operators $W(\rho_\veps)T$ (Lemma
\ref{lm:surya}-(ii)). Thus it suffices to prove that such products are
norm limit of operators $W(\mu)$ with $\mu\in L^1(E)$.

Let $\rho\in C^\infty_{\rmc}(E)$ identified with the measure
$\rho\lambda_E$. Then
\[
  W(\rho)T=W(\rho)W(T^{\#})= W(\rho\oast T^{\#})
\]
where $\rho\oast T^{\#}$ makes sense because both $\rho$ and $T^{\#}$
are compactly supported distributions.  Our purpose is to show
\begin{equation}\label{e:mut}
  \rho\oast T^{\#} = u\lambda_{_{E}} \quad
  \text{for some } u\in C^\infty_{\rmc}(\Xi) 
\end{equation}
and this clearly finishes the proof of the theorem.

We begin by proving that the twisted convolution $\rho\oast T^{\#}$
also has compact support contained in $E$, and we also deduce precise
information about its structure.  Let $E'$ be a vector subspace of
$\Xi$, supplementary to $E$. Since $T^{\#}$ is a distribution whose
compact support is contained in $E$, we may use \cite[Th.\
2.3.5]{Hor1} to get a representation of $T^{\#}$ as a finite sum
\[
  T^{\#} = {\textstyle\sum_{\alpha}} \, u_\alpha
  \otimes \partial_{\ss{E'}}^\alpha\delta.
\]
Here the tensor product refers to the identification
$\Xi = E \times E'$, $u_\alpha$ are compactly supported distributions
on $E$, and $\partial_{\ss{E'}}^\alpha\delta$ are derivatives of the
Dirac measure at zero in directions of $E'$. Then by \cite[Th.\
6.27]{Rud} we may write each $u_\alpha$ as a finite sum of derivatives
(in the directions of $E$) of continuous functions with compact
support on $E$ so $T^{\#}$ is a finite sum
\begin{equation}\label{eq:Tdiez}
T^{\#} = {\textstyle\sum_{\alpha,\beta}} \, \partial_\sse^\beta
v_{\alpha\beta} \otimes \partial_{\ss{E'}}^\alpha\delta,
\quad\text{where } v_{\alpha\beta}\in C_\rmc(E) .
\end{equation}

\begin{lemma}\label{lm:teknik}
  There is a finite family $\{w_\alpha\}_{|\alpha|\leq k}$ of
  functions in $C^\infty_{\rmc} (E)$ such that
\begin{equation}\label{eq:teknik}
  \rho\oast T^{\#} =
  {\textstyle\sum_{|\alpha|\leq k}} \, w_\alpha
  \otimes \partial_{\ss{E'}}^\alpha \delta. 
\end{equation}
\end{lemma}

\begin{proof}
  We will give a representation of the form \eqref{eq:teknik} for each
  term of the sum in \eqref{eq:Tdiez}, then we obtain
  \eqref{eq:teknik} by adding these representations.  We simplify the
  writing and set $v=v_{\alpha\beta}$ and
  $T^{\#} =\partial_\sse^\beta v\otimes
  \partial_{\ss{E'}}^\alpha\delta$ with $v\in C_\rmc(E)$. If
  $\theta \in C^\infty_{\rmc} (\Xi)$ then, by \eqref{eq:tconv} and
  with the abuse of notation mentioned in \S\ref{sss:312}, the action
  of the distribution $\rho\oast T^{\#}$ on $\theta$ is
\begin{align*}
  \bracet{\theta}{\rho\oast T^{\#}}
  & =
    \int_\Xi\int_\Xi  \theta(\xi +\eta)\,
    \e^{\frac{\rmi}{2}\sigma(\eta,\xi)}\rho(\xi)\d_E\xi\,
    T^{\#}(\eta)\dd\eta \\
  & = 
    \int_{E}\rho(\xi) \d_E\xi
    \int_\Xi\e^{\frac{\rmi}{2}\sigma(\eta,\xi)}\,\theta(\xi +\eta)\,
    T^{\#}(\eta)\dd\eta \\
  & \equiv
    \int_{E} \left\langle
    \e^{\frac{\rmi}{2}\xi_\sigma}\,\theta(\xi +\cdot\,)\,,\,T^{\#}
    \right\rangle
    \rho(\xi) \d_E\xi.
\end{align*}
If $\eta\in \Xi$ we denote by $y$ and $y'$ its components in $E$ and
$E'$ respectively.  Then $T^{\#}$ may be written
$T^{\#}(\eta)\equiv\partial_{y}^\beta
v(y)\partial_{y'}^\alpha\delta(y')=v^{(\beta)}(y)\delta^{(\alpha)}(y')$
which of course is slightly formal and must be interpreted in the
sense of distributions. Anyway, the action of the distribution
$T^{\#}$ on $\e^{\frac{i}{2}\xi_\sigma}\theta(\xi +\cdot\,)$ is given
by
\begin{align*}
  &\int_{E'} \d_{\ss{E'}} y'
    \int_{E}\e^{\frac{\rmi}{2}\sigma(y+y',\xi)}\,\theta(\xi+y+y')
    v^{(\beta)}(y)\delta^{(\alpha)}(y') \dd_\sse y \\
  & = (-1)^{|\alpha|} \int_{E} \e^{\frac{\rmi}{2}\sigma(y,\xi)}
    \left[\partial_{y'}^\alpha
    (\e^{\frac{i}{2}\sigma(y',\xi)}\theta(\xi+y+y') )\right]_{y'=0}
    v^{(\beta)}(y) \d_\sse y \\
  & = \sum_{\gamma\leq \alpha}(-1)^{|\alpha|}\int_{E}
    \e^{\frac{\rmi}{2}\sigma(y,\xi)} v^{(\beta)}(y)
    \varphi_{\alpha-\gamma}(\xi)
    \left[ \partial_{y'}^\gamma \theta(\xi+y+y')
    \right]_{y'=0}\dd_\sse y \\
  & =\sum_{\gamma\leq \alpha}(-1)^{|\alpha|+|\beta|} 
       \varphi_{\alpha-\gamma}(\xi) \int_{E} v(y)
    \partial_{y}^\beta\left(\e^{\frac{\rmi}{2}\sigma(y,\xi)}
    \left[\partial_{y'}^\gamma \theta(\xi
    +y+y')\right]_{y'=0}\right)\dd_\sse y .
\end{align*}
for some polynomials $\vphi_{\alpha-\gamma}(\xi)$ coming from the
derivatives of order $\alpha-\gamma$ of
$\e^{\frac{\rmi}{2}\sigma(y',\xi)}$ with respect to $y'$ at $y'=0$.
Then we get
\begin{align*}
  \langle \theta , \rho\oast T^{\#}\rangle
  &= \sum_{\gamma\leq \alpha}(-1)^{|\alpha|+|\beta|}
    \int_{E} \dd_\sse\xi\, \rho(\xi)\,\vphi_{\alpha-\gamma}(\xi)\oast \\ 
  & \oast  \int_{E}\! v(y) 
    \partial_{y}^\beta\left(\e^{\frac{\rmi}{2}\sigma(y,\xi)}
    \left[\partial_{y'}^\gamma \theta(\xi
    +y+y')\right]_{y'=0}\right)\dd_\sse y .
\end{align*}
By taking into account the relation
\begin{align*}
\partial_{y}^\beta &\left(\e^{\frac{i}{2}\sigma(y,\xi)}
\left[\partial_{y'}^\gamma \theta(\xi +y+y')\right]_{y'=0}\right)\\
& = \sum_{\lambda\leq\beta} \binom{\beta}{\lambda}
\left(\partial_{y}^{\beta-\lambda}\e^{\frac{i}{2}\sigma(y,\xi)}\right)
\partial_{y}^{\lambda}\left[\partial_{y'}^\gamma
\theta(\xi +y+y')\right]_{y'=0}
\end{align*}
and since
$\partial_{y}^{\lambda}
\left[\partial_{y'}^\gamma \theta(\xi +y+y')\right]_{y'=0}
=\partial_{\xi}^{\lambda}\left[\partial_{y'}^\gamma \theta(\xi
  +y+y')\right]_{y'=0}$ 
we get
\begin{align*}
\bracet{\theta}{\rho\oast T^{\#}}
&= \sum_{\gamma,\lambda}
C_{\gamma\lambda}\int_{E}\dd_\sse \xi \int_{E}
\partial_{\xi}^\lambda\left(\rho(\xi)\,\vphi_{\alpha-\beta}(\xi)
\partial_{y}^{\beta-\lambda}\e^{\frac{\rmi}{2}\sigma(y,\xi)}\right) \\
& \hspace{30 mm} \cdot v(y)
\left[\partial_{y'}^\gamma \theta(\xi+y+y')\right]_{y'=0}\dd y\\ 
&=\sum_{\gamma,\lambda} C_{\gamma\lambda}
\int_{E} \dd_\sse\tau \int_{E}\partial_{\tau}^\lambda
\left(\rho(\tau -y)\vphi(\tau -y)\partial_{y}^{\beta-\lambda}
\e^{\frac{\rmi}{2}\sigma(y,\tau)}\right) \\
& \hspace{30 mm} \cdot v(y)
\left[\partial_{y'}^\gamma \theta(\tau +y')\right]_{y'=0}\dd y.
\end{align*}
Now it is clear that there is a finite number of functions
$w_{\gamma}\in C_{\rmc}^{\infty}(E)$ such that
\[
\bracet{\theta}{\rho\oast T^{\#}}= 
\sum_\gamma \int_{E}w_{\gamma}(\tau)\,
\left[\partial_{y'}^{\gamma} \theta(\tau+y')\right]_{y'=0}\dd_\sse\tau  
=\sum_{\gamma}
\bracet{\theta}{w_{\gamma}\otimes\partial_{\ss{E'}}^{\gamma}\delta} .
\qedhere
\]
\end{proof}

From Lemma \ref{lm:teknik} we see that (\ref{e:mut}) holds if we may
take $k=0$ in the representation \eqref{eq:teknik} of
$\rho\oast T^{\#}$.  So the proof of the Theorem \ref{th:mainthfd} is
finished once we show that {\em the operator $W(\rho\oast T^{\#})$ is
  bounded if and only if $k=0$\/}. Before starting this last step of
the proof, we describe a direct sum decomposition of the symplectic
space $\Xi$ determined by $E$.

\ssubsection{}\label{sss:315}
\addtocontents{toc}{\SkipTocEntry}

Recall that the center of $E$ is $E^{\rmc}\doteq E\cap E^\sigma$ and
set $\bar{E}:=E+E^\sigma$.  A subspace is symplectic if its center
reduces to $0$.  Since $\Xi$ is finite dimensional, by using
$ E^{\sigma\sigma}=E$ it is easy to see that
$E^{\rmc}= E^{\sigma\rmc}=\bar{E}^{\rmc}=\bar{E}^{\sigma}$ and
$\bar{E}=E^{\rmc\sigma}=\overline{E^{\rmc}}=\overline{E^\sigma}$.  In
what follows, we shall denote by $\oplussymp$ and $\oplus$ the
symplectic (respectively vector) direct sum, and by $\perp$ the
symplectic orthogonality between elements of $\Xi$.  Let now
$G\subset E$ such that $E = G\oplus E^{\rmc}$ and
$F\subset E^{\sigma}$ such that $E^{\sigma} = F\oplus E^{\rmc}$.  Then
{\em $G$ and $F$ are symplectic\/}: indeed, suppose $\xi \in G$ and
$\xi \perp G$; then $E^{\rmc} \subset E^\sigma$ shows that
$\xi \perp E^{\rmc}$, so $\xi \perp (G+E^{\rmc})$, thus
$\xi \in E^{\sigma}$, which means that
$\xi \in E \cap E^\sigma = E^{\rmc}$; then $G\cap E^{\rmc}=0$ shows
$\xi =0$.  Hence $H\equiv G\oplussymp F$ {\em is a symplectic subspace
  too, and we have $\bar{E} =H \oplussymp E^{\rmc}$\/}.  For, if
$\xi \in H\cap E^{\rmc}$ then there are $\xi_{\ss{G}} \in G$ and
$\xi_{\ss{F}} \in F$ such that
$\xi_{\ss{G}} = \xi - \xi_{\ss{F}} \in
E^{\sigma}-E^{\sigma}=E^{\sigma}$ which shows that
$\xi_{\ss{G}} \perp G \subset E$.  Since $G$ is symplectic,
$\xi_{\ss{G}} =0$ and the same holds for $\xi_{\ss{F}}$. On the other
hand,
$H+E^{\rmc}=G+F+E^{\rmc}=(G+E^{\rmc})+(F+E^{\rmc})
=E+E^{\sigma}=\bar{E}$.  Further, remark that {\em $H^\sigma$ is
  also a symplectic space and $E^{\rmc}$ is a Lagrangian subspace of
  $H^\sigma$\/}.  Indeed, $E^{\rmc} \perp H$ thus
$E^{\rmc}\subset H^\sigma$ and is isotropic.  If it where not maximal,
there would exist some $\xi \in H^\sigma \setminus E^{\rmc}$ with
$\xi \perp E^{\rmc}$. Thus $\xi \perp (E^{\rmc}+H)=\bar{E}$,
i.e.\ $\xi \in \bar{E}^\sigma =E^{\rmc}$, which is absurd.
 
Now let $K$ be a Lagrangian subspace of $ H^\sigma$ such that
$ H^\sigma = E^{\rmc} \oplus K$ (this will be a Lagrangian decomposition
of the symplectic space $ H^\sigma$).
Then $\Xi$ splits as
\[
\Xi = H \oplussymp H^\sigma = G \oplussymp F\oplussymp (E^{\rmc}
\oplus K) = E \oplus F \oplus K.
\]

\ssubsection{}\label{sss:316}
\addtocontents{toc}{\SkipTocEntry}

Let us go back now to the proof of the fact that $W(\rho\oast T^{\#})$
is bounded only if $k=0$.  Note that we may take $E'=F \oplus K$.  Any
$\xi \in \Xi$ may be uniquely written as
$\xi = \eta + \zeta_{\ss{F}} +\zeta_{\ss{K}}$ with $\eta \in E$,
$\zeta_{\ss{F}}\in F$ and $\zeta_{\ss{K}}\in K$.  In the next formulas
we abbreviate $\lambda_E(\dd\eta)=\d\eta$,
$\lambda_F(\dd\zeta_{\ss{F}})=\dd\zeta_{\ss{F}}$ and
$\lambda_K(\dd\zeta_{\ss{K}})=\dd\zeta_{\ss{K}}$.  Then, by
\eqref{eq:Wmu} and \eqref{eq:teknik},
\[
W(\rho\oast T^{\#}) \!=\! \sum_{|\alpha|\leq k}\int_{F\oplus K}\int_{E}
W(\eta +\zeta_{\ss{F}} +\zeta_{\ss{K}})w_{\alpha}(\eta)
[\partial_{\ss{F\oplus K}}^{\alpha}\delta](\zeta_{\ss{F}} +\zeta_{\ss{K}})
\d\eta\dd\zeta_{\ss{F}}\dd\zeta_{\ss{K}}
\]
which can be further developed as follows: 
\begin{align*}
&\sum_{|\alpha|\leq k}\int_{F\oplus K}\dd\zeta_{\ss{F}}\dd\zeta_{\ss{K}}
W(\zeta_{\ss{F}} +\zeta_{\ss{K}}) \int_{E}
\e^{\frac{\rmi}{2}\sigma(\zeta_{\ss{F}} +\zeta_{\ss{K}},\eta)}\,W(\eta)
w_{\alpha}(\eta)[\partial_{\ss{F\oplus K}}^{\alpha}\delta]
(\zeta_{\ss{F}} +\zeta_{\ss{K}}) \d\eta \\
& = \sum_{|\alpha|\leq k} \int_{F\oplus K}
\dd\zeta_{\ss{F}}\dd\zeta_{\ss{K}}W(\zeta_{\ss{F}})W(\zeta_{\ss{K}})
\int_{E} \e^{\frac{\rmi}{2}\sigma(\zeta_{\ss{K},\eta})}W(\eta)
w_{\alpha}(\eta)[\partial_{\ss{F\oplus K}}^{\alpha}\delta]
(\zeta_{\ss{F}} +\zeta_{\ss{K}}) \d\eta \\
& = \sum_{|\alpha|+|\beta|\leq k} (-1)^{|\alpha|+|\beta|}
[\partial_{\ss{F}}^{\alpha}W(\zeta_{\ss{F}})]_{\zeta_{\ss{F}}=0}
\!\!\int_{E}\! [\partial_{\ss{K}}^{\beta}\{W(\zeta_{\ss{K}})
  \e^{\frac{\rmi}{2}\sigma(\zeta_{\ss{K}},\eta)}\}]_{\zeta_{\ss{K}}=0}
W(\eta)w_{\alpha}(\eta)\d\eta \\
& = \sum_{|\alpha|+|\beta|\leq k}
[\partial_{\ss{F}}^{\alpha}W(\zeta_{\ss{F}})]_{\zeta_{\ss{F}}=0}
[\partial_{\ss{K}}^{\beta}W(\zeta_{\ss{K}})]_{\zeta_{\ss{K}}=0}
W(w_{\alpha\beta}).
\end{align*}
We have used the relations $F \perp K$, hence
$W(\zeta_{\ss{F}}+\zeta_{\ss{K}})=W(\zeta_
{\ss{F}})\,W(\zeta_{\ss{K}})$, and $F \perp E$, hence
$\sigma(\eta,\zeta_{\ss{F}})=0$.  Further, we took into account that
the derivatives with respect to $\zeta_{\ss{K}}$ at $\zeta_{\ss{K}}=0$
of $\e^{\frac{\rmi}{2}\sigma(\zeta_{\ss{K}},\eta)}$ are polynomial
functions of $\eta \in E$ and $w_{\alpha}\in C^{\infty}_{\rmc}(E)$,
thus there are $w_{\alpha\beta}\in C_\rmc^\infty(E)$ such that the
last equality be true with $W(w_{\alpha\beta})=\int_E
W(\eta)w_{\alpha\beta}(\eta)\d\eta$. 

Let $F=\Omega \oplus \Omega^*$ and $G=\Gamma \oplus \Gamma^*$ be some
Lagrangian decompositions of the symplectic subspaces $F$ and $G$.
Thus we get a decomposition
\[
\Xi = (\Omega \oplus \Omega^*)\oplussymp
(\Gamma \oplus \Gamma^* )\oplussymp ( E^{\rmc}\oplus K)
\]
each of the three parentheses being a symplectic space equipped with a
Lagrangian decomposition. We also have
$E=(\Gamma \oplus \Gamma^* )\oplussymp E^{\rmc}$.  Then
$[\partial_{\ss{F}}^{\alpha}W(\zeta_{\ss{F}})]_{\zeta_{\ss{F}}=0}\equiv
\phi_\alpha (q_{\ss{\Omega}},p_{\ss{\Omega}})$ and
$[\partial_{\ss{K}}^{\beta}W(\zeta_{\ss{K}})]_{\zeta_{\ss{K}}=0}
\equiv \psi_\beta (p_{\ss{E^{\ss c}}})$ with $\phi_\alpha$ and
$\psi_\beta$ polynomial functions of degrees $|\alpha|$ and $|\beta|$
respectively and $q,p$ are the position and momentum observables
associated to the space in the index.  Observe that the notation
$p_{\ss{E^{\ss c}}}$ is justified by the identification
$(E^{\rmc})^* = K$. Moreover, the operators $W(\mu_{\alpha\beta})$ can
be written in the form
$\omega_{\alpha\beta}(q_{\ss\Gamma},p_{\ss\Gamma},q_{\ss{E^{\ss c}}})$
for some smooth function $\omega_{\alpha\beta}$, expression which can
be rigorously interpreted in terms of the Weyl calculus (this fact,
however, does not play any role in what follows). Thus we have
\[
W(\rho\oast T^{\#})= \sum_{|\alpha|+|\beta|\leq k} \phi_\alpha
(q_{\ss{\Omega}},p_{\ss{\Omega}})\, \psi_\beta (p_{\ss{E^{\ss c}}})\,
\omega_{\alpha\beta}(q_{\ss\Gamma},p_{\ss\Gamma},q_{\ss{E^{\ss c}}})
\]
and we have to prove that if this operator is bounded then necessarily 
$k=0$.
 
From now on we may assume that we work in the Schr\"odinger
representation associated to $X=\Omega\oplus\Lambda$, where
$\Lambda=\Gamma\oplus E^{\rmc}$.  Then our Hilbert space factorizes as
$\ch=L^2(\Omega)\otimes L^2(\Lambda)$ and, if we denote
$\phi_{\alpha} (q_{\ss{\Omega}},p_{\ss{\Omega}})$ by $S_\alpha$ and
$\sum_{\beta} \psi_{\beta}(p_{\ss{E^{\ss c}}})
\omega_{\alpha\beta}(q_{\ss\Gamma},p_{\ss\Gamma},q_{\ss{E^{\ss c}}})$
by $T_\alpha$ we have
$W(\rho\oast T^{\#})=\sum_{|\alpha|\leq k}S_\alpha\otimes T_\alpha$.
We show first that $S_\alpha \in B(L^2(\Omega))$ for all $\alpha$; in
fact they are complex multiples of the identity operator on
$L^2(\Omega)$.
 
Notice that $T_\alpha\in B({\mbs}(\Lambda))$ and that  we may
assume that the family 
$\{T_\alpha\}_{|\alpha| \leq k}$ is linearly independent in
$B({\mbs}(\Lambda))$.
Let $\rond T$ denote the vector space of operators
$L: {\mbs}(\Lambda)\rarrow {\mbs}(\Lambda)$
of the form $L=\sum_{i}\ket{u_i}\bra{v_i}$ (finite sum), with
$u_i , v_i \in {\mbs}(\Lambda)$.
Then we may realize $\rond T$ as a space of linear forms on
$B({\mbs}(\Lambda))$ by defining
$B({\mbs}(\Lambda))\ni T \mapsto
\tr{LT}=\sum_{i}\braket{v_i}{Tu_i}$. Notice that
${\rond T}$ is a subspace of the dual of $B({\mbs}(\Lambda))$
which separates the points
(for, if $\tr{LT}=\tr{LT'}$ for all $L\in \rond T$, then
$\braket{u,(T-T')v}=0$ for all $u,v \in {\mbs}(\Lambda)$).
It follows that we may find a finite family of operators
$\{L_\alpha\} \in \rond T$ such that 
\[
\tr{L_\alpha T_\beta} =\delta_{\alpha\beta}
\mbox{\ for all\ } \alpha,\beta \mbox{\ with\ } |\alpha|\leq k,\ 
|\beta| \leq k.
\]
Indeed, if $\rond V$ is the finite dimensional vector subspace
generated by the family $\{T_\alpha\}$ (which is a basis in it) we may
find a basis $\{\Psi_\alpha\}$ in the dual space ${\rond V}'$ such
that $\Psi_\alpha (T_\beta)=\delta_{\alpha\beta}$.  Thus we are
reduced to prove that for each $\Psi \in \rond V'$ we may find
$L\in \rond T$ such that $\Psi(T)=\tr{LT}$ for all $T \in \rond V$,
i.e.\ that the mapping
${\rond T}\ni L \mapsto \tr{L\,\cdot\,}\!  \mid_{\!_{\ss{\rond V}}}\in
{\rond V}'$ is surjective.  Equivalently, the dual mapping
${\rond V}\equiv{\rond V}''\ni T \mapsto \tr{\,\cdot\, T}\in {\rond
  T}'$ has to be injective.  But this is true, since $\rond T$
separates points.
 
Let now $L_\alpha =\sum_{i}\ket{u_{\alpha i}} \bra{v_{\alpha i}}$ with
$u_{\alpha i},v_{\alpha i} \in {\mbs}(\Lambda)$ and, for
$u,v \in {\mbs}(\Omega)$, let us compute
\begin{align*}
\sum_{i} & \braket{v\otimes v_{\alpha i}}{W(\rho\oast T^{\#})
u\otimes u_{\alpha i}} \\
& = \sum_i\sum_{|\gamma|\leq k}\braket{v}{S_\gamma u}
\braket{v_{\alpha i}}{T_\gamma u_{\alpha i}}
=
\sum_{|\gamma|\leq k}\braket{v}{S_\gamma u}\tr{L_\alpha T_\gamma}
=
\braket{v}{S_\alpha u}.
\end{align*}
Since $W(\rho\oast T^{\#})$ is supposed bounded, this shows that all
$S_\alpha$ are bounded. But $S_\alpha$ are polynomials
$\phi_{\alpha} (q_{\ss{\Omega}},p_{\ss{\Omega}})$, so these
polynomials have to be of degree zero, i.e.\ $\phi_\alpha$ are complex
numbers.
 
Now we can write
$W(\rho\oast T^{\#})= 1\otimes\sum_\alpha \phi_\alpha T_\alpha$.  Thus
it remains to show that if an operator of the form
$B=\sum_{\beta} \psi_\beta(p_{\ss{E^{\ss c}}})W(\mu_{\beta})$ is
bounded in $L^2(\Lambda)$, then the polynomials $\psi_\beta$ are
constants. By rearranging the sum we can assume that
$B=\sum_\gamma p_{\ss{E^{\ss c}}}^\gamma W(\mu_{\gamma})$ (a basis in
$E^{\rmc}$ has been chosen).  Let $a=\lambda\zeta$, with
$\lambda \in \R$ and $\zeta\in E^{\rmc} $ and for any
$u,v\in {\mbs}(\Lambda)$, denote
$u_a = \exp{[\rmi\bracet{q_{\ss{E^{\ss c}}}}{a}]}u$ and
$v_a = \exp{[\rmi\bracet{q_{\ss{E^{\ss c}}}}{a}]}v$.  Then
$|\braket{u_a}{Bv_a}|\leq C\,\|u\|\,\|v\|$, and on the other hand
\begin{align*}
&|\braket{u_a}{Bv_a}| = |\braket{u}{\sum_{|\gamma|\leq m}
(p_{\ss{E^{\ss c}}}+a)^{\gamma}W(\mu_\gamma) v}|\\
&= \lambda^m \biggl| \sum_{|\gamma|= m}\!\!
\zeta^\gamma\braket{u}{W(\mu_\gamma) v} +\!\!\!
\sum_{|\gamma|< m}\sum_{\delta\leq \gamma}
\!\!\binom{\gamma}{\delta}\lambda^{-m+|\gamma|-|\delta|}
\zeta^{\gamma-\delta}
\braket{u}{p_{\ss{E^{\ss c}}}^{\delta}W(\mu_\gamma) v}
\biggr|.
\end{align*}
By taking $\lambda\rarrow \infty$ we see that we must have
$\sum_{|\gamma|= m}\zeta^\gamma\braket{u}{W(\mu_\gamma) v} =0$ for all
$\zeta \in E^{\rmc}$ if $m\geq 1$.  But this implies
$\braket{u}{W(\mu_\gamma) v} =0$ for all $|\gamma|\geq 1$ by a
standard argument.  This ends the proof of Theorem \ref{th:mainthfd}.

\subsection{The ideal
  \texorpdfstring{$\rf(\Xi)$}{F(Xi)}}\label{ss:rkr}

Recall that $L^1(\Xi)$ is a Banach $*$-algebra for the twisted
convolution and the definition $W(f) = \int_\Xi W(\xi) f(\xi) \d\xi$
extends $W$ to an injective morphism $W\colon L^1(\Xi)\to B(\ch)$.
Then $\rf(\Xi)$ is the norm closure in $B(\ch)$ of the set of such
$W(f)$, hence is a $C^*$-algebra of operators on $\ch$ and is the
smallest graded ideal of $\rf$.  By Definition \ref{df:wkfield}, if
$\{\xi_1,\dots,\xi_n\}$ is a generating set for $\Xi$ then
\begin{equation}\label{eq:b-g-xi}
  \rf(\Xi)= C^*(\phi(\xi_1))\cdot
  C^*(\phi(\xi_2)) \cdot\ldots\cdot C^*(\phi(\xi_n)) .
\end{equation}
The commutant algebra $\Com(\Xi)$ and the multiplicity of $W$ are
introduced in the Remark \ref{re:stvon}.  We use the following
convention: if $T\in B(\ch)$ and the symbol $T^{(*)}$ appears in a
relation, then that relation has to be satisfied both by $T$ and by
$T^*$. Note that the assertion (a) below is a version of the
Kolmogorov-Riesz compactness criterion.

\begin{theorem}\label{th:RKR}
Let $W$ be a representation of a finite dimensional symplectic space
$\Xi$.
\begin{itemize}
\item[{\rm(1)}]
  $\rf(\Xi)\cdot \Com(\Xi)=\{T\in
  B(\ch)\mid\lim_{\xi\to0}\|(W(\xi)-1)T^{(*)}\|=0\}$. 
\item[{\rm(2)}] $W$ is irreducible
  $\Leftrightarrow \rf(\Xi)=K(\ch) \Leftrightarrow \rk(\Xi)\supset
  K(\ch)$ and then
  \begin{equation}  \label{eq:propirr}
    \rf(\Xi)=\{T\in B(\ch) \mid
    {\textstyle\lim_{\xi\to0}}\|(W(\xi)-1)T\|=0\}.
  \end{equation}
  
\item[{\rm(3)}] $W$ is of finite multiplicity
  $\Leftrightarrow \rf(\Xi)\subset K(\ch) \Leftrightarrow
  \rf(\Xi)\cdot\Com(\Xi)=K(\ch)$; then 
\begin{itemize}
\item[{\rm(a)}] $\cj\subset\ch$ is relatively compact
  $\Leftrightarrow \lim_{\xi\to0}\sup_{h\in\ck}\|(W(\xi)-1)h\|=0$;
\item[{\rm(b)}] $T\in B(\ch)$ is compact $\Leftrightarrow$
  $\lim_{\xi\to0}\|(W(\xi)-1)T\|=0$.
\end{itemize}

\item[{\rm(4)}] $W$ is of infinite multiplicity if and only if
  $\rk(\Xi)\cap K(\ch)=0$.
\end{itemize}
\end{theorem}

\begin{proof}

  By the Stone-Von Neumann theorem we may assume $\ch=\ch_0\otimes\ck$
  with $\ck$ finite dimensional and $W(\xi)=W_0(\xi)\otimes1$ with
  $W_0$ is an irreducible representation of $\Xi$ on $\ch_0$.
  
  We first discuss (2).  If $W$ is irreducible then $\rf(\Xi)= K(\ch)$
  by a simple argument in the Schr\"odinger representation and this
  implies $K(\ch)\subset\rk(\Xi)$. On the other hand, if $W$ is not
  irreducible we cannot have $K(\ch)\subset\rk(\Xi)$. Indeed, if
  $\dim\ck>1$ then an operator of the form $K\otimes F$ with
  $K\in B(\ch_0)$ and $F\in B(\ck)$ of rank one is compact on $\ch$
  but does not belong to $\rk(\Xi)$. Then \eqref{eq:propirr} follows
  from \cite[Corollary 3.5]{GI0}.

  From $\rk_\ssw(\Xi)=\rk_{\ssw_0}(\Xi)\otimes 1_{\ck}$ (4).  And
  $\rf^\ssw(\Xi)=K(\ch_0)\otimes 1_{\ck}$ hence $W$ is of finite
  multiplicity if and only if $\rf^\ssw(\Xi)\subset K(\ch)$ which
  proves the first assertion of (3).

  Next we prove (1).  It is easy to see that $\Com^\ssw(\Xi)$ coincides
  with the commutant $\rf^\ssw(\Xi)'$ of the algebra $\rf^\ssw(\Xi)$
  (\conf the argument which proves (a) of Proposition
  \ref{pr:properties}).  On the other hand, the commutant of
  $K(\ch_0)\otimes 1_{\ck}$ coincides with the commutant of its weak
  closure, which is $B(\ch_0)\otimes 1_{\ck}$, and by the
  commutation theorem for tensor products the commutant of the last
  algebra is $1_{\ch_0}\otimes B(\ck)$. Thus
  $\Com^\ssw(\Xi)=1_{\ch_0}\otimes B(\ck)$ hence
  \[
    \rf^\ssw(\Xi)\cdot \Com^\ssw(\Xi)= K(\ch_0)\otimes 1_{\ck} \cdot
    1_{\ch_0}\otimes B(\ck) =K(\ch_0)\otimes B(\ck)
  \]
  where the last tensor product is the spatial tensor product of the
  $C^*$-algebras $K(\ch_0)$ and $B(\ck)$. Now the assertion (1) of
  the theorem follows from \cite[Th.\ 3.8]{GI0}.

  If $W$ is of finite multiplicity then $B(\ck)=K(\ck)$ hence
  $\Com(\Xi)=1_{\ch_0}\otimes K(\ck)$ and thus
  $\rf(\Xi)\cdot\Com(\Xi)=K(\ch_0)\otimes K(\ck)=K(\ch)$. Thus due
  to (1) an operator $T\in B(\ch)$ is compact if and only if
  $\lim_{\xi\to0}\|(W(\xi)-1)T^{(*)}\|=0$. But we have to eliminate
  the condition on $T^*$ in order to get the assertion (b) of (3).

  The implications $\Rightarrow$ in (a) and (b) of (3) are clear for
  any $W$ because $W$ is strongly continuous. Now we prove the
  implications $\Leftarrow$ assuming (a) is known in the irreducible
  case.

  Assume $W$ of finite multiplicity and let $\{e_1,\dots,e_n\}$ an
  orthonormal basis in the finite dimensional Hilbert space
  $\ck$. Then any $h\in\ch$ can be uniquely written as
  $h=\sum_ih_i\otimes e_i$ with $h_i\in\ch_0$ and
  $\|h\|^2=\sum_i\|h_i\|^2$ hence
  $\|(W(\xi)-1))h\|^2=\sum_i\|(W_0(\xi)-1) h_i\|^2$ Let
  $\cj\subset\ch$ and $\cj_i=\{h_i\mid h\in\cj\}\subset\ch_0$. If
  $\lim_{\xi\to0}\sup_{h\in\cj}\|(W(\xi)-1)h\|=0$ then we have
  $\lim_{\xi\to0}\sup_{h_i\in\cj_i}\|(W_0(\xi)-1)h_i\|=0$. Since part
  (a) of (3) is true in the case of irreducible representations, it
  follows that $\cj_i$ is relatively compact in $\ch_0$. Hence each
  $\cj_i\otimes e_i$ is relatively compact in $\ch$ so
  $\cj\subset\sum_i\cj_i\otimes e_i$ is relatively compact because a
  finite sum of relatively compact sets is relatively compact. This
  finishes the proof of part (a) of (3) for an arbitrary $W$ of finite
  multiplicity. Part (a) of (3) is an immediate consequence: if
  $T\in B(\ch)$ and $\lim_{\xi\to0}\|(W(\xi)-1)T\|=0$ then
  $\lim_{\xi\to0}\sup_{\|h\|\leq1}\|(W(\xi)-1)Th\|=0$ hence the set
  $\{Th\mid \|h\|\leq1\}$ is relatively compact which means that $T$
  is compact.

  Finally, assume $W$ is irreducible. If $\cj$ is bounded then part
  (a) of (3) follows from \cite[Theorem 3.4]{GI0}. In order to
  eliminate the boundedness condition we use an improved
  Kolmogorov-Riesz criterion \cite{Su,HHM}. We may assume that $\Xi$
  is the phase space associated to a configuration space $X$ and that
  $W$ is the corresponding Schr\"odinger representation.  Equip $X$
  with an Euclidean structure and denote $\chi_r$ is the
  characteristic function of the region $|x|>r$.  We prove first that
  the condition $\lim_{\xi\to0}\sup_{u\in\cj}\|W(\xi)h-h\|=0$ implies
  $\|\chi_r(q)h\|\to0$ as $r\to\infty$ uniformly in $h\in\cj$. This
  follows from the following more precise estimate which is of some
  independent interest: there is a number $c$ depending only on
  $\dim X$ such that for any $r>0$
  \begin{equation}\label{eq:abelard}
    \|\chi_{r}(q)h\|\leq c\sup_{|a|<4/r}
    \|\e^{\rmi\bracet{q}{a}}h -h\| \quad\forall u\in\ch .
  \end{equation}
  To prove it note that for any $\varepsilon>0$ we have
  \[
    \int_{|a|<\varepsilon}\|\e^{\rmi\bracet{q}{a}}h-h\|^2\d a = \int_X
    \d x |h(x)|^2 \int_{|a|<\varepsilon}
    2\big(1-\cos\bracet{x}{a}\big) \d a .
  \]
  In order to estimate the last integral for a fixed $x$ we identify
  $X=\R^d$ with the help of an orthonormal basis $e_1,\dots,e_d$ such
  that $x=|x|e_1$. If we set $a_1=s$ and $(a_2,\dots,a_d)=t$ then the
  ball $|a|\leq\varepsilon$ is described by the conditions
  $-\varepsilon\leq s\leq\varepsilon$ and $t^2\leq\varepsilon^2-s^2$,
  hence there are numbers $C',C''$ depending only on $d$ such that
  \begin{align*}
    &\int_{|a|<\varepsilon} 2\big(1-\cos\bracet{x}{a}\big) \d a
    =C' \int_{0}^\varepsilon
      \big(1-\cos(|x|s)\big) (\varepsilon^2-s^2)^{\frac{d-1}{2}} \d s \\
    &\geq C' \int_{0}^{\varepsilon/2}
      \big(1-\cos(|x|s)\big) (\varepsilon^2-s^2)^{\frac{d-1}{2}} \d s
      \geq C''  \varepsilon^{d-1}
      \int_{0}^{\varepsilon/2}\big(1-\cos(|x|s)\big) \d s.
  \end{align*}
  Then we get
  \begin{align*}
    \varepsilon^{-d}
    \int_{|a|<\varepsilon} &\|\e^{\rmi\bracet{q}{a}}h -h\|^2\d a
     \geq C'' \int_X\d x |h(x)|^2
    \int_{0}^{1/2}\big(1-\cos(\varepsilon|x|s)\big) \d s \\
     &=(C''/2) \int_X\d x |h(x)|^2
    \big(1-\sin(\varepsilon|x|/2)/(\varepsilon|x|/2) \big)
  \end{align*}
  If $\varepsilon|x|>4$ then the last parenthesis is $>1/2$.  Thus
  there is a number $C$ depending only on the dimension of $X$ such
  that for any $\varepsilon>0$ and any $h\in\ch$
\begin{equation}\label{eq:abel}
  \|\chi_{4/\veps}(q)h\|^2\leq C \varepsilon^{-d}
  \int_{|a|<\varepsilon} \|\e^{\rmi\bracet{q}{a}}h -h\|^2\d a. 
\end{equation}
This estimate is better than \eqref{eq:abelard}. Now the assertion (a)
follows from \cite[Theorem 1]{HHM}.
\end{proof}

\subsection{HVZ theorem}\label{ss:hvz}

The original HVZ theorem, due to Hunziker, Van Winter and Zhislin,
gives a description of the essential spectrum of a non-relativistic
$N$-body Hamiltonian in terms of the spectra of its internal
hamiltonians \cite{ABG,CFKS,DeG1}. Proposition 8.4.2 from \cite{ABG} is
an abstract version of this theorem valid for $C^*$-algebras graded by
finite semilattices that is very easy to prove and covers the usual
versions of the theorem, see \S 9.4 and Ch.\ 10 of \cite{ABG}. Our
results here cover all the observables affiliated to the field
algebra.

We begin with an abstract $C^*$-algebra versions of the HVZ theorem
which follows from Theorems \ref{th:gahvz} and \ref{th:gahvzcs}. Since
$\Xi$ is finite dimensional $\rl(\Xi)$ is an ideal of $\rl$ and if
$\mu\in\rl$ its $\rl$-essential spectrum $\rl\hyphen\spe(\mu)$ is
defined by \eqref{eq:rcspe}.  In a finite multiplicity representation of
$\Xi$ this is the essential spectrum of the operator corresponding to
$\mu$.

\begin{theorem}\label{th:ahvz}
  The map $\rl\to\bigoplus_{E\in\mbh(\Xi)}\rl_\sse$ defined by
  $\mu\mapsto\big(\cp_\sse\mu\big)_{E\in\mbh(\Xi)}$ is a morphism with
  kernel $\rl(\Xi)$ which gives an embedding
  $\rl/\rl(\Xi)\hookrightarrow\bigoplus_{E\in\mbh(\Xi)}\rl_\sse$.  For
  any $\mu\in\rl$ the set $\{\cp_\sse\mu \mid E\in \mbh(\Xi)\}$ is
  compact in $\rl$ and
  \begin{equation}\label{eq:hvzspexi}
    \rl\hyphen\spe(\mu)=\ccup_{E\in\mbh(\Xi)}\spec(\cp_\sse\mu).
  \end{equation} 
  Let $\cs\subset\mbg(\Xi)$ a subsemilattice containing $\Xi$ and
  $\cs_{\max}$ the set of maximal elements of $\cs\setminus\{\Xi\}$. 
  If $\mu\in\rl(\cs)$ then
  \begin{equation}\label{eq:hvzspexics}
    \rl\hyphen\spe(\mu)=\ccup_{E\in\cs_{\max}}\spec(\cp_\sse\mu) .
  \end{equation} 
  If $\cs$ is a countable set then $\cp_E$ is expressed in terms of
  translations at infinity as in \eqref{eq:liminfcs}.
\end{theorem}

We give now a Hilbert space version of these algebraic results.

\begin{lemma}\label{lm:nocomp} $ $\\[1mm]
  {\rm(1)} If $\dim\Xi<\infty$ and $W$ is of finite multiplicity then
  $\rf(\Xi)=\rf\cap K(\ch)$. \\[1mm]
  {\rm(2)} If $\dim\Xi<\infty$ and $W$ is of infinite multiplicity or
  $\dim\Xi=\infty$ then $\rf\cap K(\ch)=0$.
\end{lemma}

\begin{proof}
  If $W$ is a representation of an arbitrary symplectic space and if
  $\xi\neq0$ the spectrum of $\phi(\xi)$ is equal to $\R$ and purely
  absolutely continuous, hence
  \begin{equation}\label{eq:wdecay}
    \wlim_{r\to\infty}W(r\xi)=0 \quad\text{if } \xi\neq0 .
  \end{equation}
  Thus if $T\in\rf$ is compact then $\slim_{r\to\infty}T W(r\xi)=0$
  $\forall\xi\neq0$ hence $\cp_\sse T=0$ with the notations of
  Proposition \ref{pr:Fxisigma}.  If $\dim\Xi<\infty$ we get
  $\cp_\sse T=0$ for any hyperplane $E$ and then Theorem \ref{th:ahvz}
  implies $T\in\rf(\Xi)$. Thus $\rf\cap K(\ch)\subset\rf(\Xi)$ and if
  $W$ of finite multiplicity then the inverse inclusion is true by
  Theorem \ref{th:RKR}.  If $W$ is of infinite multiplicity then
  $\rf\cap K(\ch)=0$ by Remark \ref{re:isonot}.
  
  Now assume $\dim \Xi=\infty$ and let $T\in\rf\cap K(\ch)$. Then
  $\cp_\sse T=0$ if $E=\xi^\sigma$ with $\xi\neq0$ by the argument at the
  beginning of this proof. If $\varepsilon>0$ there is $F\in\mbg(\Xi)$
  and $S\in\rf_\ssf$ such that $\|S-T\|<\varepsilon$. If
  $\xi\in F^\sigma$ is nonzero then $F\subset\xi^\sigma=E$ hence
  $\cp_\sse S=S$ and $\cp_\sse T=0$ so $\|S\|<\varepsilon$ hence
  $\|T\|<2\varepsilon$.  Since $\varepsilon$ is arbitrary we get
  $T=0$.
\end{proof}

\begin{proposition}\label{pr:nohvz}
  If $\Xi$ is infinite dimensional or if $\Xi$ is finite dimensional
  and $W$ is of infinite multiplicity then $\rf\cap K(\ch)=0$, hence
  $\spe(T)=\spec(T)$ if $T\win\rf$. 
 \end{proposition}

 This is a consequence of \eqref{eq:easy} and Lemma \ref{lm:nocomp}.
 Then we have a general HVZ type theorem:

\begin{theorem}\label{th:hvz1}
  Assume $\Xi$ finite dimensional and $W$ of finite multiplicity.  If
  $T\win\rf$ then
\begin{equation}\label{eq:hvz1}
  \spe(T) = \ccup_{E\in\mbh(\Xi)}\spec(\cp_\sse T) .
\end{equation}
Moreover, for any $E\in\mbh(\Xi)$ and any nonzero $\xi\in E^\sigma$ we
have
\begin{equation}\label{eq:hvz2}
\cp_\sse T=\slim_{|r|\rarrow\infty} W(r\xi)^*T W(r\xi) .
\end{equation}
\end{theorem}

\begin{proof}  
  We will deduce this from Theorem \ref{th:ahvz} by taking into
  account the canonical isomorphisms between the algebras $\rl$ and
  $\rf$. For any operator $T\in\rf$ let us denote {\small
    $\Xi$}-$\spe(T)$ the spectrum of the image of $T$ in the quotient
  algebra $\rf/\rf(\Xi)$. Then \eqref{eq:hvzspexi} give us 
  \begin{equation}\label{eq:hvzxi}
    \text{\small $\Xi$-}\spe(T)=\ccup_{E\in\mbh(\Xi)}\spec(\cp_\sse T).
  \end{equation} 
  Taking into account Proposition \ref{pr:Fxisigma}, it remains only
  to prove that $\text{\small $\Xi$-}\spe(T)$ coincides with the
  essential spectrum of $T$ in the Hilbert space sense. But (1) of
  Lemma \ref{lm:nocomp} implies that $\rf/\rf(\Xi)$ is the image of
  $\rf$ in the Calkin algebra.
\end{proof}

The next result is an immediate consequence of Theorem \ref{th:ahvz}
(by the same argument as above).  In the simple case when $\cs$ is
finite this is a generalisation of the usual $N$-body HVZ theorem,
\conf Ch.\ 10, \S 9.4 and Proposition 8.4.2 in \cite{ABG}.

\begin{theorem}\label{th:hvz3}
  Let $\Xi$ be finite dimensional and $W$ of finite multiplicity.If
  $\cs\subset\mbg(\Xi)$ is a subsemilattice containing $\Xi$ and
  $\cs_{\max}$ the set of maximal elements of $\cs\setminus\{\Xi\}$,
  then
  \begin{equation}\label{eq:hvz3}
    \spe(T)=\ccup_{E\in\cs_{\max}}\spec(\cp_\sse T) \quad
    \forall\, T\win\rf(\cs).
  \end{equation} 
\end{theorem}

One may use Theorem \ref{th:hvz1} to show that some operators are not
affiliated to $\rf$. Note first the following consequence of
\eqref{eq:hvz2}.

\begin{corollary}\label{co:lim0}
  Assume $\Xi$ finite dimensional and $W$ of finite multiplicity. Then
  if $T\in\rf$ we have $\slim_{r\to\infty}W(r\xi)^*T W(r\xi)=0$
  $\forall\xi\neq0$ if and only if $T$ is a compact.
\end{corollary}

From this it follows that the generator of the dilation group
associated to a Lagrangian decomposition of $\Xi$ is a self-adjoint
operator not affiliated to $\rf$. We prove this for $\dim\Xi=2$, the
general case is similar.  We may assume $\Xi=T^*\R$ and we work in the
Schr\"odinger representation on $\ch=L^2(\R)$.  Then $(qp+pq)/2$ is
essentially self-adjoint on $C_\rmc^\infty(\R)$ and its closure is a
self-adjoint operator $\omega$ such that
$(\e^{\rmi t\omega}f)(x)=\e^{t/2}f(\e^tx)$, so $\omega$ is the
\emph{generator of the dilation group} on $\ch$.  Clearly
$\e^{\rmi t\omega}q\e^{-\rmi t\omega}= \e^{t}q$ and
$\e^{\rmi t\omega}p\e^{-\rmi t\omega}= \e^{-t}p$ and the field
operators are of the form $aq+bp$ with $a,b\in\R$ hence
$\e^{\rmi t\omega}\rf\e^{-\rmi t\omega}=\rf$ $\forall t\in\R$ and the
dilations induce a group of automorphisms of $\rf$.

\begin{proposition}\label{pr:dilnaf}
The self-adjoint operator  $\omega$ is not affiliated to $\rf$.
\end{proposition}

\begin{proof}
  We will prove that
  \begin{equation}\label{eq:wdelta}
    \wlim_{|r|\to\infty}W(r\xi)^*(\omega+\rmi)^{-1}W(r\xi)=0
    \quad \forall \xi\in\Xi\setminus\{0\} .
  \end{equation}
  If $\omega$ is affiliated to $\rf$ then
  $(\omega+\rmi)^{-1}\in\rf$. Then by \eqref{eq:wdelta} and Corollary
  \ref{co:lim0} the operator $(\omega+\rmi)^{-1}$ is compact which is
  impossible because the spectrum of $\omega$ is purely absolutely
  continuous. It remains to prove \eqref{eq:wdelta}.  Note that
  \begin{equation}\label{eq:resdil}
    [\rmi(\omega+\rmi)^{-1}f](x)=\int_1^\infty f(\lambda x)
    \frac{\dd\lambda}{\lambda^{3/2}} .
  \end{equation}
  Indeed, since
  $\rmi(\omega+\rmi)^{-1}=\int_0^\infty\rme^{\rmi t\omega-t}\dd t$
  this follows by a change of variables in
\[
\rmi[(\omega+\rmi)^{-1}f](x)=
\int_0^\infty \rme^{-t} [\rme^{\rmi t\omega}f](x)  \dd t=
\int_0^\infty \e^{-t/2}f(\e^tx)  \dd t .
\]
Now let $\xi=(x,k)$. Then
\[
[W(\xi)^*\rmi(\omega+\rmi)^{-1}W(\xi)f](y)=
\int_1^\infty \e^{\rmi(\lambda-1)k(y-x)}
f(\lambda(y-x)+x)
  \frac{\dd\lambda}{\lambda^{3/2}} 
\]
hence
\[
|\braket{W(r\xi)g}{\rmi(\omega+\rmi)^{-1}W(r\xi)f}|\leq
\int_1^\infty \frac{\dd\lambda}{\lambda^{3/2}}
\left|\int_\R \dd y \, \e^{\rmi(\lambda-1)rky}
\bar{g}(y) f(\lambda y +(1-\lambda) rx) \right| .
\]
Let us denote $\ci_r$ the last integral over $y$. Then
\[
  \left|\ci_r\right|\leq \|g\|
  \left[\int_\R |f(\lambda y +(1-\lambda) rx)|^2 \dd y \right]^{1/2}
  =\lambda^{-1/2}\|g\|\|f\|
\]
We have to show that
$\braket{W(r\xi)g}{\rmi(\omega+\rmi)^{-1}W(r\xi)f}\to0$ if
$r\to\infty$ for any $f,g\in L^2(\R)$. Clearly it suffices to assume
$f,g$ continuous with compact support and then, by the preceding
estimate and the Lebesgue dominated convergence theorem, it suffices
to prove that $\ci_r\to0$ for any $\lambda>1$. Since $\xi\neq0$, if
$x=0$ then $k\neq0$ then $\ci_r\to0$ by the Riemann-Lebesgue lemma. If
$x\neq0$ and $r$ is large enough then
$\bar{g}(y) f(\lambda y +(1-\lambda) rx)=0\ \forall y$ hence
$\ci_r=0$. This proves the assertion. 
\end{proof}

\section{Subalgebras of \texorpdfstring{$\rf$}{F} and
  Hamiltonians affiliated to them} \label{s:LDsab}
\protect\setcounter{equation}{0} 
\renewcommand\thesubsubsection{\thesubsection.\arabic{subsubsection}}

\subsection{Half-Lagrangian decompositions and magnetic
  fields}\label{ss:HLdec}

The decompositions of a symplectic space considered here and the
representations associated to them will enable us to describe rather
explicitly certain subalgebras of interest of $\rf$.

Let $X$ be a finite dimensional real vector space, $X^*$ its dual, and
$\bracet{\cdot}{\cdot}:X\times X^*\to\R$ the duality form. Then the
Fourier transform%
\footnote{\ This definition and that of the scalar product in $L^2(X)$
  require the choice of a Lebesgue measure $\d x$ on $X$ but the
  choice is irrelevant in the construction of the observables and
  algebras of interest.}
$(F u)(k)=\int_X \e^{-\rmi\bracet{x}{k}}u(x)\d x$ induces an
isomorphism of $L^2(X)$ onto $L^2(X^*)$. The \emph{position and
  momentum observables} $q$ and $p$ are defined as follows: if
$\varphi$ and $\psi$ are (equivalences classes of) complex Borel
functions on $X$ and $X^*$ then $\varphi(q)=M_\varphi$ and
$\psi(p)=F^{-1}M_\psi F$ where $M_\varphi$ and $M_\psi$ are the
operators of multiplication by $\varphi$ in $L^2(X)$ and by $\psi$ in
$L^2(X^*)$. We set
\begin{equation}\label{eq:BK}
\mbfb(X)=B(L^2(X)) \quad\text{and}\quad\mathbf{K}(X)=K(L^2(X))
\end{equation}
and identify functions $\varphi\in L^\infty(X),\psi\in L^\infty(X^*)$
with operators $\varphi(q),\psi(p)\in\mbfb(X)$, so
\begin{equation}\label{eq:minbedI}
  C_\rmb(X)\subset L^\infty(X)\subset\mbfb(X) \quad\text{and}\quad
  C_\rmb(X^*)\subset L^\infty(X^*)\subset\mbfb(X).
\end{equation}
If $X$ is an Euclidean space $X^*$ is identified with $X$ as usual but
the spaces $L^\infty(X)$ and $L^\infty(X^*)$ are different when viewed
as subspaces of $\mbfb(X)$: the first one is the space of $\varphi(q)$
while the second is the space of $\varphi(p)$ with
$\varphi\in L^\infty(X)$. Or $L^\infty(X^*)=F^{-1}L^\infty(X)F$.

If $k\in X^*$ and $x\in X$ let $\bracet{q}{k}=\varphi(q)$ with
$\varphi=\bracet{\cdot}{k}$ and $\bracet{x}{p}=\psi(p)$ with
$\psi=\bracet{x}{\cdot}$. These are self-adjoint operators on $\ch$
such that $[k(q), x(p)]=\rmi \bracet{x}{k}$ and
\begin{equation}\label{eq:qpI} 
 (\e^{\rmi\bracet{q}{k}}f)(y)=\e^{\rmi\bracet{y}{k}}f(y),\quad
    (\e^{\rmi\bracet{x}{p}}f)(y)=f(x+y) .
\end{equation}
The \emph{phase space} of $X$ is the cotangent space
$T^*X=X\oplus X^*$ equipped with the symplectic form
$\sigma(\xi,\eta) = \bracet{y}{k}-\bracet{x}{l}$ if $\xi=x+k,\eta=y+l$
with $x,y\in X$ and $k,l\in X^*$, \conf \cite[\S21.1]{Hor2}.  A simple
modification of this symplectic form allows us to introduce constant
magnetic fields into the formalism and so to treat N-body systems
which interact with an external asymptotically constant magnetic field
(we refer to \cite{GL} for a detailed study of such systems). The
constant magnetic field may be interpreted as a bilinear
anti-symmetric form $\beta:X\times X\to\R$ and then clearly the next
relation defines a symplectic form on $T^*X$:
\begin{equation}\label{eq:LD}
  \sigma(\xi,\eta) = \beta(x,y) +\bracet{y}{k}-\bracet{x}{l} .
\end{equation}
In physical terms $X$ is the configuration space of a system, $T^*X$
is its phase space, and $\beta$ represents \emph{an external constant
  magnetic field}.  A $C^*$-algebraic approach to these questions,
including non constant magnetic fields, was proposed in \cite{GI2}.

We reconsider this construction in the symplectic setting. A finite
dimensional symplectic space $\Xi$ is even dimensional, say
$\dim\Xi=2n$, and if
$E,F$ are subspaces of $\Xi$:\\[1mm]
(1) an isotropic $E$ is Lagrangian if and only if $\dim E =n$;\\
(2) $\{E\mid E\text{ is Lagrangian}\}$ is a
smooth submanifold of $\mbg(\Xi) $ of dimension $n(n+1)/2$; \\
(3) $E$ is isotropic (coisotropic) $\Leftrightarrow$ $E$ is
contained in (contains) a Lagrangian subspace;\\
(4) $\dim E=n$ $\Rightarrow$ there is a Lagrangian $F$ such that
$\Xi=E\oplus F$;\\
(5) $E$ is symplectic $\Leftrightarrow E^\sigma$ is symplectic
$\Leftrightarrow \Xi=E+E^\sigma \Leftrightarrow \Xi=E\oplus E^\sigma$.

\begin{definition}\label{df:hlag}
  A \emph{half-Lagrangian decomposition} of $\Xi$ is a couple
  $(X,X^*)$ of subspaces of $\Xi$ such that $\dim X=\dim\Xi/2$, $X^*$
  is Lagrangian, and $X\cap X^* ={ 0}$. If $X$ is also Lagrangian
  we call $(X,X^*)$ a \emph{Lagrangian decomposition}.
\end{definition}

Under the conditions of the definition we have $\Xi=X\oplus X^*$ and
we usually say that $\Xi=X\oplus X^*$ is a half-Lagrangian or
Lagrangian decomposition of $\Xi$.  Let $\beta$ be the bilinear
anti-symmetric form on $X$ defined by the restriction
$\beta=\sigma|X\times X$.  If $\xi,\eta\in\Xi$ and $\xi=x+k,\eta=y+l$
are their decompositions in $X\oplus X^*$ and since $X^*$ is isotropic
\begin{equation}\label{eq:magnet}
  \sigma(\xi,\eta) = \beta(x,y)+\sigma(x,l)+\sigma(k,y).
\end{equation}
A half-Lagrangian decomposition is Lagrangian if and only if
$\beta=0$. 

Since $X^*$ is isotropic the map $k\ni X^*\mapsto\sigma(k,\cdot)$ is
an injective, hence bijective, map of $X^*$ onto the dual of $X$ and
we will identify these spaces by setting
\begin{equation}\label{eq:dualL}
  k(x)\equiv\bracet{x}{k}\doteq\sigma(k,x) \text{ for } x\in X
  \text{ and } k\in X^*.
\end{equation}
Then \eqref{eq:magnet} becomes \eqref{eq:LD}. So the choice of a
half-Lagrangian decomposition amounts to choosing the configuration
space $X$ of a system in an external constant magnetic field
$\beta$. The decomposition is Lagrangian if and only if the magnetic
field is zero.

If $W$ is any representation of $\Xi$ on a Hilbert space $\ch$, for
$x\in X,k\in X^*$ and we set
\begin{equation}\label{eq:wpxqk3}
    W^p_x=W(x,0) \text{ and } W^q_k=W(0,k).
\end{equation}
Then if $x,y\in X$ and $k,l\in X^*$ the CCR \eqref{eq:ccr1I} imply
\begin{align}
& W^p_{x+y}= \e^{\frac{\rmi}{2}\beta(x,y)} W^p_x W^p_y,\quad
W^q_{k+l}= W^q_kW^q_l, \quad
W^p_x W^q_k =\e^{\rmi\bracet{x}{k}} W^q_k W^p_x,  \label{eq:wpxqk2} \\
& W(\xi)= \e^{-\frac{\rmi}{2}\bracet{x}{k}} W^p_x W^q_k
= \e^{\frac{\rmi}{2}\bracet{x}{k}}W^q_kW^p_x\quad\text{where }\xi=x+k.
\label{eq:wpxqk1}  
\end{align}
The \emph{Schr\"odinger representation} associated to the preceding
half-Lagrangian decomposition is the irreducible representation of
$\Xi$ acting in $\ch=L^2(X)$ as follows:
\begin{equation}\label{eq:SchrobI}
  (W(\xi)f)(y) =
  \e^{\frac{\rmi}{2}\sigma(k,x)+\rmi\sigma(k+\frac{x}{2},y)}
  f(y+x) \quad \text{if}\quad \xi=x+k\in X\oplus X^*.
\end{equation}
Clearly this may be rewritten as
\begin{equation}\label{eq:wqpI}
  W(\xi)= \e^{\frac{\rmi}{2}(\beta(x,q)+\bracet{x}{k})}
  \e^{\rmi\bracet{q}{k}} \e^{\rmi \bracet{x}{p}}
  =  \e^{\frac{\rmi}{2}(\beta(x,q)-\bracet{x}{k})}
  \e^{\rmi\bracet{x}{p}} \e^{\rmi\bracet{q}{k}}.
\end{equation}
From $[\varphi(q),\bracet{x}{p}]=\rmi(x\varphi')(q)$ if $\varphi$ is
$C^1$ and $x\varphi'$ its derivative in the direction $x$, we get
\begin{equation}\label{eq:bp}
[\bracet{x}{p}, \beta(y,q)]=\rmi\beta(x,y),
\end{equation}
in particular $[\bracet{x}{p}, \beta(x,q)]=\rmi\beta(x,x)=0$. Thus, in
terms of \eqref{eq:wpxqk3}, we have
\begin{equation}\label{eq:wpxqk4}
  W^p_x=\e^{\rmi(\bracet{x}{p} + \frac{1}{2}\beta(x,q))}
  =\e^{\frac{\rmi}{2}\beta(x,q)}\e^{\rmi \bracet{x}{p}}
\text{ and } W^q_k=\e^{\rmi\bracet{q}{k}} .
\end{equation}
Moreover, the Baker-Campbell-Hausdorff formula gives
\begin{equation}\label{eq:wqpI1}
  W(\xi) = \e^{\rmi\phi(\xi)} \text{ with }
  \phi(\xi)= \bracet{q}{k}+\bracet{x}{p} + \half\beta(x,q).
\end{equation}
Here $\phi(\xi)$ is the \emph{field operator at the point
  $\xi=x+k\in\Xi$}.

We keep the notation $\beta$ for the map
$X\ni z\mapsto\beta(\cdot,z)\in X^*$; this is the unique linear map
$\beta:X\to X^*$ such that $\beta(x,y)=\bracet{x}{\beta y}$ for all
$x,y\in X$.  Then $\beta^*:X^{**}=X\to X^*$ and the anti-symmetry of
the bilinear form $\beta$ is equivalent to $\beta^*=-\beta$.  Clearly
$\beta(x,q)=\bracet{x}{\beta q}$ where $\beta q$ is the map $X\to X^*$
defined by $(\beta q)(y)=\beta y=\beta(\cdot,y)$.  Thus, if we set
$p_\beta=p+\frac{1}{2}\beta q$, then the field operator may be written
$\phi(\xi)=\bracet{q}{k}+\bracet{x}{p_\beta}$.

We will now prove that the Hamiltonian of a non-relativistic system of
particles in a constant magnetic field is affiliated to $\rf$.  With
the help of the perturbative criterion of Theorem \ref{th:afgrad} one
may then show that Hamiltonians of $N$-body systems with singular
$n$-body interactions ($1\leq n\leq N$) and asymptotically constant
external magnetic fields are affiliated to $\rf$. The arguments are
similar to those of \S\ref{ss:nbb} but we do not develop this topic
(this could be done along the lines of \cite[\S4]{GI2}).  The next
lemma is valid in any irreducible representation of a finite
dimensional symplectic space $\Xi$.

\begin{lemma} \label{lm:aff} Let $\eta=\{\eta_1,\dots,\eta_n\}$
  be a generating set for the subspace $E\in\mbg(\Xi)$. Then the form
  sum $\Delta_\eta=\sum_{k=1}^n\phi(\eta_k)^2$ is a self-adjoint
  operator affiliated to $\rf(E)$.
\end{lemma}

\begin{proof}
  We show that $T=(\Delta_\eta+1)^{-1}\in\rf(E)$ by checking the three
  conditions of Theorem \ref{th:mainthfd}.  Note that the form domain
  of $\Delta_\eta$ is $\ck=\ccap_k\dom(\phi(\e_k))$.  We embed as
  usual $\ck\subset\ch=\ch^*\subset\ck^*$.  By \eqref{eq:ccr1I} the
  operators $W(\xi)$ leave invariant $\ck$ so extend to bounded
  operators in $\ck^*$. Below we consider $T$ as an operator
  $\ch\rarrow\ck$ and $\ck^*\rarrow\ch$ (by taking adjoints) and
  $\Delta_\eta$ as operator $\ck\rarrow\ck^*$. Clearly
  $[W(\xi),T]=T[\Delta_\eta,W(\xi)]T$ and
  \begin{align*}
    [\Delta_\eta,W(\xi)]&=
    W(\xi)\tsum{k}{}
    \big[\big(\phi(\eta_k)+\sigma(\eta_k,\xi)\big)^2
     -\phi(\eta_k)^2\big]\\ 
    &=W(\xi) \tsum{k}{}\Big[2\phi(\eta_k)\sigma(\eta_k,\xi)
      +\sigma(\eta_k,\xi)^2\big]  
  \end{align*}
  If $\xi\in E^\sigma$ the last term is zero, so the last condition of
  Theorem \ref{eq:ccr1I} is satisfied. Then
  \[
    \Vert[W(\xi),T]\Vert\leq\tsum{k}{}\Big[2|\sigma(\eta_k,\xi)|\,\Vert
    \phi(\eta_k)T\Vert +\sigma(\eta_k,\xi)^2\Big],
  \]
  hence the first condition of Theorem \ref{th:mainthfd} is satisfied
  too.  Now observe that if $\xi=\sum_{k=1}^n\xi_k$ an induction
  argument gives
  \[
    W(\xi)=\exp\big({\textstyle
      \sum_{j<k}}\sigma(\xi_j,\xi_k)/2\rmi\big) W(\xi_1)\dots
    W(\xi_n).
  \]
  If $\xi\in E$ then we can write it as $\xi=\sum t_k \eta_k$ with real
  $t_k$. By using the preceding formula we see that in order to prove
  the second condition of Theorem \ref{th:mainthfd} it suffices to
  show that $\Vert[W(t \eta_k)-1]T\Vert\rarrow0$ when $t\rarrow 0$ for
  each $k$.  But
\[
\Vert \big(W(t \eta_k)-1\big)T\Vert \leq
\Vert \big(\e^{\rmi t\phi(\eta_k)}-1\big)
\big( \phi(\eta_k)^2+1\big)^{-1/2}\Vert
\,\Vert\big(\phi(\eta_k)^2+1\big)^{1/2}T\Vert
\]
so the assertion follows immediately from the estimate
$\phi(\eta_k)^2\leq \Delta_\eta$.
\end{proof}

The case of interest here is $E=X$; recall that $\Xi=X\oplus X^*$ is a
half-Lagrangian decomposition.  If $\eta$ is a basis of $X$ then
$\Delta_\eta$ is the ``free'' Hamiltonian of a system with configuration
space $X$ in the constant magnetic field $\beta=\sigma|X\times
X$. More general such ``free'' Hamiltonians are the self-adjoint
operators affiliated to the $C^*$-algebra $\rf(X)$ which, by
\eqref{eq:wpxqk3} and \eqref{eq:wpxqk4}, is the norm closure in
$\mbfb(X)$ of the set of operators
\[
  W(\mu)=\int_{X}W_x^p u(x)\d x =\int_{X}\e^{\rmi\bracet{x}{p_\beta}} u(x)\d x
  \text{ with } u\in L^1(X).
\]
We could write this as $\hat{u}(p_\beta)$ and think of the elements of
$\rf(X)$ as functions of class $C_0$ of $p_\beta$, but this is rather
formal because the components of the vector-operator $p_\beta$ do not
commute: if $\beta\neq0$ then $\{W_x^p\}_{x\in X}$ is only a
projective representation of $X$, \conf \eqref{eq:wpxqk2}.

\begin{example}\label{ex:cross}{\rm The simplest physically
    interesting example is a non-relativistic particle with
    configuration space $X=\R^3$ interacting with the constant
    magnetic field $b\in X$. If $x,y\in X$ let $x\times y$ be their
    cross product and define the bilinear anti-symmetric form
    $\beta:X\times X\to\R$ by $\beta(x,y)=\bracet{x\times y}{b}$. The
    symplectic space will be the phase space $\Xi=X\oplus X$ equipped
    with the symplectic form \eqref{eq:LD}, \ie
    \begin{equation}\label{eq:cross}
      \sigma(\xi,\eta) = \bracet{x\times y}{b}
      +\bracet{y}{k}-\bracet{x}{l} .
    \end{equation}
    Since
    $\beta(x,y)=\bracet{x\times y}{b}=\bracet{x}{y\times b}$ we
    get $\beta q=q\times b$ hence $p_\beta=p+\frac{1}{2}q\times b$.
    From \eqref{eq:wqpI1} we get for $\xi=(e,0)\equiv e\in X$
    \[ \phi(e)= \bracet{e}{p} + \half\beta(e,q) = \bracet{e}{p} +
      \half\bracet{e\times q}{b} = \bracet{e}{p} +
      \half\bracet{e}{q\times b} \] If $e_1,e_2,e_3$ the natural basis
    of $X=\R^3$ then $\bracet{e_j}{q\times b}=(q\times b)_j$ and
    \[
      (q\times b)_1=q_2b_3-q_3b_2,\quad (q\times b)_2=q_3b_1-q_1b_3,
      \quad (q\times b)_3=q_1b_2-q_2b_1
    \]
    hence $\phi(e_j)=p_j+\half(q_kb_l-q_lb_k)$ with usual notation
    rules. Then 
    \[
      \tsum{}{} \phi(e_j)^2=
      (p_1+\half(q_2b_3-q_3b_2)^2
      +p_2+\half(q_3b_1-q_1b_3))^2
      +(p_3+\half(q_1b_2-q_2b_1))^2.
    \]
    Recall that the vector field $a:X\to X$ defined by
    $a(x)=\frac{1}{2}b\times x=-\frac{1}{2}x\times b$ is the magnetic
    vector potential, \ie $\nabla\times a=b$. Thus $p_\beta=p-a$ and
    \[
      \tsum{j=1}{3} \phi(e_j)^2=(p-a)^2= H 
    \]    
    which is the Hamiltonian of a non-relativistic particle with mass
    $2$ and charge $1$ interacting with the constant magnetic field
    $b\in X$. Also  $H=p_\beta^2$   with the usual rule that the
    square of $p_\beta$ is the sum of the squares of its components
    ($p_\beta$ is the magnetic momentum). 

  }\end{example}

\begin{remark}\label{re:cross}{\rm

    In the preceding example we worked with the cotangent space
    $T^*X= X\oplus X$ but modified its symplectic form, \conf
    \eqref{eq:cross}.  Since the field algebras associated to
    symplectic spaces of dimension $6$ are isomorphic, we should be
    able to work in $T^*X$ with the standard symplectic form
    \eqref{eq:standardLD} (\ie \eqref{eq:LD} with $\beta$ replaced by
    $0$) and this is indeed possible. Now the field operator at the
    point $\xi=(x,k)$ is $\phi(x,k)=\bracet{q}{k}+\bracet{x}{p}$. We
    take $\eta_j=(\half b\times e_j,e_j)\in T^*X$ with $j=1,2,3$ and
    denote $E$  the subspace they generate. Clearly
    $\phi(\eta_j)=p_j+\half(q_kb_l-q_lb_k)$ which are the same
    operators as in Example \ref{ex:cross} hence 
    $H=\sum_{j=1}^{3} \phi(\eta_j)^2=(p-a)^2$ as before but this time
    $H$ is affiliated to $\rf(E)$. 
    
  }\end{remark}

\subsection{Subalgebras of
  \texorpdfstring{$\rf$}{F}} \label{ss:gcsab} 

The following examples of $C^*$-subalgebras of the field algebra play
a role in a generalized version of the N-body problem. By Remark
\ref{re:isonot} it suffices to work in an irreducible representation
of $\Xi$.  The symplectic space is finite dimensional and the
subalgebras that we construct depend on a Lagrangian decomposition
$\Xi=X\oplus X^*$ of $\Xi$, \conf Definition \ref{df:hlag}. By
assertions (2) and (4) in \S\ref{ss:HLdec} there are infinitely many
such choices and the algebras constructed below depend on the
choices.

We fix a Lagrangian decomposition $\Xi=T^*X$ and we work in the
Schr\"odinger representation associated to it, see \S\ref{ss:HLdec} for
framework and notations of. Thus $\ch=L^2(X)$, the field operator at
the point $\xi=(x,k)$ is $\phi(\xi) =\bracet{x}{p}+\bracet{q}{k}$, and
the $C^*$-algebra generated by these operators is the field algebra
$\rf\subset \mbfb(X)=B(L^2(X))$.

Recall that by \eqref{eq:minbedI} we identify $\cbu(X)$ and $\cbu(X^*)$
with the $C^*$-subalgebras of $\mbfb(X)$ consisting of the operators
$\varphi(q)$ and $\psi(p)$ respectively, with $\varphi\in\cbu(X)$ and
$\psi\in\cbu(X^*)$.

If $Y\in\mbg(X)$ then $C_0(X/Y)\subset \cbu(X)$ as in \S\ref{ss:gc0}.
If $Z\in\mbg(X^*)$ then we also have $C_0(Z^*)\subset\cbu(X)$ because
$Z^*=X/Z^\perp$ with
$Z^\perp=\{x\in X \mid \bracet{x}{k}=0\ \forall k\in Z\}$. Thus
\begin{equation}\label{eq:imbed}
C_0(X/Y),C_0(Z^*) \subset\cbu(X)\subset\mbfb(X). 
\end{equation}
On the other hand, if
$Y^\perp=\{k\in X^*\mid \bracet{x}{k}=0\ \forall x\in Y\}$ then
$Y^*= X^*/Y^\perp$. Hence
\begin{equation}\label{eq:inbed}
C_0(X^*/Z),C_0(Y^*) \subset\cbu(X^*)\subset\mbfb(X).
\end{equation}
The Grassmann $C^*$-algebras $\cg_\ssx$ and $\cg_{\ss{X^*}}$ defined
in \eqref{eq:cgx1} are embedded in $\mbfb(X)$ as above.

We may now give ``explicit'' descriptions of some components of $\rf$.

\begin{proposition}\label{pr:rfyz}
  If $Y\in\mbg(X)$ and $Z\in\mbg(X^*)$ then
  \begin{equation}\label{eq:rfyz}
  \rf(Y)=C_0(Y^*), \ \rf(Z)=C_0(X/Z^\perp), \
  \rf(Y+Z)=C_0(Y^*)\cdot C_0(X/Z^\perp) . 
  \end{equation}
In particular $\rf(X)=C_0(X^*)$, $\rf(X^*)=C_0(X)$, and for any
$Y\subset X$ 
\begin{equation}\label{eq:rfysigma}
  \rf(Y^\perp)=C_0(X/Y)\quad\text{and}\quad
  \rf(Y^\sigma) =C_0(X/Y)\cdot C_0(X^*).
\end{equation}
Moreover 
\begin{align}
  \rf_\ssx & =\tsum{Y\in\mbg(X)}{\rmc}\rf(Y)
          =\tsum{Y\in\mbg(X)}{\rmc} C_0(Y^*),\label{eq:rfx1}\\
  \rf_{\ss{X^*}} & =\tsum{Z\in\mbg(X^*)}{\rmc}\rf(Z)
              =\tsum{Y\in\mbg(X)}{\rmc} C_0(X/Y),\label{eq:rfx2}\\
  \rf_{\ss{\supset X}}& =\tsum{Y\in\mbg(X)}{\rmc}\rf(Y^\sigma)
                   =\tsum{Y\in\mbg(X)}{\rmc}C_0(X/Y)\cdot C_0(X^*).
              \label{eq:rfXsygma}
\end{align}
\end{proposition}

\begin{proof}
  
  $\rf(Y)$ is the norm closure of the set of operators
  $W(\mu)=\int_{Y}W(y) u(y)\d y$ with $u\in L^1(Y)$, \conf Definition
  \ref{df:wkfield}.  Then \eqref{eq:wqpI} gives
  $W(\mu)=\int_{Y} \e^{\frac{\rmi}{2}\bracet{y}{p}} u(y)\d
  y=c\hat{u}(p)$ where $c$ is a constant depending on the definition
  of the Fourier transform. Since
  $\hat u\in C_0(Y^*)\subset \cbu(X^*)$ we get $\rf(Y)=C_0(Y^*)$.  By
  the same definition and with the notation \eqref{eq:clinspan} we
  have
  $\rf(Y)= C^*(\bracet{y_1}{p})\cdot C^*(\bracet{y_2}{p})
  \cdot\ldots\cdot C^*(\bracet{y_n}{p})$ where $y_1,\dots,y_n$ is a
  basis of $Y$, which gives a second proof of the same result.  If
  $Z\subset X^*$ the computation of $\rf(Z)$ is similar: if we take
  $\xi=(x,k)=(0,z)$ in \eqref{eq:wqpI} we get
  $W(\xi) = \e^{\rmi\bracet{q}{z}}$ hence if $u\in L^1(Y)$ then
  $W(\mu)=\int_Y \e^{\rmi\bracet{q}{z}} u(y)\d y=c\hat{u}(q)$ where
  $\hat{u}\in C_0(Z^*)$ and $Z^*=X/Z^\perp$, \conf the comments before
  the statement of the proposition. Then by using \eqref{eq:product}
  and we get
  \[
    \rf(Y+Z)=\rf(Y)\cdot\rf(Z)=\rf(Z)\cdot\rf(Y)
    =C_0(X/Z^\perp)\cdot C_0(Y^*) .
  \]
  The first relation in \eqref{eq:rfysigma} follows from the second
  one in \eqref{eq:rfyz} by taking $Z=Y^\perp$ since
  $(Y^\perp)^\perp=Y$. The second relation in \eqref{eq:rfysigma} is a
  consequence of the third one in \eqref{eq:rfyz} since
  $Y^\sigma=X+Y^\perp$ but it is instructive to give a direct proof.
  $\rf(Y^\sigma)$ is the closure of the set of operators
  $W(\mu)=\int_{Y^\sigma}W(\xi) \mu(\d\xi)$ with
  $\mu\in L^1(Y^\sigma)$. If
  $\mu(\xi)=\e^{\frac{\rmi}{2}\bracet{x}{k}} u(x)v(k)$ with
  $u\in C_\rmc(X)$ and $v\in C_\rmc(Y^\perp)$ then by \eqref{eq:wqpI}
\begin{align*}
  W(\mu) &=\int_X\int_{Y^\perp}W(x,k)
           \e^{\frac{\rmi}{2}\bracet{x}{k}}u(x)v(k) \d x\d k \\
  &=\int_X\int_{Y^\perp} \e^{\rmi\bracet{x}{p}}\e^{\rmi\bracet{q}{k}}
    u(x)v(k) 
    \d x\d k =c \hat{u}(-p)\hat{v}(-q)
\end{align*}
where $\hat{u}\in C_0(X^*)$ and $\hat{v}\in C_0(X/Y)$ are the Fourier
transforms of $u$ and $v$, we identify $Y^\perp=(X/Y)^*$, and $c$ is
a constant depending on the normalization chosen in the definition of
the Fourier transforms. Thus we have
\[
  \hat{u}(-p)\hat{v}(-q) \in C_0(X^*)\cdot C_0(X/Y)= C_0(X/Y)\cdot
  C_0(X^*).
\]
This implies $\rf(Y^\sigma)=C_0(X/Y)\cdot C_0(X^*)$ because the linear
space generated by the functions $\mu$ is dense in
$L^1(Y^\sigma)$. Finally, $Z\mapsto Z^\perp$ being a bijective map
$\mbg(X^*)\to\mbg(X)$, \eqref{eq:rfyz} implies \eqref{eq:rfx1} and
\eqref{eq:rfx2}.
\end{proof}

\begin{remark}\label{re:rfxcg}{\rm
In terms of the Grassmannian $C^*$-algebras introduced in
\S\ref{ss:gc0}, we have
\begin{equation}\label{eq:rfx12}
  \rf_\ssx=\cg_{\ss{X^*}},\quad \rf_{\ss{X^*}}=\cg_{\ssx},\quad
  \rf_{\ss{\supset X}}=\cg_\ssx\cdot C_0(X^*)=\cg_\ssx\rtimes X
\end{equation}
where $\cg_\ssx\rtimes X$ is the crossed product of $\cg_\ssx$ by the
action of $X$, \conf \cite{GI5}. On the other hand,
\eqref{eq:fephi} implies $\rf_{\ssx}=C^*(\phi(\xi)\mid \xi\in X)$ and 
$\xi\in X$ means $\xi=(x,0)$ with $x\in X$ so
$\phi(\xi) =\bracet{x}{p}+\bracet{q}{k}=\bracet{x}{p}$.
With a similar argument in the case of $X^*$, we get
\begin{equation}\label{eq:gensetF}
  \rf_\ssx =C^*(\bracet{x}{p}\mid x\in X) \quad\text{and}\quad
  \rf_{\ss{X^*}}=C^*(\bracet{q}{k}\mid k\in X^*).
\end{equation}
This and \eqref{eq:rfx12} give the first assertion of Proposition
\ref{pr:phialgA}.
}\end{remark}

\begin{proposition}\label{pr:fctrf}
  $L^\infty(X)\cap\rf=\rf_{\ss{X^*}}=\cg_\ssx$ and
  $L^\infty(X^*)\cap\rf=\rf_{\ssx}=\cg_{\ss{X^*}}$.
\end{proposition}

\begin{proof}
  Note that we use the embeddings \eqref{eq:minbedI}.  The inclusion
  $L^\infty(X)\cap\rf\supset\cg_\ssx$ is obvious. And
  $L^\infty(X)\cap\rf\subset\rf_{\ss{X^*}}^\com=\rf_{\ss{X^*}}=\cg_\ssx$
  by \eqref{eq:rfx2} and Corollary \ref{co:comutant}.
\end{proof}

The algebras $\rf(Y^\sigma)$ with $Y\in\mbg(X)$ are components of the
algebras generated by N-body type Hamiltonians, \conf \S\ref{ss:nbb},
hence it is useful to have ``explicit'' descriptions of them. Below we
give two such descriptions, consequences of Theorem \ref{th:mainthfd}
and Proposition \ref{pr:rfyz}.

\begin{proposition}\label{pr:rfYsigma1}
  If $Y\subset X$ and $T\in \mbfb(X)$ then $T\in\rf(Y^\sigma)$
  if and only if 
  \begin{compactenum}
  \item[{\rm(i)}]  $[\e^{\rmi\bracet{x}{p}},T] = 0$ for all $x\in Y$,
   \smallskip 
 \item[{\rm(ii)}]
  ${\displaystyle\lim_{x\to0}}\|[\e^{\rmi\bracet{x}{p}},T]\|=0$ and
  ${\displaystyle\lim_{x\to0}}\|(\e^{\rmi\bracet{x}{p}}-1)T\|=0$,
  \smallskip
\item[{\rm(iii)}]
     ${\displaystyle\lim_{k\to0}}\|[\e^{\rmi\bracet{q}{k}},T]\|=0$ and
     ${\displaystyle\lim_{k\to0,k\in Y^\perp}}
     \|(\e^{\rmi\bracet{q}{k}}-1)T\|=0$. 
\end{compactenum}
\end{proposition}

\begin{remark}\label{re:condii}{\rm The second condition in (ii) is
    equivalent to $T=\varphi(p)T_0$ for some $\varphi\in C_0(X^*)$ and
    $T_0\in\mbfb(X)$, \conf \cite[Lemma 3.8]{GI4}.  }\end{remark}

In the next proposition $F_Y$ is the Fourier transform associated to
$Y$ and we keep the notation $F_Y$ for
$F_Y\otimes1: L^2(Y)\otimes L^2(Y')\to L^2(Y^*)\otimes L^2(Y')\equiv
L^2(Y^*;L^2(Y'))$.

\begin{proposition}\label{pr:rfYsigma2}
  Let $Y,Y'$ be subspaces of $X$ such that $X=Y\oplus Y'$ and let us
  identify $L^2(X)=L^2(Y)\otimes L^2(Y')$. Then
  \begin{equation}\label{eq:rfYsigma2}
    \rf(Y^\sigma)= C_0(Y^*)\otimes \mbfk(Y').
  \end{equation}
  Thus $T\in \mbfb(X)$ belongs to $\rf(Y^\sigma)$ if and only if
  $F_YTF_Y^{-1}$ is the operator of multiplication by a function
  $\what T\in C_0(Y^*;\mbfk(Y'))$ in $L^2(Y^*;L^2(Y'))$.
\end{proposition}

\begin{proof}
  We identify $X/Y=Y'$ hence
  $C_0(X^*)=C_0(Y^*\oplus Y'^*)= C_0(Y^*)\otimes C_0(Y'^*)$ and
  $C_0(Y')\cdot C_0(Y'^*)=\mbfk(Y')$. Then by \eqref{eq:rfysigma}
  \begin{equation*}
    \rf(Y^\sigma)= C_0(Y')\cdot [C_0(Y^*)\otimes C_0(Y'^{*})]
    = C_0(Y^*)\otimes \mbfk(Y'). \qedhere
  \end{equation*}
\end{proof}

\subsection{N-body Hamiltonians and beyond}\label{ss:nbb}

Hamiltonians generalizing N-body Hamiltonians are defined
by their affiliation to graded $C^*$-subalgebras $\rf(\cs)$ where
$\cs\subset\mbg(\Xi)$ is a subsemilattice.  The case of finite $\cs$
is easy but important because large classes of N-body type
Hamiltonians are affiliated to such algebras. By (3) of Theorem
\ref{th:remarc} and (2) of Lemma \ref{lm:subsemi}, if $\cs$ is a
finite subset of $\mbg(\Xi)$ then $\rf(\cs)$ is a subalgebra of $\rf$
if and only if $\cs$ is a subsemilattice and then $\rf(\cs)$ is an
$\cs$-graded $C^*$-subalgebra of $\rf$. However, we are mostly
interested in the case of infinite $\cs$.

 In this section we fix a Lagrangian decomposition
$\Xi=X\oplus X^*$ and consider semilattices of the form
$\cs=\ct^\sigma$ with $\ct\subset\mbg(X)$ stable under intersections. 
We recall the notations
\[
  \ct\subset\mbg(X) \Rightarrow \ct^\perp\doteq
  \{Y^\perp\mid  Y\in\ct\}\subset\mbg(X^*) \text{ and }
  \ct^\sigma\doteq \{Y^\sigma\mid  Y\in\ct\}
  \subset\mbg_{\ss{\supset X}}(\Xi).
\]  
Since $Y\mapsto Y^\perp$ and $Y\mapsto Y^\sigma$ are bijective maps
$\mbg(X)\to\mbg(X^*)$ and $\mbg(X)\to\mbg_{\ss{\supset X}}(\Xi)$ which
reverse the order, we get
\[
  \ct\,\text{is}\,\cap\!\text{-stable} \Leftrightarrow
  \ct^\perp\text{is a subsemilattice of}\,\mbg(X^*)
  \Leftrightarrow
  \ct^\sigma\text{is a subsemilattice of}\,\mbg_{\ss{\supset X}}(\Xi).
\]
Recall that $\cg_\ssx$ is the Grassmann C*-algebra, \conf
\S\ref{ss:gc0}.  Taking into account \eqref{eq:rfysigma}, we have
\[
  \mathring\rf(\ct^\perp)=\mathring\cg_\ssx(\ct)\doteq
  \tsum{Y\in\ct}{}C_0(X/Y)    \text{ and }
\rf(\ct^\perp)=\cg_\ssx(\ct)\doteq \tsum{Y\in\ct}{\rmc}C_0(X/Y).
\]
These are stable under translations and conjugation subspace of
$\cbu(X)$ and if $\ct$ is finite then $\mathring\cg_\ssx(\ct)$ is
closed, so $\mathring\cg_\ssx(\ct)=\cg_\ssx(\ct)$. And
$\cg_\ssx(\{Y\})\equiv\cg_\ssx(Y)=C_0(X/Y)$.

Our first goal is to show that $\rf(\cs)$ is the $C^*$-algebra
generated by Hamiltonians of N-body type.  Clearly:

\begin{lemma}\label{lm:czerot}
  Let $\ct_0\subset\mbg(X)$ and let $\ct$ be the set of intersections
  of subspaces from $\ct_0$.  Then the unital $C^*$-subalgebra of
  $\cbu(X)$ generated by $\cg_\ssx(\ct_0)$ is $\cg_\ssx(\ct)$.
\end{lemma}

\begin{remark}\label{re:czerot}

  {\rm An intersection of subspaces of $\ct_0$ is a subspace of the
    form $\cap_{i\in I}Y_i$ with $Y_i\in\ct_0$. Since $X$ is finite
    dimensional we may assume $I$ finite, but it could be empty and
    then $\cap_{i\in \emptyset}Y_i=X$, hence we always have
    $X\in\ct$. Thus $\C=C_0(X/X)\subset \cg_\ssx(\ct)$.}

\end{remark}

We refer to \cite{GI3, GI5 ,GI4} for a presentation of the notion of
crossed product of a $C^*$-algebra by the action of a group adapted to
our context (see also Proposition \ref{pr:normcont}).  If
$\ca\subset\cbu(X)$ is a $C^*$-subalgebra stable under translations
then the norm closed linear space $\ca\cdot C_0(X^*)$ of operators on
$L^2(X)$ generated by the products $\varphi(q)\psi(p)$ with
$\varphi\in\ca$ and $\psi\in C_0(X^*)$ is a $C^*$-algebra canonically
isomorphic to the crossed product $\ca\rtimes X$, \conf \cite[Th.\
2.17]{GI5}.  We identify $\ca\rtimes X=\ca\cdot C_0(X^*)$ so one may
take this as definition of the crossed product.

A function $h:X^*\to\R$ is called \emph{divergent} if
$\lim_{k\to\infty}h(k)=+\infty$.

\begin{theorem}\label{th:nbham}
  Let $\ct_0\subset\mbg(X)$ and $\ct$ the set of intersections of
  subspaces from $\ct_0$. Then $\cs=\ct^\sigma$ is a subsemilattice
  of $\mbg_{\ss{\supset X}}(\Xi)$ and
  \begin{equation}\label{eq:nbham}
    \rf(S)=\cg_\ssx(\ct)\rtimes X = \cg_\ssx(\ct)\cdot C_0(X^*) . 
  \end{equation}
  If $h:X^*\to\R$ is continuous and divergent $\rf(\cs)$ coincides
  with the $C^*$-algebra generated by the self-adjoint operators
  $H=h(p+k)+v(q)$ with $k\in X^*$ and $v\in\mathring\cg_\ssx(\ct_0)$
  real.
\end{theorem}

\begin{proof}
  By \eqref{eq:fcs} and \eqref{eq:rfysigma} we have
  \[
    \rf(\cs)=\tsum{E\in\cs}{\rmc}\rf(E)
    =\tsum{Y\in\ct}{\rmc}\rf(Y^\sigma)
    =\tsum{Y\in\ct}{\rmc} C_0(X/Y)\cdot C_0(X^*)
    =\cg_\ssx(\ct)\rtimes X
  \]
  which proves \eqref{eq:nbham}. By Remark \ref{re:czerot} we have
  $C_0(X^*)\subset\rf(\cs)$ hence for an arbitrary $h:X^*\to\R$
  continuous and divergent the self-adjoint operator $h(p)$ is
  affiliated to $C_0(X^*)$ hence to $\rf(\cs)$. Then if
  $v\in \cg_\ssx(\ct)$ real $H=h(p)+v(q)$ is a self-adjoint operator
  and for large enough negative $z$ we have a norm convergent
  development
  \[
    (z-H)^{-1}=\tsum{n\geq}{}
    (z-h(p))^{-1}\big(v(q)(z-h(p))^{-1}\big)^n.
  \]
  But $(z-h(p))^{-1}\in C_0(X^*)$ hence
  $v(q)(z-h(p))^{-1}\in \cg_\ssx(\ct)\cdot C_0(X^*)$ and $(z-H)^{-1}$
  belongs to $\rf(\cs)$. Thus any $h(p)+v(q)$ is affiliated to
  $\rf(\cs)$ and it remains to prove that for any fixed $h$ the
  algebra $\rf(\cs)$ is equal to the algebra $\rc$ generated by
  operators of the form $H=h(p+k)+v(q)$ with $k\in X^*$ and
  $v\in \mathring\cg_\ssx(\ct_0)$ real. Clearly we get the same $\rc$
  if we allow $v\in \cg_\ssx(\ct_0)$ and the preceding argument gives
  $\rc\subset\rf(\cs)$.
  
  We will follow now the proof of \cite[Pr.\ 4.1]{GI3} with
  $\ca=\cg_\ssx(\ct_0)$. Then $\ca$ is not necessarily a
  $C^*$-subalgebra but is a closed stable under translations and
  conjugation subspace of $\cbu(X)$.  Then the quoted proof gives
  $v(q)\psi(p)\in \rc$ for all $v\in \cg_\ssx(\ct_0)$ and
  $\psi\in C_0(X^*)$. In other terms,
  $\cg_\ssx(\ct_0)\cdot C_0(X^*)\subset\rc$ which in fact is
  equivalent to $\cg_\ssx(Y)\cdot C_0(X^*)\subset\rc$ for all
  $Y\in\ct_0$. Observe that
  $\cg_\ssx(Y)\cdot C_0(X^*)= C_0(X^*)\cdot\cg_\ssx(Y)$ because this
  is the $C^*$-algebra $\rf(Y^\perp)$.  Then for any $Y,Z\in\ct_0$
\begin{align*}
  &[\cg_\ssx(Y)\cdot C_0(X^*)\big]\cdot [\cg_\ssx(Z)\cdot C_0(X^*)]=
  C_0(X^*)\cdot\cg_\ssx(Y)\cdot\cg_\ssx(Z)\cdot C_0(X^*)=\\
  &C_0(X^*)\cdot \cg_\ssx(Y\cap Z) \cdot C_0(X^*)=
  \cg_\ssx(Y\cap Z)\cdot C_0(X^*)
\end{align*}
hence $\cg_\ssx(Y\cap Z)\cdot C_0(X^*)\subset\rc$. Thus
$\cg_\ssx(Y)\cdot C_0(X^*)\subset\rc$ for all $Y\in\ct$ which is clearly
equivalent to $\rf(\cs)=\cg_\ssx(\ct)\cdot C_0(X^*)\subset\rc$.
\end{proof}

\begin{corollary}\label{co::nbham}
  If $h:X^*\to\R$ is continuous and divergent then
  \[
    \rf_{\ss{\supset X}}=C^*\big(h(p+k)+v(q) \mid k\in X^*
    \text{ and } v\in \mathring\cg_\ssx \text{ real} \big)
  \]
\end{corollary}

\begin{remark}\label{re:mbfp}{\rm If
    $\cs\subset\mbg_{\ss{\supset X}}(\Xi)$ is a subsemilattice then 
    $\rf(\cs)\cdot C_\rmb(X^*)=\rf(\cs)$ by \eqref{eq:nbham}.
    Thus if $T\in\rf(\cs)$ and $\psi\in C_\rmb(X^*)$ then $\psi(p)T$
    and  $T\psi(p)$ belong to $\rf(\cs)$.}\end{remark} 

We will apply this proposition in the case of a system of $N$
particles moving in $\R^\nu$. Then $X=\R^{N\nu}$, we write
$x=(x_1,\dots,x_N)$ with $x_i\in\R^\nu$, and denote $q_i,p_i$ the
momentum and position observables of the $i$-th particle, hence
$q=(q_1,\dots,q_N)$, $p=(p_1,\dots,p_N)$ (if $\nu>1$ this is not the
notation from \S\ref{ss:i22}).  If the particles interact only via
one-body and two-body forces the Hamiltonian is \label{p:stnb}
\begin{equation}\label{eq:NB01}
  H=\tsum{i=1}{N}h_i(p_i-k_i) +\tsum{i=1}{N} v_i(q_i)
  +\tsum{1\leq i<j\leq N}{} v_{ij}(q_i-q_j)  
\end{equation}
where $h_1,\dots,h_N$ are real continuous divergent functions on
$\R^\nu$ and $h_i(p_i-k_i)$ is interpreted as the kinetic energy of
the i-th particle if we take $k_i$ as origin in the momentum space.
The potentials $v_i,v_{ij}$ are real functions in $C_0(\R^\nu)$ but we
may write them in the form required in Proposition \ref{pr:nbham1} as
follows.  Define $\rho_i, \rho_{ij}:X\to \R^\nu$ by $\rho_i(x)=x_i$
and $\rho_{ij}(x)=x_j-x_i$ for $i<j$.  These are linear surjective
maps with $\ker\rho_i=X_i=\{x\mid x_i=0\}$ and
$\ker\rho_{ij}=X_{ij}=\{x\mid x_i=x_j\}$, so $\rho_i$ and $\rho_{ij}$
induce bijective linear maps $X/X_i\to\R^\nu$ and
$X/X_{ij}\to\R^\nu$. Then $v_i(q_i)=v_i\circ\rho_i(q)$ and
$v_{ij}(q_i-q_j)=v_{ij}\circ\rho_{ij}(q)$ hence by keeping the
notation $v_i$ for $v_i\circ\rho_i$ and $v_{ij}$ for
$v_{ij}\circ\rho_{ij}$ the potential part in \eqref{eq:NB01} becomes
$\sum v_i(q) +\sum v_{ij}(q)$ where $v_i\in C_0(X/X_i)$ and
$v_{ij}\in C_0(X/X_{ij})$.  Then the $C^*$-algebra generated by the
Hamiltonians \eqref{eq:NB01} is given by Proposition \ref{pr:nbham1}:
$\ct_0$ is the set of subspaces $X_i,X_{ij}$ hence it suffices to
compute the set of intersections of such subspaces. This is an
exercice, we describe the result below.

Let $[N]=\{1,2,\dots,N\}$. A nonempty subset of $[N]$ is called
\emph{cluster}.  A \emph{sub-partition} of $[N]$ is a (possibly
empty) set of pairwise disjoint clusters; if $a$ is a sub-partition
then $\cup a$ is the union of its clusters.  We denote $\pi$ the set
of sub-partitions.  If $a\in\pi$ and $i,j\in [N]$ we write
$i \stackrel{a}{\sim} j$ if $i,j$ belong to the same cluster of $a$
and we define
\begin{equation}
  X_a=\{x\in X \mid x_i=x_j \text{ if } i\stackrel{a}{\sim}j
  \text{ and } x_i=0 \text{ if } i\nin\cup a\}. 
\end{equation}
Then $\ct=\{X_a\mid a\in\pi\}$. Thus, \emph{if we set
$\cs=\{X_a^\sigma \mid a\in\pi\}$ then the $C^*$-algebra generated by
the Hamiltonians \eqref{eq:NB01} with arbitrary $k_i\in\R^\nu$ and
$v_i,v_{ij}\in C_0(\R^\nu)$ is }
\begin{equation}\label{eq:nbrfcs}
  \rf(\cs)= \tsum{a\in\pi}{} \rf(X_a^\sigma).
\end{equation}
Note that, although we started with Hamiltonians involving only one
and two body interactions, the Hamiltonians affiliated to $\rf(\cs)$
may involve $k$-body interactions with any $k\leq N$, and much more in
fact.  We discus this question below.

We fix a \emph{subsemilattice $\cs\subset\mbg(\Xi)$ with $\min\cs=X$}
and a \emph{continuous positive and divergent function
  $h:X^*\to\R$}. Then $h(p)$ is a kinetic energy operator affiliated
to $\rf(X)=C_0(X^*)$ and our goal is to build self-adjoint operators
$H=h(p)+V$ affiliated to $\rf(\cs)$. We use the terminology and
results of \S\ref{ss:perturb}.  Let $\ch_h=D(h(p)^{1/2})$ be the form
domain of $h(p)$ and $\ch_h\subset\ch\subset\ch_h^*$ the associated
scale. Theorem \ref{th:afgrad} and Corollary \ref{co:yeshvz} imply:

\begin{theorem}\label{th:afgradII}
  For each $E\in\cs$ let $V(E)\in B(\ch_h,\ch_h^*)$ a symmetric
  operator such that \\[1mm]
  (1) $V(X)=0$ and $\varphi(p)V(E)(h(p)+1)^{-1/2} \in \rf(E)$
  $\forall\varphi\in C_\rmc(X^*)$,\\[1mm]
  (2) the family $\{V(E)\}_{E\in\cs}$ is norm summable in
  $B(\ch_h,\ch_h^*)$,\\[1mm]
  (3) $V(E)\geq-\mu_\sse h(p)-\nu_\sse$ with $\mu_\sse,\nu_\sse\geq0$,
  $\sum_{E\in\cs}\mu_\sse<1$, and
  $\sum_{E\in\cs}\nu_\sse<\infty$.\\[1mm] 
  Let $V=\sum_{E\in\cs}V(E)$ and $V_\sse=\sum_{F\leq E}V(F)$ for any
  $E\in\cs$. Then the form sums $H=h(p)+V$ and $H_\sse=h(p)+V_\sse$
  are bounded from below self-adjoint operators, with form domain
  $\ch_h$, strictly affiliated to $\rf(\cs)$, and
  $\cp_\sse H=H_\sse \ \forall E\in\cs$. If $\Xi\in\cs$ and
  $\cs_{\max}$ is the set of maximal elements of $\cs\setminus\{\Xi\}$
  then
    \begin{equation}\label{eq:hvzNb}
      \spe(H) = \ccup_{E\in\cs_{\max}}\spec(\cp_\sse H) .
    \end{equation}
\end{theorem}

We give explicit versions of condition (1) of Theorem
\ref{th:afgradII} under stronger assumptions on $h$. Since
$E=Y^\sigma$ with $Y\in\mbg(X)$ we have
$T=\varphi(p)V(E)(h(p)+1)^{-1/2}\in\rf(E)$ if and only if $T$
satisfies the conditions of one of the Propositions \ref{pr:rfYsigma1}
or \ref{pr:rfYsigma2}.  Since $\e^{\rmi\bracet{x}{p}}$ and
$(h(p)+1)^{-1/2}$ are commuting operators,
$\{\e^{\rmi\bracet{x}{p}}\}_{x\in X}$ induces strongly continuous
unitary groups in $\ch_h$ and $\ch_h^*$, but $\ch_h$ is not stable
under the operators $\e^{\rmi\bracet{q}{k}}$.  In order to simplify
later statements we will impose a rather strong growth condition on
$h$ which ensures $\e^{\rmi\bracet{q}{k}}\ch_h\subset\ch_h$; the
condition can be relaxed so as to cover hypoelliptic operators for
example (see \cite[\S4.3]{DG1}) but this is not really relevant in our
context.

We equip $X$ with an Euclidean structure, denote $|\cdot|$ its norm,
set $\jap{x}=(1+|x|^2)^{1/2}$, and identify $X^*=X$. Then
$\ch^s\equiv\ch^s(X)$ is the Sobolev space of order $s\in\R$ defined
by the norms $\|u\|_s=\|\jap{p}^su\|$. Then $\ch^0 =\ch$,
$\ch^s\subset\ch^t$ if $s\geq t$, and
$\ch^s\subset \ch \subset \ch^{-s}= (\ch^s)^*$ if $s \geq 0$. Clearly
$\{\e^{\rmi\bracet{x}{p}}\}_{x\in X}$ and
$\{\e^{\rmi\bracet{q}{k}}\}_{k\in X}$ are strongly continuous
$C_0$-groups of bounded operators in $\ch^s$ for any real $s$.

\begin{definition}\label{df:small}
If $s \geq 0$ let $\mbfb^s(X) \doteq B(\ch^s,\ch^{-s})$ and
\begin{align}
  \mbfb_0^s(X)&=\{T\in\mbfb^s(X) \mid T:\ch^{s}\to\ch^{-t}
                \text{ is compact if } t>s\} \label{eq:mbfb01}\\
             &=\{T\in\mbfb^s(X) \mid \varphi(p)T:\ch^{s}\to\ch
  \text{ is compact if } \varphi\in C_\rmc(X)\}. \label{eq:mbfb02}   
\end{align}
If $T\in \mbfb_0^s(X)$ we say that \emph{$T:\ch^{s}\to\ch^{-s}$ is
  small at infinity}.
\end{definition}

If the condition in \eqref{eq:mbfb01} is satisfied for one $t$ it is
satisfied for any $t$.  The elements of $\mbfb_0^s(X)$ are the
operators $T\in\mbfb^s(X)$ which decay at infinity in a weak
sense. For example, $\varphi(q)\in\mbfb^0(X)$ if and only if
$\varphi:X\to\C$ is a bounded Borel function and
$\varphi(q)\in\mbfb_0^0(X)$ means that this function satisfies
$\lim_{a\to\infty}\int_{|x-a|}|\varphi(x)|\d x=0$ \cite{DG}.

If $h$ is a real functions on $X$  and $s>0$ is a number we write
$h(x)\sim |x|^{2s}$ if 
\begin{equation}\label{eq:s2cond}
  c'|k|^{2s} \leq h(k) \leq c''|k|^{2s}
  \quad\text{for some constants } c', c'' \text{ and all large } k.
\end{equation}
From now on we assume that this condition is satisfied.  Clearly this
is equivalent to $\ch_h=\ch^s$. Thus $\ch_h^*=\ch^{-s}$ and the $V(E)$
from Theorem \ref{th:afgradII} are symmetric operators
$V(E)\in\mbfb^s(X)$. Recall that $E=Y^\sigma$ for some $Y\in\mbg(X)$.

\begin{lemma}\label{lm:stlemma}
  Let $V(E)\in\mbfb^s(X)$ such that
  $[\e^{\rmi\bracet{x}{p}},V(E)] = 0 \ \forall x\in Y$. Assume that
  there is $t>s$ such that the following conditions hold in norm in
  $B(\ch^s,\ch^{-t})$: 
\begin{align}
  & [\e^{\rmi\bracet{x}{p}},V(E)] \to0 \quad\text{if }
    x\to0 \label{eq:21}\\  
  & [\e^{\rmi\bracet{q}{k}},V(E)] \to0 \quad\text{if }
    k\to0, \label{eq:31}\\ 
  & (\e^{\rmi\bracet{q}{k}}-1) V(E) \to0 \quad\text{if } k\in Y^\perp,
    k\to0 .\label{eq:41} 
\end{align}
Then $\varphi(p)V(E) (h(p)+1)^{-1/2}\in \rf(Y^\sigma)$ for any
$\varphi\in C_\rmc(X)$.
\end{lemma}

\begin{proof}
  We have
  \[
    \varphi(p)V(E) (h(p)+1)^{-1/2}=\varphi(p)V(E)\jap{p}^{-s}
    \cdot \jap{p}^{s}(h(p)+1)^{-1/2}
  \]
  hence by \eqref {eq:s2cond} and Remark \ref{re:mbfp} (with
  $\cs=\{E\}$) we have for any $\varphi\in C_\rmc(X)$
  \[
    \varphi(p)V(E) (h(p)+1)^{-1/2}\in \rf(E) \Leftrightarrow
    \varphi(p)V(E)\jap{p}^{-s}\in \rf(E).
  \]  
  $T=\varphi(p)V(E)\jap{p}^{-s}\in \rf(Y^\sigma)$ if and only if $T$
  satisfies the conditions of Proposition \ref{pr:rfYsigma1} and it is
  easy to check that they are equivalent to those in the statement
  of the lemma.
\end{proof}

Let $Y'$ be a complementary subspace of $Y$ in $X$. Then
$X=Y\oplus Y'$ and we identify
\[
L^2(X)\cong L^2(Y)\otimes L^2(Y')=L^2(Y;L^2(Y')).
\]
Let $F_\ssy$ be the Fourier transform associated to $Y$ acting in
$L^2(Y)$ and keep the notation $F_\ssy$ for $F_\ssy\otimes1$ which
acts in $L^2(X)=L^2(Y)\otimes L^2(Y')$.  Let $q_\ssy,p_\ssy$ be the
position and momentum observables associated to $Y$, so
$\varphi(p_\ssy)=F_\ssy^{-1}\varphi(q_\ssy)F_\ssy$ for any Borel
function $\varphi:Y\to\C$, hence $C_0(Y^*)=F_\ssy^{-1}C_0(Y)F_\ssy$
(see \S\ref{ss:i22} and \S\ref{ss:HLdec}). Then by
\eqref{eq:rfYsigma2}
\[
  F_\ssy^{-1} C_0(Y;\mbfk(Y')F_\ssy
  =F_\ssy^{-1}\big(C_0(Y)\otimes\mbfk(Y')\big)F_\ssy
  =C_0(Y^*)\otimes\mbfk(Y')=\rf(Y^\sigma).
\]
Thus if $\Phi:Y\to\mbfk(Y')$ is of class $C_0$ and we denote
$\Phi(q_\ssy)$ the operator of multiplication by $\Phi$ in
$L^2(X)=L^2(Y;L^2(Y'))$, then
$\Phi(p_\ssy)\doteq F_\ssy^{-1} \Phi(q_\ssy) F_\ssy\in\rf(Y^\sigma)$.
This definition extends to functions on $Y$ with values operators on
Sobolev spaces on $Y'$.

\begin{lemma}\label{lm:vey}
  Let $V_\ssy:Y\to\mbfb_0^s(Y')$ norm continuos and such that for any
  $y$ the operator $V_\ssy(y):\ch^s(Y')\to\ch^{-s}(Y')$ is symmetric
  and satisfies $\pm V_\ssy(y)\leq c_Y\jap{|y|+|p_{Y'}|}^{2s}$ for
  some constant $c_Y$. Then
  $\varphi(p) V_\ssy(p_\ssy) (h(p)+1)^{-1/2}\in \rf(Y^\sigma)$ for any 
  $\varphi\in C_\rmc(X)$.
\end{lemma}

Indeed, a straightforward computation shows that $V(E)=V_\ssy(p_\ssy)$
satisfies the conditions of Lemma \ref{lm:stlemma}. Alternatively,
Lemma \ref{lm:vey} is a consequence of Proposition \ref{pr:rfYsigma2},
\conf the proof of \cite[Th.\ 4.6]{DG}. Theorem \ref{th:afgradV} is a
consequence Theorem \ref{th:afgradII} and Lemma \ref{lm:vey}.

\begin{theorem}\label{th:opafrf}
  Let $\cs\subset\mbg(\Xi)$ a subsemilattice with $\min\cs=X$ and
  $h:X\to\R$ continuous with $h(x)\sim |x|^{2s}$ for some $s>0$. 
Let $V:\ch^s\to\ch^{-s}$ symmetric such that  $V\geq -\mu h(p)-\nu$
with $\mu<1,\nu\geq0$ and satisfying the following equivalent
conditions:  
\begin{equation}\label{eq:tVs}
  \varphi(p)V\jap{p}^{-s} \in \rf(\cs)\ \forall\varphi\in
  C_\rmc(\R)  \Leftrightarrow  
  \jap{p}^{-t}V\jap{p}^{-s} \in \rf(\cs) \text{ for some }t>s. 
\end{equation}
Then the form sum $H=h(p)+V$ is a self-adjoint operator strictly
affiliated to $\rf(\cs)$.
\end{theorem}

\begin{proof}

  We first prove the equivalence in \eqref{eq:tVs}.  Since
  $\varphi(p)\in\rf(\cs)$ the implication $\Rightarrow$ is
  clear. Assume that the right hand side is satisfied and let
  $\varphi$ be a continuous function on $X$ such that
  $0\leq\varphi\leq1$, $\varphi(x)=1$ if $|x|<1$, and $\varphi(x)=0$
  if $|x|>2$. If $\varepsilon>0$ then
  $\varphi(\varepsilon p)V\jap{p}^{-s}\in\rf(\cs)$ and
  \[
    \jap{p}^{-t}V \jap{p}^{-s}- \jap{p}^{-t}\varphi(\varepsilon p) V
    \jap{p}^{-s}= \jap{p}^{-(t-s)}(1-\varphi(\varepsilon p))
    \cdot\jap{p}^{-s}V \jap{p}^{-s} .
  \]
  The second term on the left hand side belongs to $\rf(\cs)$ and
  \[
    \|\jap{p}^{-(t-s)}(1-\varphi(\varepsilon p)) \cdot\jap{p}^{-s}V
    \jap{p}^{-s}\| \leq \sup_{|\varepsilon k|>1}\jap{k}^{-(t-s)}
    \|\jap{p}^{-s}V\jap{p}^{-s}\| .
  \]
  The right hand side above clearly tends to zero when
  $\varepsilon\to0$ hence $\jap{p}^{-t}V \jap{p}^{-s}\in\rf(\cs)$.

  We will use the notations and results of Theorem \ref{th:affg}.  The
  self-adjoint operator $H_0=h(p)$ is bounded from below and strictly
  affiliated to $\rf(X)=C_0(X^*)$ hence strictly affiliated to
  $\rf(\cs)$. Its form domain is $\ch^s$ whose adjoint space is
  $\ch^{-s}$ hence $V$ is a form perturbation of $H_0$, hence the form
  sum $H=H_0+V$ is a bounded from below self-adjoint operator. If
  $\varphi\in C_\rmc(\R)$ then $\psi=\varphi\circ h\in C_\rmc(X)$ and
  $\varphi(H_0)=\psi(p)$ so
  $\varphi(H_0)V(|H_0|+1)^{-1/2}\in\rf(\cs)$. Indeed, we have
  \[
    \varphi(H_0)V(|H_0|+1)^{-1/2}=\psi(p)V\jap{p}^{-s}\cdot
    \jap{p}^s(|h(p)|+1)^{-1/2}
  \]
  and $\jap{p}^s(|h(p)|+1)^{-1/2}=\theta(p)$ for some
  $\theta\in C_{\rmb}(X)$ so the right hand side above belongs to
  $\rf(\cs)\cdot C_\rmb(X)=\rf(\cs)$ by Remark \ref{re:mbfp}. So the
  result follows from Theorem \ref{th:affg}.
\end{proof}

\section{Extended field C*-algebra}
\label{s:efa}  
\protect\setcounter{equation}{0}

Let $\Xi$ be a symplectic space and $W$ a representation of $\Xi$ on a
Hilbert space $\ch$. The field $C^*$-algebra $\rf$ as defined in (1)
of Definition \ref{df:wkfield} contains, in some sense, only functions
of the fields $\phi(\xi)$. If the representation $W$ is not
irreducible then many other physically interesting observables are not
affiliated to $\rf$, \eg the Hamiltonians involving spin
interactions. Here we introduce a $\mbg(\Xi)$-graded $C^*$-algebra
$\re$ of operators on $\ch$ which contains $\rf$ and fixes this flaw.

In the next definition, suggested by Theorem \ref{th:mainthfd}, we
introduce the components of $\re$.

\begin{definition}\label{df:intrinsic}
  If $E\in\mbg(\Xi)$ then $\re(E)\equiv\re_\ssxi^\ssw(E)$ is the set
  of $T\in B(\ch)$ such that:\\[1mm]
  {\rm(i)} \ \ $\|[W(\xi),T]\| \to 0$ if $\xi\to0$ in $\Xi$,\\[1mm]
  {\rm(ii)} \ $[W(\xi),T] = 0$ if $\xi \in E^\sigma$,\\[1mm]
  {\rm(iii)} $\|(W(\xi)-1)T\| \to 0$ if $\xi \in E$ and
  $\xi \rarrow 0$.\\[1mm]
  In particular $\re(0)=\Com(\Xi)$.
\end{definition}

If $\dim\Xi=\infty$ the limit in (i) is
defined in the following equivalent ways:\\
(1)  for any $\xi\in\Xi$ we have $\lim_{r\to0}\|[W(r\xi),T]\|=0$,\\
(2) for any $F\in\mbg(\Xi)$ we have
$\lim_{\xi\in F,\xi\to0}\|[W(\xi),T]\|=0$. \\
Thus (i) can be reformulated as: \emph{the map
$\xi\mapsto W(\xi)TW(\xi)^*$ is norm continuous on finite dimensional
subspaces of $\Xi$}.  To prove (1)$\Rightarrow$(2) note that
$\|[W(\xi),T]\|=\|\rw(\xi)T-T\|$ where $\rw(\xi)$ is the automorphism
of $B(\ch)$ defined by $\rw(\xi)T=W(\xi)TW(\xi)^*$ and if
$e_1,\dots,e_n$ is a basis in $F$ and $\xi=\sum_i r_ie_i$ then
$\rw(\xi)=\rw(r_1e_1)\dots\rw(r_ne_n)$.

\begin{proposition}\label{pr:rdxi}
  Assume $\Xi$ finite dimensional. Then
  $\re(\Xi)=\rf(\Xi)\cdot\Com(\Xi)$ and $W$ is of finite multiplicity
  if and only if $\re(\Xi)=K(\ch)$. 
\end{proposition}

\begin{proof}
  If $\Xi$ is finite dimensional and we take $E=\Xi$ then only (i) and
  (iii) are nontrivial and (i)+(iii) is equivalent to
  $\lim_{\xi\to0}\|(W(\xi)-1)T^{(*)}\|=0$ hence by (1) of Theorem
  \ref{th:RKR} we have $\re(\Xi)=\rf(\Xi)\cdot\Com(\Xi)$. Then we use
  (3) of Theorem \ref{th:RKR}.
\end{proof}  

\begin{proposition}\label{pr:tatar}
  $\re(E)$ is a non-degenerate $C^*$-subalgebra of $B(\ch)$ and
  $\rf(E)\subset\re(E)$.
\end{proposition}

\begin{proof}

  It is easily seen that $\re(E)$ is a norm closed subalgebra of
  $B(\ch)$. If $T\in\re(E)$ then clearly $\|[W(\xi),T^*]\| \to 0$ if
  $\xi\to0$ in $\Xi$ and $[W(\xi),T^*] = 0$ if $\xi \in
  E^\sigma$. Moreover if $\xi \rarrow 0$ in $E$ then
  $\|T^*(W(\xi)-1)\| \to 0$ hence also $\|(W(\xi)-1) T^*\| \to 0$
  because $\|[W(\xi),T^*]\| \to 0$. Thus $T^*\in\re(E)$ hence $\re(E)$
  is a $C^*$-algebra. Below we show $\rf(E)\subset\re(E)$ and since
  $\rf(E)$ is non-degenerate, $\re(E)$ is also non-degenerate.
  
  To prove the last assertion of the proposition it suffices to show
  that $W(\mu)\in {\re}(E)$ if $\mu\in L^1(E)$.  For any $\xi \in \Xi$
  and $\mu\in M(E)$ we have
  \[ [W(\xi),W(\mu)] = \int_{E} [W(\xi),W(\eta)] \mu(\dd\eta) =
    \int_{E} (\e^{\rmi\sigma(\eta,\xi)}-1) W(\eta)W(\xi)\mu(\dd\eta) .
  \]
  Thus $\|[W(\xi),W(\mu)]\| \rarrow 0$ if $\xi \rarrow 0$. Moreover,
  the above commutator is clearly zero if $\xi \in E^\sigma$. It
  remains to show that $W(\mu)$ verifies (iii) of Definition
  \ref{df:intrinsic} if $\mu\in L^1(E)$. Then
  $\mu(\d\eta)=\rho(\eta)\d_E\eta$ with $\rho$ is a function on $E$.
  If $\xi \in E$ we get from \eqref{eq:ccr1I}:
  \begin{align*}
    &\| (W(\xi) -1) W(\mu)\| 
      = \biggl\| \int_{E} (\e^{-\frac{\rmi}{2}\sigma(\xi,\eta)}
      W(\xi + \eta)-W(\eta))\rho(\eta) \d_E\eta \biggr\| \\ 
    &\leq \int_{E} |\e^{-\frac{\rmi}{2}\sigma(\xi,\eta)} -1|
      \,|\rho(\eta)|\d_E\eta +
      \int_{E}|\rho(\eta -\xi) - \rho(\eta)|\d_E\eta .
  \end{align*}
  Clearly both integrals tend to zero as $\xi\rarrow 0$ in $E$.
\end{proof}

\begin{remark}\label{re:extmain}

  {\rm If $W$ is of finite multiplicity then
    $\re(E)=\rf(E)\cdot\Com(\Xi)$: the inclusion $\supset$ is obvious
    and the proof of the converse is similar to that of Theorem
    \ref{th:mainthfd} given in \S\ref{ss:idc}.}

\end{remark}

Now let $\rho$ be an integrable function on $E$ such that
$\int_{E}\rho(\eta) \d_E\eta=1$ for some Lebesgue measure $\d_E\eta$
on $E$.  For $\veps >0$ set
$\rho_\veps (\eta) = \veps^{-m}\rho(\eta/\veps)$, where $m=\dim E$,
and 
\begin{equation}\label{eq:wrhoeps}
  W(\rho_\varepsilon)=\int_EW(\eta)\rho_\varepsilon(\eta)\d_E\eta
  =\int_EW(\varepsilon\eta)\rho(\eta)\d_E\eta .
\end{equation}

\begin{lemma}\label{lm:surya}
{\rm(i)} $W(\rho_\veps)\in {\re}(E)$ and
$\slim_{\veps\rarrow 0}W(\rho_\veps)=1$ on $\ch$.\\[1mm]
\hspace*{19mm}{\rm (ii)}
If $T\in {\re}(E)$ then
$\lim_{\veps\rarrow 0}\|(W(\rho_\veps)-1)T\|=0$.
\end{lemma}

\begin{proof}
  For (i), see the proof of Proposition \ref{pr:properties}-(a).
  Then, by Definition \ref{df:intrinsic}-(iii) we have
\[
\| (W(\rho_\veps)-1)T\| \leq
\int_{E}
\|(W(\veps\eta) -1)T\|\,|\rho(\eta)| \lambda_E(\dd\eta)
\rarrow 0 \text{ as } \veps \rarrow 0. \qedhere
\]
\end{proof}

\begin{lemma}\label{lm:rcliminf}
  If $\xi\nin E^\sigma$, and $T\in\re(E)$ then
  $\slim_{r\to\infty} T W(r\xi)=0$
\end{lemma}

\begin{proof}
  By Lemma \ref{lm:surya} we have
  $\lim_{\varepsilon\to 0} TW(\rho_\varepsilon)=T$ in norm hence it
  suffices to have $\slim_{r\to\infty} SW(r\xi)=0$ for $S\in\rf(E)$
  which is true by Lemma \ref{lm:liminfinity}.  
\end{proof}

\begin{theorem}\label{th:rcgrad}
  The family of $C^*$-subalgebras $\{\re(E)\}_{E\in\mbg(\Xi)}$ is
  linearly independent and satisfies
  $\re(E)\cdot\re(F)\subset\re(E+F)$. 
\end{theorem}

\begin{proof}

  The proof of the linear independence of the family of $\re(E)$ is
  similar to the proof of Theorem \ref{th:remarc}-(1).  Let $\cs$ be
  a finite set of finite dimensional subspaces of $\Xi$ and for each
  $F\in\cs$ let $T(F)\in\re(F)$ such that $\sum_{F\in\cs}T(F)=0$; we
  have to show $T(E)=0$ $\forall E$.  If $F\in\cs$ and
  $F\not\subset E$ then $E^\sigma\not\subset F^\sigma$ hence
  $E^\sigma\cap F^\sigma$ is a strict subspace of $E^\sigma$.  Then we
  may choose $\xi\in E^\sigma$ which does not belong to any of these
  subspaces, \ie $\xi\nin F^\sigma$ if $F\in\cs$ is not a subset of
  $E$, which implies $\slim_{r\to\infty} T(F)W(r\xi)=0$ for all such
  $F$ by Lemma \ref{lm:rcliminf}.  On the other hand, if $F\subset E$
  then $\xi\in E^\sigma\subset F^\sigma$ hence
  $W(r\xi)^*T(F)W(r\xi)=T(F)$ by (ii) of Definition
  \ref{df:intrinsic}.  By taking $r\to\infty$ in the identity
  $\sum_{F\in\cs}W(r\xi)^*T(F)W(r\xi)=0$ we thus get
  $\sum_{F\subset E}T(F)=0$ for all $E\in\cs$.  If $E$ is minimal in
  $\cs$ this implies $T(E)=0$. Then $\sum_{F\in\cs_1}\mu(F)=0$ if
  $\cs_1$ is the set of elements of $\cs$ which are not minimal. By
  repeating the above argument for $\cs_1$ we get $T(E)=0$ for all $E$
  minimal in $\cs_1$, etc.

  Now we prove $\re(E)\cdot\re(F)\subset\re(E+F)$. Let $S\in\re(E)$ and
  $T\in\re$, we must show $S T\in\re(E+F)$. Since
  $(E+F)^\sigma=E^\sigma\cap F^\sigma$, the conditions (i) and (ii) of
  Definition \ref{df:intrinsic} are obviously satisfied, it remains to
  show $\|(W(\xi)-1)S T\|\to0$ if $\xi\to0$ in $E+F$. If $E\not\subset
  F$ let $G$ be subspace of $F$ supplementary to $E\cap F$, hence
  $E+F=E+G$ and $E\cap G=0$. Then any $\xi\in\ E+F$ has a unique
  decomposition $\xi=\eta+\zeta$ with $\eta\in E, \zeta\in G$ and
  $\xi\to0$ in $E+F$ is equivalent to $\eta\to0$ in $E$ and
  $\zeta\to0$ in $G$ (hence in $F$). Then
  \[
    \e^{\frac{\rmi}{2}\sigma(\zeta,\eta)}W(\xi)S T= W(\eta)W(\zeta)S T
    = W(\eta)[W(\zeta),S] T + W(\eta)S W(\zeta)T
  \]
  and the first term on the right hand side tends in norm to zero as
  $\zeta\to0$ while $\e^{\frac{\rmi}{2}\sigma(\zeta,\eta)}\to1$ and
  $W(\eta)S W(\zeta)T\to S T$ in norm if $\eta\to0$ and $\zeta\to0$.
\end{proof}

\begin{definition}\label{df:extrc}
  The \emph{extended field $C^*$-algebra} associated to the
  representation $W$ is
  \begin{equation}\label{eq:rcgrad}
    \re\equiv\re_\ssxi^{\ssw}\doteq \tsum{E\in\mbg(\Xi)}{\rmc} \re(E)
  \end{equation}
  and is a $\mbg(\Xi)$-graded $C^*$-algebra of bounded operators on
  $\ch$.
\end{definition}

In Proposition \ref{pr:rdxi} we already gave an example where the
algebra $\re$ is more natural than $\rf$. We give one more example
involving coupling of finite multiplicity representations, which
covers for example N-body Dirac and Pauli operators.

Let $\Xi_1,\dots,\Xi_N$ finite dimensional symplectic spaces equipped
with finite multiplicity representations $W_1,\dots,W_N$ on Hilbert
space $\ch_1,\dots,\ch_N$. The coupling of these systems is defined as
follows: $\Xi\doteq\Xi_1\oplus\dots\oplus\Xi_N$ (symplectic direct
sum) equipped with the natural representation on
$\ch\doteq\ch_1\otimes\dots\otimes\ch_N$ defined by 
\[
  \xi=\xi_1\oplus\dots\oplus\xi_N \Rightarrow
  W(\xi)=W_1(\xi_1)\otimes\dots\otimes W_N(\xi_N) .
\]
Then the extended field algebra $\re$ associated to $(\Xi,W)$ contains
$\re_1\otimes\dots\otimes\re_N$ and Hamiltonians affiliated to it may
be constructed with the help of Theorem \ref{th:affg} starting with
``free'' Hamiltonians of the form
$H_1\otimes1\otimes\dots\otimes1+\dots+ 1\otimes\dots\otimes1\otimes
H_N$.

\appendix

\section{C*-algebras graded by Grassmannians}\label{s:graslat}
\protect\setcounter{equation}{0}

\subsection{Main facts}\label{ss:cgrad}

Let $\Xi$ be a real vector space and $\rc$ a $C^*$-algebra graded by a
subsemilattice $\cs$ of $\mbg(\Xi)$. We use the definitions and
notations of \S\ref{ss:i0}, thus $\rc=\sum_{E\in\cs}^\rmc\rc(E)$
for a given family $\{\rc(E)\}_{E\in\cs}$ of $C^*$-subalgebras of
$\rc$ satisfying $\rc(E)\rc(F)\subset\rc(E+F)$ for all $E,F\in\cs$ and
such that the linear sum $\mathring\rc\doteq\sum_{E\in\cs} \rc(E)$ is
direct and dense in $\rc$.

If $\ct\subset\cs$ we set $\mathring\rc(\ct)=\sum_{E\in\ct}\rc(E)$ and
$\rc(\ct)=\sum_{E\in\ct}^\rmc\rc(E)$. If $\ct$ is a subsemilattice
then clearly $\rc(\ct)$ is a $\ct$-graded $C^*$-algebra with
components $\rc(E)$; subalgebras of this form are called \emph{graded
  $C^*$-subalgebras} of $\rc$. Below we mention some properties of
$\rc(\ct)$: (1) is \cite[Pr.\ 1.6]{M1}, (2) is an exercise, and (3) is
\cite[Pr.\ 3.2]{DG}.

\begin{lemma}\label{lm:subsemi}
{\rm(1)} If $\ct$ is finite then $\sum_{E\in\ct}\rc(E)$ is closed,
hence $\rc(\ct)=\sum_{E\in\ct}\rc(E)$. \\[1mm]
{\rm(2)} If $\ct$ is finite and $\rc(E)\cdot\rc(F)=\rc(E+F)$ for
all $E,F\in\ct$, then $\rc(\ct)$ is a subalgebra of $\rc$ if and
only if $\ct$ is a subsemilattice of $\cs$.\\[1mm]
{\rm(3)} $\rc=\bigcup_\ct\rc(\ct)$ union over all countable
subsemilattices of~$\cs$.
\end{lemma}

For any subspace $E$ of $\Xi$ the sets
\begin{equation}\label{eq:sas}
  \cs_\sse=\{F\in\cs\mid F\subset E\},\
  \cs'_\sse=\{F\in\cs\mid F\not\subset E\},\
 \cs_{\ss{\supset E}}=\{F\in\cs\mid F\supset E\}
\end{equation}
are subsemilattices. If $E\in\cs$ the next theorem is \cite[Th.\
3.1]{DG0} while \cite[Pr.\ 1.10]{M1} is a more general result; the
last assertion is clear if $T\in \mathring\rc$ which is dense in
$\rc$, etc.

\begin{theorem}\label{th:gmain}
  $\rc_\sse\doteq\rc(\cs_\sse)$ is a $C^*$-subalgebra and
  $\rc'_\sse\doteq\rc(\cs'_\sse)$ an ideal of $\rc$ such that
  $\rc=\rc_\sse+\rc'_\sse$ and $\rc_\sse\cap \rc'_\sse=0$. The
  projection $\cp_\sse:\rc\to\rc_\sse$ determined by this direct sum
  decomposition is a morphism.  $\{\cp_\sse T\mid E\in\cs\}$
  is relatively compact $\forall\ T\in\rc$.
\end{theorem}

The projection morphism (\conf \S\ref{ss:i0}) $\cp_\sse$ is the unique
continuous map $\rc\to\rc$ such that
\begin{equation}\label{eq:pmorph}
  \cp_\sse\tsum{G}{}T(G)=\tsum{G\subset E}{}T(G)
\end{equation}
where $T(G)\in\rc(G)$ and $T(G)=0$ but for a finite number of $G$.
Clearly if $E\subset F$ then $\rc_\sse\subset \rc_\ssf$ and
$\cp_\sse\cp_\ssf=\cp_\ssf\cp_\sse=\cp_\sse$.  More generally:

\begin{proposition}\label{pr:ecapf}
  If $E,F$ are subspaces of $\Xi$ then
\begin{equation}\label{eq:ecapf}
\rc_{\ss{E\cap F}}=\rc_\sse\cap\rc_\ssf \quad\text{and}\quad
  \cp_{\ss{E\cap F}}=\cp_\sse\cp_\ssf=\cp_\ssf\cp_\sse
\end{equation}
\end{proposition}

\begin{proof}

  We first prove the relation $\cp_\sse\cp_\ssf=\cp_{\ss{E\cap F}}$.
  Since $\rc$ is the closure of $\sum_{G\in\cs}\rc(G)$ it suffices to
  show that the restrictions of the maps $\cp_\sse\cp_\ssf$ and
  $\cp_{\ss{E\cap F}}$ to each
  $\rc(G)$ are equal. There are three possibilities:\\[1mm]
  (1) $G\subset E\cap F$: then each of the operators
  $\cp_\sse, \cp_\ssf, \cp_{\ss{E\cap F}}$ is the identity on $\rc(G)$
  so
  $\cp_\sse\cp_\ssf=\cp_{\ss{E\cap F}}$ on $\rc(G)$;\\[1mm]
  (2) $G\not\subset F$: then
  $\cp_\ssf\rc(G)= \cp_{\ss{E\cap F}}\rc(G)=0$
  hence $\cp_\sse\cp_\ssf=\cp_{\ss{E\cap F}}$ on $\rc(G)$;\\[1mm]
  (3) $G\subset F$ but $G\not\subset E$: then $\cp_\ssf$ is the
  identity on $\rc(G)$, $\cp_\sse\rc(G)=0$,
  $\cp_{\ss{E\cap F}}\rc(G)=0$ hence
  $\cp_{\ss{E\cap F}}=\cp_\sse\cp_\ssf$ on $\rc(G)$.\\[1mm]
  Thus $\cp_\sse\cp_\ssf=\cp_{\ss{E\cap F}}$. The operator
  $\cp_{\ss{E\cap F}}$ is a projection of $\rc$ onto
  $\rc_{\ss{E\cap F}}$ and we clearly have
  $\rc_{\ss{E\cap F}}\subset \rc_\sse\cap\rc_\ssf$. On the other hand
  each of the operators $\cp_\sse,\cp_\ssf$ is the identity on
  $\rc_\sse\cap\rc_\ssf$ hence $\cp_{\ss{E\cap F}}=\cp_\sse\cp_\ssf$
  is also the identity on $\rc_\sse\cap\rc_\ssf$ so
  $\rc_{\ss{E\cap F}}\supset \rc_\sse\cap\rc_\ssf$.
\end{proof}

\subsection{Quotient algebra}\label{ss:quot}

We now assume that the subsemilattice $\cs$ has a greatest element
$\sup\cs$ or, equivalently, $\sup_{E\in\cs}\dim E<\infty$. Indeed, in
this case there is $G\in\cs$ of maximal dimension and if $E\in\cs$ is
not included in $G$ then $E+G\in\cs$ and $G\subsetneq E+G$ so
$\dim G<\dim(E+G)$ which is not possible, hence $G$ is the greatest
element of $\cs$.

Let $\cs_{\max}$ be the set of subspaces $E\in\cs$ such that
$E\neq \sup\cs$ and there are no other elements of $\cs$ between $E$
and $\sup\cs$; in other terms, $\cs_{\max}$ is the set of maximal
elements of $\cs\setminus\{\sup\cs\}$.  Then each element of $\cs$
distinct of $\sup\cs$ is majorated by an element of $\cs_{\max}$
because if $F\in\cs,F\neq \sup\cs$, among the elements $E\in\cs$ which
are $\neq \sup\cs$ and contain $F$ there is one which has maximal
dimension and this one belongs to $\cs_{\max}$.

Clearly $\rc(\sup\cs)$ is an ideal of $\rc$. Let
$\cp\colon\rc\to\rc/\rc(\sup\cs)$ be the canonical surjection of $\rc$
onto the quotient $C^*$-algebra $\rc/\rc(\sup\cs)$.  
The main interest of the graded structure is that it gives an explicit
description of $\rc/\rc(\sup\cs)$ and $\cp$. Theorem \ref{th:smax} is
a consequence of the more general \cite[Pr.\ 3.4]{DG}.  If $\cs$ is
finite, which covers the N-body type situations, the proof is very
easy \cite[Th.\ 8.4.1]{ABG}.

If $\ra_i$ are $C^*$-algebras then $\bigoplus_{i\in I}\ra_i$ is the
$C^*$-algebra consisting of families $A=(A_i)_{i\in I}$ with
$A_i\in\ra_i$ satisfying $\|A\|\doteq\sup_i\|A_i\|<\infty$ and with
the usual algebraic operations. 

\begin{theorem}\label{th:smax}
  If $S=\sup\cs$ exists then the map
    \begin{equation}\label{eq:ssup}
      \cp:\rc\to{\textstyle\bigoplus_{E\in\cs_{\max}}}\rc_\sse
      \quad\text{defined by}\quad
      \cp (T)=\big(\cp_\sse T\big)_{E\in\cs_{\max}} 
  \end{equation}
  is a morphism with $\ker\cp=\rc(\sup\cs)$. This gives a canonical
  embedding
  \begin{equation}\label{eq:squotembeded}
    \rc/\rc(\sup\cs)\hookrightarrow\ooplus_{E\in\cs_{\max}}\rc_\sse .
  \end{equation}
\end{theorem}

If $T\win\rc$ we define
\begin{equation}\label{eq:rcspe}
  \text{$\rc$-essential spectrum of } T \equiv \rc\hyphen\spe(T)
  \doteq \text{ spectrum of } \cp(T) .
\end{equation}
In later concrete situations $\rc$ is realized on a Hilbert space
$\ch$ such that $\rc(\sup\cs)=\rc\cap K(\ch)$ (Lemma \ref{lm:nocomp})
and then $\spe(T)=\rc\hyphen\spe(T)$ for any $T\win\rc$ by Atkinson's
theorem \cite[Th.\ 3.3.2]{Arv} that we recall here. This gives a
general version of the HVZ theorem.

Fix a Hilbert space $\ch$.  The quotient $C^*$-algebra
$\cc(\ch)\doteq B(\ch)/K(\ch)$ is the \emph{Calkin algebra} of $\ch$
and $\pi:B(\ch)\to\cc(\ch)$ is the canonical morphism.  Then for any
$A\in B(\ch)$ we call \emph{localization at infinity} of $A$ its image
$\hat{A}\doteq\pi(A)$ in $\cc(\ch)$; Atkinson's theorem says that the
\emph{essential spectrum of $A\in B(\ch)$ is the spectrum of
  $\hat A$:}
\begin{equation}\label{eq:specess1}
\spe(A)=\spec(\hat A).
\end{equation}
We extend this to observables as follows: if
$A\colon C_0(\R)\to B(\ch)$ is a morphism, hence an observable
affiliated to $B(\ch)$, then $\hat A\doteq\pi\circ A$ is a morphism
$C_0(\R)\to\cc(\ch)$, hence an observable affiliated to $\cc(\ch)$,
that we call \emph{localization at infinity of $A$}, and we define the
\emph{essential spectrum of $A$ as the spectrum of $\hat A$}, so
\eqref{eq:specess1} remains valid.  Thus $\spe(A)$ is the set of real
$\lambda$ such that for any $\theta\in C_0(\R)$ with
$\theta(\lambda)\neq0$ the operator $\theta(A)$ is not compact.

If $\rc\subset B(\ch)$ is a $C^*$-subalgebra then $\rc\cap K(\ch)$ is
an ideal of $\rc$ hence the quotient
$\what\rc\doteq\rc/\big(\rc\cap K(\ch)\big)$ is a $C^*$-subalgebra of
$\cc(\ch)$ called \emph{localization at infinity of $\rc$}. If the
observable $A$ is affiliated to $\rc$ its localization at infinity
$\what{A}$ is affiliated to $\what\rc$. Thus
\begin{equation}\label{eq:easy}
\rc\cap K(\ch)=0 \Longrightarrow
\spe(A)=\spec(A) \quad\forall A\win\rc
\end{equation}
because the restriction to $\rc$ of the canonical morphism
$B(\ch)\to\cc(\ch)$ is injective.

\subsection{Essential spectrum}\label{ss:gradess}

Assume now $\Xi$ finite dimensional and let $\rc$ be a
$\mbg(\Xi)$-graded $C^*$-algebra. Since $\Xi$ is the maximal element
of $\mbg(\Xi)$, if $T\win\rc$ then $\rc\hyphen\spe(T)=\spec(\cp T)$ is
the spectrum of the image of $T$ in the quotient algebra
$\rc/\rc(\Xi)$.

\begin{theorem}\label{th:gahvz}
  If $T\in\rc$ then $\{\cp_\ssh T \mid H\in \mbh(\Xi)\}$ is a compact
  subset of $\rc$ and
  \begin{equation}\label{eq:ghvzspexi}
    \rc\hyphen\spe(T)=\ccup_{H\in \mbh(\Xi)}\spec(\cp_\ssh T).
  \end{equation} 
  This relation is also valid for the observables affiliated to $\rc$.
\end{theorem}

Theorem \ref{th:smax} does not suffices to prove \eqref{eq:ghvzspexi},
the main new point is to show that the union in \eqref{eq:ghvzspexi}
is closed for all the elements of $\rc$. The argument requires some
preparations.

\begin{lemma}\label{lm:gessp}
  Let $\ra$ be a $C^*$-algebra, $I$ a set, and $\ra^{[I]}$ the
  $C^*$-algebra of bounded functions $I\to\ra$. If
  $A=(A_i)_{i\in I}\in \rc^{[I]}$ and $\ca=\{A_i\mid i\in I\}$ is a
  compact subset of $\ra$ then $\spec(A)=\cup_{i\in I}\spec(A_i)$.
\end{lemma}

\begin{proof}
  We mays clearly assume that $\ra$ is unital.  If $\lambda\in\C$ then
  $A-\lambda=(A_i-\lambda)_{i\in I}$ is invertible in $\ra^{[I]}$ if
  and only if each $A_i-\lambda$ is invertible in $\ra$ and
  $\|(A_i-\lambda)^{-1}\|\leq C$ for some number $C$ independent of
  $i$. We have to prove that the last condition is automatically
  satisfied if $\ca$ is compact.  In this case $\ca-\lambda$ is also
  compact so we may assume without loss of generality that
  $\lambda=0$. If $\ri$ is the set of invertible elements of $\ra$,
  then $\ri$ is an open subset of $\ra$ and the map $S\mapsto S^{-1}$
  is continuous on $\ri$. Since $\ca$ is a compact subset of $\ri$, we
  see that $\{A_i^{-1}\mid i\in I\}$ is a compact hence
  $\sup_i\|A_i^{-1}\|<\infty$.
\end{proof}

 \begin{lemma}\label{lm:EH}
   $T\in\rc$ and $E\in\mbg(\Xi)$,
   $E\neq \Xi\Rightarrow\exists H\in\mbh(\Xi)$ such that
   $\cp_\sse T=\cp_\ssh T$.
 \end{lemma}

 \begin{proof}
   We saw in Lemma \ref{lm:subsemi} that there is a countable
   subsemilattice $\cs\subset\mbg(\Xi)$ such that $T\in\rc(\cs)$.  We
   will prove that there is $H\in\mbh(\Xi)$ such that $E\subset H$
   but $F\not\subset H$ if $F\not\subset E$.

   We have $F\not\subset E$ if and only if
   $E^\perp\not\subset F^\perp$ (orthogonals in $\Xi^*$) hence if and
   only if $E^\perp\cap F^\perp$ is a strict subspace of
   $E^\perp$. Thus $\{E^\perp\cap F^\perp\mid F\in\cs_\sse'\}$, \conf
   \eqref{eq:sas}, is a countable set of strict subspaces of $E^\perp$
   hence there is $L\in\mbp(\Xi^*)$ such that $L\subset E^\perp$ and
   $L\not\subset F^\perp$ if $F\in\cs_\sse'$. Then $E\subset L^\perp$
   and $F\not\subset L^\perp$ if $F\not\subset E$. Thus it suffices to
   take $H=L^\perp$.

   Let $S\in\rc(F)$ for some $F\in\cs$. Then $\cp_\sse S=S$ if
   $F\subset E$ and $\cp_\sse S=0$ if $F\not\subset E$. In the first
   case we also have $F\subset H$ hence $\cp_\ssh S=S$ and in the
   second case $F\not\subset H$ hence $\cp_\ssh S=0$. Thus
   $\cp_\sse=\cp_\ssh$ on each $\rc(F)$ with $F\in\cs$. Since these
   two operators are linear and continuous we get $\cp_\sse=\cp_\ssh$ on
   $\rc(\cs)$ and so $\cp_\sse T=\cp_\ssh T$.
 \end{proof}

\begin{proof}[Proof of Theorem \ref{th:gahvz}]
  
  In this proof we abbreviate $\mbh=\mbh(\Xi)$. From Theorem
  \ref{th:smax} it follows that $\rc\hyphen\spe(T)$ is the spectrum of
  the element $(\cp_\sse T)_{E\in\mbh}$ of the $C^*$-algebra
  $\oplus_{E\in\mbh}\rc_\sse$.  Let $\rc^{\ss{[\mbh]}}$ be the
  $C^*$-algebra of bounded functions $\mbh\to\rc$ with the $\sup$
  norm. Then $\oplus_{E\in \mbh}\rc_\sse$ is a $C^*$-subalgebra of
  $\rc^{[\mbh]}$ hence for an element of $\oplus_{E\in\mbh}\rc_\sse$
  its spectrum in $\oplus_{E\in\mbh}\rc_\sse$ coincides with its
  spectrum in $\rc^{[\mbh]}$. Thus the spectrum of
  $(\cp_\sse T)_{E\in\mbh}$ in $\oplus_{E\in \mbh}\rc_\sse$ coincides
  with its spectrum in $\rc^{[\mbh]}$.  Thus by Lemma \ref{lm:gessp}
  it suffices to show that $\cR=\{\cp_\sse T \mid E\in \mbh\}$ is a
  compact set in $\rc$. By the last assertion of Theorem
  \ref{th:gmain} it suffices to show that $\cR$ is norm closed. We
  will use the notation $P_\sse =\cp_{\ss{E^\perp}}$ for
  $E\in\G(\Xi^*)$ hence $\cR=\{P_{\ss{L}}T\mid L\in\mbp\}$ where
  $\mbp\equiv\mbp(\Xi^*)$.

  We make two remarks concerning the $P_\sse $.  If $E,F\in\G(\Xi^*)$
  then by \eqref{pr:ecapf}
  \begin{equation}\label{eq:sums}
    P_\sse P_\ssf=\cp_{\ss{E^\perp}}\cp_{\ss{F^\perp}} =
    \cp_{\ss{E^\perp\cap F^\perp}} =\cp_{\ss{(E+F)^\perp}} =P_{\ss{E+F}}.
  \end{equation}
  Then if $P_\sse ^\prime=1-P_\sse $ we have
  $ P_\sse -P_\ssf= P_\sse P_\ssf^\prime-P_\sse ^\prime P_\ssf $ hence
  $P_\sse (P_\sse -P_\ssf)= P_\sse P_\ssf^\prime$, which gives for all
  $S\in \rc$
\begin{equation}\label{eq:pepf}
  \|P_\sse P_\ssf^\prime S\|\leq \|P_\sse  S-P_\ssf S\|.
\end{equation}
Let $S$ be an accumulation point of $\cR$. Then there is a sequence of
elements $L_n\in\mbp$ with $L_n\neq L_m$ if $n\neq m$ such that
$\|P_{\ss{L_n}}T-S\|\rarrow0$.  The sequence of subspaces
$E_n=\sum_{m\geq n}L_m\neq0$ is decreasing hence there is a subspace
$E$ of dimension $k>0$ such that $E_n=E$ for large $n$.  Then for each
such $n$ one can find integers $n_1<n_2<\dots<n_k$ with $n_1=n$ such
that $E=L_{n_1}+\dots+L_{n_k}$.  If $\varepsilon>0$ then there is $N$
such that $\|P_{\ss{L_n}}T-P_{\ss{L_m}}T\|\leq\varepsilon$ if
$n,m\geq N$.  Choose a sequence $n_1=n<n_2<\dots<n_k$ as above and
denote $F_i=L_{n_1}+\dots + L_{n_{i-1}}$. The relations
\eqref{eq:pepf} and \eqref{eq:sums} imply
\[
\varepsilon\geq\|P_{\ss{L_n}}T-P_{\ss{L_{n_i}}}T\|
\geq\|P_{\ss{L_n}}P_{\ss{L_{n_i}}}^\prime T\| \geq\|P_{\ss{F_i}}
P_{\ss{L_n}}P_{\ss{L_{n_i}}}^\prime T\|=
\|P_{\ss{F_i}}P_{\ss{L_{n_i}}}^\prime T\|. 
\]
On the other hand, by using \eqref{eq:sums} again we can write 
\[
P_{\ss{L_n}}=P_{\ss{L_{n_1}}}P_{\ss{L_{n_2}}}^\prime
+P_{\ss{L_{n_1}}}P_{\ss{L_{n_2}}}= P_{\ss{L_{n_1}}}P_{\ss{L_{n_2}}}^\prime+
P_{\ss{L_{n_1}+L_{n_2}}}.
\]
Repeating this procedure we get
\[
P_{\ss{L_n}}=\tsum{i=2}{k} P_{\ss{F_i}}P_{\ss{L_{n_i}}}^\prime+ P_{\sse}.
\]
Thus we have $\norm{P_{\ss{L_n}}T-P_\sse T}\leq
(k-1)\varepsilon$. This shows that $S=P_\sse T=\cp_{\ss{E^\perp}}T$
with $E^\perp\neq \Xi$ and now it suffices to use Lemma \ref{lm:EH}
with $E$ replaced by $E^\perp$.

The extension to observables is straightforward by using the spectral
mapping theorem. If $z$ is a nonreal number and $R_T(z)$ is the
resolvent of the observable $T$, by taking into account the relation
$\cp_\ssh R_T(z)=R_{\cp_\ssh T}(z)$ due to the fact that $\cp_\ssh$ is
a morphism, it suffices to use \eqref{eq:ghvzspexi} with $T$ replaced
by $R_T(z)$.
\end{proof}

\begin{theorem}\label{th:gahvzcs}
  If $\cs\subset\mbg(\Xi)$ is a subsemilattice with $\Xi\in\cs$ and
  $T\win\rc(\cs)$ then
  \begin{equation}\label{eq:ghvzspexics}
    \rc\hyphen\spe(T)=\ccup_{E\in\cs_{\max}}\spec(\cp_\sse T) .
  \end{equation} 
\end{theorem}

\begin{proof}
  Since $\rc(\cs)/\rc(\Xi)$ is a $C^*$-subalgebra of $\rc/\rc(\Xi)$
  the spectrum of the image of $T$ in the quotient algebra
  $\rc(\cs)/\rc(\Xi)$ is the same as the spectrum of the image of $T$
  in the quotient algebra $\rc/\rc(\Xi)$ hence is given by the formula
  \eqref{eq:ghvzspexi}.  If $E\neq \Xi$ there is $H\in\mbh(\Xi)$ such
  that $E\subset H$ and then $\cp_\sse=\cp_\sse\cp_\ssh$ by
  \eqref{eq:ecapf} hence $\cp_\sse T=\cp_\sse\cp_\ssh T$ and since
  $\cp_\sse$ is a morphism we get
  $\spec(\cp_\sse T)\subset\spec(\cp_\ssh T)$. Thus
  \[
    \rc\hyphen\spe(T)= \ccup_{H\in\mbh(\Xi)}\spec(\cp_\ssh T) \supset
    \ccup_{E\in\cs_{\max}}\spec(\cp_\sse T) .
\]
Observe that the preceding argument gives
$\spec(\cp_\ssf T)\subset\spec(\cp_\sse T)$ if $F\subset E$ hence
\[
  \ccup_{E\neq \Xi}\spec(\cp_\sse T)=
  \ccup_{E\in\cs_{\max}}\spec(\cp_\sse T)
\]
\ie only the largest spaces $E$ are significant in the union.  It
remains to show that for any $T$ and any hyperplane $H$ there is
$E\in\cs$ such that $E\subset H$ and
$\spec(\cp_\sse T)=\spec(\cp_\ssh T)$.  Let $\cs_H$ be the set of
$F\in\cs$ such that $F\subset H$ and choose $E\in\cs_H$ of maximal
dimension. If $F\in\cs_H$ then $E+F\in\cs_H$ and if $F\not\subset E$
then $\dim(E+F)>\dim E$ which is impossible, hence $F\subset E$. Thus
$E$ is the greatest element of the subsemilattice $\cs_H$. Then we
clearly have $\cp_\sse T=\cp_\ssh T$ for any $T\in \mathring\rc(\cs)$
hence by continuity for any $T\in\rc(\cs)$.
\end{proof}

\subsection{Affiliation criteria}\label{ss:perturb}

We first recall the notion of sum in form sense.  Let $H_0$ be a
self-adjoint operator on $\ch$ and $\cg=D(|H_0|^{\frac{1}{2}})$ its
form domain with the graph topology.  If $\cg^*$ is the adjoint of
$\cg$ (space of continuous anti-linear forms) and if we identify
$\ch^*=\ch$ via Riesz Lemma, we get continuous dense embeddings
$\cg\subset\ch\subset\cg^*$ and $H_0$ extends to a continuous map
$\cg\to\cg^*$ for which we keep the notation $H_0$.  A symmetric
operator $V\colon\cg\to\cg^*$ is a \emph{form perturbation of $H_0$}
if there are numbers $\mu,\nu\geq0$ with $\mu<1$ such that either
$\pm V\leq \mu|H_0|+\nu$ or $H_0$ is bounded from below and
$V\geq-\mu H_0-\nu$.  Then the restriction of
$H=H_0+V\colon\cg\to\cg^*$ to $D(H)\doteq\{g\in\cg\mid Hg\in\ch\}$ is
a self-adjoint operator on $\ch$, still denoted $H$, called \emph{form
  sum of $H_0$ and $V$}.  By \cite[Th.\! 2.8, Lm.\!  2.9]{DG}:

\begin{theorem}\label{th:affg}
  Assume that $H_0$ is strictly affiliated to a $C^*$-algebra
  $\rc\subset B(\ch)$. If $V$ is a form perturbation of $H_0$ such
  that $\varphi(H_0)V(|H_0|+1)^{-1/2} \in \rc$ for all
  $\varphi\in C_\rmc(\R)$ then the form sum $H=H_0+V$ is strictly
  affiliated to $\rc$.
\end{theorem}

Let $\Xi$ be an arbitrary real vector space, $\cs\subset\mbg(\Xi)$ a
subsemilattice with a \emph{least element} $X$, and $\rc$ an
$\cs$-graded $C^*$-algebra of operators on a Hilbert space $\ch$.  We
give here methods for constructing self-adjoint operators
affiliated to $\rc$. Note that
\begin{equation*}
  \rc(X)\cdot\rc(E)=\rc(E)\cdot\rc(X)\subset\rc(E)
  \quad\forall\ E\in\cs.
\end{equation*}
Proposition \ref{pr:afgrad} is very easy but interesting because $V$
is not required to be decomposable.

\begin{proposition}\label{pr:afgrad}
  Let $H_\ssx$ be a self-adjoint operator on $\ch$ affiliated to
  $\rc(X)$. If $V\in\rc$ is symmetric then $H\doteq H_\ssx+V$ is
  affiliated to $\rc$ and $\cp_\sse H=H_\ssx+\cp_\sse V$.
\end{proposition}

\begin{proof}
  If $z\in\C$ we set $R_\ssx(z)=(z-H_\ssx)^{-1}$ and $R(z)=(z-H)^{-1}$
  whenever the inverses exist.  If $\Im z$ is large enough then
  $\|VR_\ssx(z)\|<1$ and then, by using the relation
  $z-H=(1-VR_\ssx(z))(z-H_\ssx)$, we see that $z-H$ is invertible and
  \begin{equation}\label{eq:serexp}
    R(z)=R_\ssx(z) (1-VR_\ssx(z))^{-1}=\tsum{n\geq0}{}R_\ssx(z)
    (VR_\ssx(z))^n. 
  \end{equation}
  We have $R_\ssx(z)\in\rc(X)$ hence $VR_\ssx(z)\in\rc$ and since
  $\cp_\sse:\rc\to\rc_\sse$ is a morphism
  \[
    \cp_\sse[VR_\ssx(z)]=\cp_\sse[V]\cp_\sse[R_\ssx(z)]=V_\sse R_\ssx(z)
  \]
  where $V_\sse=\cp_\sse V$. The series in \eqref{eq:serexp} is norm
  convergent hence $R(z)\in\rc$ and
  \begin{align*}
    \cp_\sse[R(z)] &=\tsum{n\geq0}{}
      \cp_\sse[R_\ssx(z)](\cp_\sse[VR_\ssx(z)])^n
    =\tsum{n\geq0}{}R_\ssx(z)(V_\sse R_\ssx(z))^n\\
    &=R_\ssx(z) (1-V_\sse R_\ssx(z))^{-1}=(z-H_\ssx-V_\sse)^{-1}.
  \end{align*}
  Since
  $\|V_\sse R_\ssx(z)\|=\|\cp_\sse[VR_\ssx(z)]\|\leq\|VR_\ssx(z)\|$
  the series is norm convergent. 
\end{proof}

Now we consider unbounded $V$.  Let $H_\ssx$ be a positive
self-adjoint operator on $\ch$ with form domain $\cg$. For each
$E\in\cs$ let $V(E):\cg\to\cg^*$ symmetric such that $V(X)=0$ and \\[1mm]
(1) the family $\{V(E)\}_{E\in\cs}$ is norm summable in
$B(\cg,\cg^*)$;\\[1mm]
(2) $V(E)\geq-\mu_\sse H_\ssx-\nu_\sse$ with $\mu_\sse,\nu_\sse\geq0$,
$\sum_{E\in\cs}\mu_\sse<1$, and $\sum_{E\in\cs}\nu_\sse<\infty$.\\[1mm]
Denote $V=\sum_{E\in\cs}V(E)$ and $V_\sse=\sum_{\ssf\leq\sse}V(F)$.
Then the form sums $H=H_\ssx+V$ and $H_\sse=H_\ssx+V_\sse$ are bounded
from below self-adjoint operators on $\ch$ with form domain equal to
$\cg$.  Then by using Lemma 2.9 and Theorem 3.5 from \cite{DG}, we get

\begin{theorem}\label{th:afgrad}
  Assume $\rc(X)\cdot\rc(E)=\rc(E)\ \forall E\in\cs$. If $H_\ssx$ is
  strictly affiliated to $\rc(X)$ and
\begin{equation}\label{eq:afcond}
  \varphi(H_\ssx)V(E)(H_\ssx+1)^{-1/2} \in \rc(E) \quad\forall
  \varphi\in 
  C_\rmc(\R) \text{ and } \forall E\in\cs  
\end{equation}
then $H$ is strictly affiliated to $\rc$ and $\cp_\sse H=H_\sse$ for
all $E\in\cs$.
\end{theorem}

\end{document}